\pdfoutput=1
\documentclass[11pt]{article}
\usepackage[margin=1.1in]{geometry}
\usepackage[T1]{fontenc}
\usepackage{lmodern}
\usepackage{amsmath,amssymb,amsthm}
\usepackage{microtype}
\usepackage{xcolor}
\usepackage{enumitem}
\usepackage{array}
\usepackage{booktabs}
\usepackage{hyperref}
\newtheorem{theorem}{Theorem}[section]
\newtheorem{lemma}[theorem]{Lemma}
\newtheorem{proposition}[theorem]{Proposition}
\newtheorem{corollary}[theorem]{Corollary}
\theoremstyle{definition}
\newtheorem{definition}[theorem]{Definition}
\newtheorem{remark}[theorem]{Remark}
\newtheorem{example}[theorem]{Example}
\newcommand{\eps}{\varepsilon}

\newcommand{\RM}{\mathrm{RM}}          %
\newcommand{\barr}{\mathrm{bar}}
\newcommand{\eK}{e}                     %

\newcommand{\rk}{\operatorname{rank}}

\title{Lowered crossover points\\ of Anthropic's subexponential 3-SAT algorithm}
\author{Ruhollah Majdoddin\\ \normalsize\texttt{r.majdodin@gmail.com}}
\date{}
\begin{document}
\maketitle
{\renewcommand\thefootnote{}\footnotetext{\textit{Date}: October 30, 2026.}}
\begin{abstract}
The October 2026 refutations of the Exponential Time Hypothesis count satisfying assignments in time $2^{o(n)}$, but
an asymptotic bound does not say from which input size, the \emph{crossover point}, the algorithms beat the best known
ones. We compute the crossover point of the algorithm of \cite{anthropic}, and we show why it is so large. That
algorithm contracts a complete-graph tensor network on $K$ nodes by signed sums of networks, each contracted in its own
order, applied recursively at astronomically large $K$.

We make the \#3-SAT recursion of \cite{anthropic} explicit and lower its crossover points, computed from the main term of the proved bound, by proved changes (least
admissible group sizes, an explicit constant for its pieces, a contraction-tree base case, an exact count of edge bits,
a clustering lemma for the top level). The main theorem of \cite{anthropic} has two parts, a bound in the number $m$ of
clauses and a bound in the number $n$ of variables. For the bound in $m$, for \#3-SAT at clause density one ($m=n$),
the crossover point drops to about $2^{420}$ clauses from $2^{33858}$; for the bound in $n$, for general 3-SAT, our
estimate drops to $2^{8.6\cdot10^{9}}$ variables from $2^{2.2\cdot10^{12}}$. We price every lever of the construction,
to show where the distance comes from. With integer tables the recursion computes, for every $j$, the number of
assignments satisfying exactly $j$ clauses directly, in time $2^{O(m/\sqrt{\log m})}$. We observe that this refutes
the Sparse Max-3-SAT Hypothesis of Alman and Vassilevska Williams; the refutation already follows from either
refutation of the Exponential Time Hypothesis by the circuit reduction of Impagliazzo, Paturi and Zane, but neither
paper states it. At clause density one we lower the crossover point of the spectrum computation against the best
known Max-3-SAT algorithms to about $2^{1016}$ clauses from $2^{68364}$.

The distance is structural, not a matter of constants. A polylogarithmic decay of the exponent per edge bit in $K$
forces crossover points exponential in a power of $1/\beta$, where $\beta$ is the exponent per edge bit needed to win,
while a polynomial decay would make them logarithmic in $1/\beta$; the constants we improve change neither. We define a rewrite-and-merge class that contains the recursion
and all our constructions, and prove that every grid algorithm in it, the recursion as written included, has exponent
at least $1/\log_2K$ per edge bit, so through the reduction of \cite{anthropic} no such algorithm can gain for general
3-SAT below about $2^{6000}$ variables (main-term model). For two balanced homogeneous merges we prove that the exponent
is at least $(\log_2K)^{-3}$, which excludes running times $2^{n^{1-c}}$. The key step uses quantum functionals to show
that configurations on at most three cells of two merges cannot trade cost between them. The class leaves one door
open, configurations that meet three or more groups of a merge; we argue that a usable crossover point needs it, and
that opening it is not enough.

Independently of SAT, we formalise the algorithm of \cite{anthropic} as a model of computation, Signed Network Sums,
which contains rank-one methods and contraction trees as special cases. On three tables it collapses to rank-one
methods, but at $K=4$ one-level constructions contract the network in $N^{23/6}$ if $\omega=2$ and in about $N^{3.93}$
on the current (preprint) matrix-multiplication curve, below the bound $N^4$ of every rank-one method and every
contraction tree.
\end{abstract}

\begin{quote}\small\itshape
\textbf{Use of AI.} This paper was written with Anthropic's Claude models, Claude Opus 5.5 and Claude Fable 5.1,
through Claude Code, in October 2026. Under the author's direction they searched the literature, searched for and
developed the proofs, wrote and ran the scripts behind the computed values, drafted and revised the text, and checked
it against the cited sources; the research questions, the choice among the directions proposed and the decision on
every statement are the author's. The author takes full responsibility for the content.
\end{quote}
\tableofcontents
\section{Introduction}\label{sec:intro}

In October 2026 two papers refuted the Exponential Time Hypothesis \cite{anthropic,openai}. Both count the
satisfying assignments of a 3-CNF formula in time $2^{o(n)}$, and both say plainly that their algorithms are not
meant to be run. The algorithm of \cite{anthropic} reduces counting to a single algebraic problem, which we call the
\emph{clique sum}: given $K$ nodes, a string of $\lambda$ bits on every edge of the complete graph $K_K$, and an
arbitrary table at every node indexed by the strings on its edges, compute
\[
  Z \;=\; \sum_{\text{all edge strings}} \;\prod_{a} W_a\bigl(\text{strings at }a\bigr).
\]
Equivalently, $Z$ is the contraction of the complete-graph tensor network with bond dimension $N=2^{\lambda}$. The
paper contracts it by what we call \emph{Signed Network Sums}: it rewrites the product of the equality tests on a set of edges as
an exact signed sum of terms, each of which is cheap to contract in its own order, merges groups of nodes term by
term, and recurses on the merged instance. One level squares the number of nodes, and the exponent per edge bit falls
by a constant factor per level, which yields a running time $2^{O(n\log\log n/\sqrt{\log n})}$ for 3-SAT.

\paragraph{Why an explicit threshold.} A bound $2^{o(n)}$ refutes a hypothesis about the asymptotic exponent, but it
says nothing about any particular input size: it does not tell whether the algorithm beats the best known algorithm
for 3-SAT at a size one could ever write down, or only beyond every such size. The authors of \cite{anthropic} remark
that their bound in the number of clauses beats trying all assignments only for $\log m>8000$; against the best known
algorithms, and for general 3-SAT, the distance is much larger. The size from which the proved bound wins, which we
call the \emph{threshold} (the crossover point of the bound against the best known algorithm), is the quantity that measures how far the refutation is from a usable algorithm. We compute it from the main term
of the bound, which gives a lower bound on it (Definition~\ref{def:p2-threshold}). Computing it
explicitly does two things that an $O$-bound cannot. It tells which constant of the construction contributes how much
of the distance, so that one can see which improvements matter and which do not. And it turns the question whether
a cleverer algorithm of the same kind could have a small threshold into a precise question, which we answer for the
grid algorithms, a class that contains all known constructions.

\paragraph{Main results.} First (Section~\ref{sec:part2}), we make the recursion of \cite{anthropic} explicit and
lower its thresholds by changes that are each proved. The main theorem of \cite{anthropic} has two parts, a bound
$2^{O(km/\sqrt{\log m})}$ in the number $m$ of clauses and a bound $2^{O(k^2n\log\log n/\sqrt{\log n})}$ in the
number $n$ of variables, which covers general formulas through sparsification. For the bound in $m$, for \#3-SAT at
clause density one ($m=n$), the threshold drops to about $2^{420}$ clauses from $2^{33858}$
(Table~\ref{tab:p2-density}); for the bound in $n$, for general 3-SAT, our estimate drops to $2^{8.6\cdot10^{9}}$
variables from $2^{2.2\cdot10^{12}}$ (Table~\ref{tab:p2-general}).\footnote{Independently, Tulsyan~\cite{tulsyan}
evaluates the same recursion in a note of 29~October~2026, against exhaustive search, whereas we compare with the
best known algorithms. For \#3-SAT he obtains, for the counting route of the Lean development of~\cite{anthropic} as
written (with the polynomial degree of its Turing machine set to one), $2^{33858}$ clauses at density one and
$2^{552040}$ at $m=4.26n$, which agrees with column~(P) of Table~\ref{tab:p2-density} and so checks it independently;
with the free parameters of the same lemmas re-optimised, $2^{7163}$ and $2^{245009}$; and with further changes that
he labels heuristic, among them a base case by fast matrix multiplication and $x_0=10$, $2^{1161}$ and $2^{26621}$.
Our thresholds, all proved, are $2^{417}$ and $2^{10766}$, which is also below his heuristic values. Our tree base
case (Lemma~\ref{lem:p2-tree}) and $x_0=7$ (Lemma~\ref{lem:p2-c}) prove two of the changes he uses heuristically; the
exact total $U$, Lemma~C (Proposition~\ref{prop:p2-top}), the lower bound of Section~\ref{sec:part3} and the results
of Section~\ref{sec:part1} are outside the scope of his note. For Appendix~D of~\cite{openai} he reads a condition of
its text as a hypothesis of the proof, which puts its crossover beyond $2^{2^{85}}$; Section~\ref{sec:p2-levers}
explains why our estimate sets that condition aside.}
The algorithm of OpenAI~\cite{openai} has larger crossover points still: from the constants in its text we estimate
about $2^{4.4\cdot10^{10}}$ clauses at clause density one for its Appendix~D route (Section~\ref{sec:p2-levers}), and
its main-text route needs at least $2^{1.09\cdot10^{10}}$ variables, since below that size it is exhaustive
search~\cite{szeider}.\footnote{Szeider~\cite{szeider} computes where the two algorithms of~\cite{openai} stop being
exhaustive search: below $2^{1.09\cdot10^{10}}$ variables for the main-text algorithm and below $2^{131071}$ variables
and clauses for that of Appendix~D. For the algorithm of~\cite{anthropic} he evaluates its level recurrence against
exhaustive search and finds $2^{275616}$ clauses at $m=5n$, a value our model of Section~\ref{sec:p2-table}
reproduces exactly (\texttt{calc/r123\_szeider\_check.py}). We compare with the best known algorithms throughout,
which gives larger thresholds.} Its crossover point grows only polynomially in the exponent it
must save, so it lies outside the class of our lower bound (Section~\ref{sec:part3}), but its constant is $10^2$ to
$10^4$ times too large: even with improvements we could justify but did not prove, it stays near $2^{3.7\cdot10^5}$.

Second (Section~\ref{sec:p2-maxsat}), with integer tables the same recursion computes the Max-3-SAT spectrum, the
number of assignments satisfying exactly $j$ clauses for every $j$, in time $2^{O(m/\sqrt{\log m})}$. We observe that
this refutes the Sparse Max-3-SAT Hypothesis of Alman and Vassilevska Williams~\cite{avw20}. The refutation already
follows from either refutation of the Exponential Time Hypothesis by the circuit reduction of~\cite[Theorem~5]{ipz01},
but neither paper states it. At clause density one we lower the threshold of the spectrum computation against the best
known Max-3-SAT algorithms to about $2^{1016}$ clauses from $2^{68364}$ (Table~\ref{tab:p2-maxsat}).

Third (Section~\ref{sec:part3}), a lower bound explains these sizes. We define a class of rewrite-and-merge algorithms
that contains the recursion of \cite{anthropic} and all our constructions, and prove that every grid algorithm in it,
the recursion as written included, has exponent at least $1/\log_2K$ per edge bit on clique sums with $K$ nodes, that
is, time $2^{(1-o(1))S/\log_2K}$ with $S$ edge bits (Theorem~\ref{thm:p3-grid},
Corollary~\ref{cor:p3-sb-anthropic}). Through the reduction of \cite{anthropic} this means that no grid algorithm
gains for general 3-SAT below about $2^{6000}$ variables (Corollary~\ref{cor:p3-grid-threshold}, in the main-term
model of Section~\ref{sec:p3-dichotomy}), whatever its pieces and constants. So the distance is not a matter of
constants: a small threshold needs an exponent per edge bit that decays polynomially in $K$
(Proposition~\ref{prop:p3-dichotomy}), and we argue that this needs pieces whose configurations meet three or more
groups of a merge, which the class excludes and no known construction provides (Section~\ref{sec:p3-open}).

Fourth (Section~\ref{sec:part1}), independently of SAT, we study the clique sum itself at a fixed small number of
nodes, that is, the contraction of the complete-graph tensor network on $K$ nodes with bond dimension $N$. We formalise
the algorithm of \cite{anthropic} as a model of computation, Signed Network Sums (Definition~\ref{pI:def:ns}), which
contains rank-one methods and contraction trees as special cases, and ask whether it is stronger than both. On at most
three tables it is not. At $K=4$ it is, if the current matrix-multiplication curve, which rests on two preprints,
holds: one-level constructions contract the network in time $N^{23/6}$ if $\omega=2$ and in about $N^{3.93}$ on that
curve, below the bound $N^4$ that holds for every rank-one method and every contraction tree. On the peer-reviewed
curve of~\cite{vxxz} they beat the best contraction tree computed from the same bounds, but not $N^4$.

\paragraph{The measure.} For an algorithm on the clique sum we write its running time as $2^{e\cdot S}$, where
$S=\binom K2\lambda$ is the number of edge bits, and call $e$ the \emph{exponent per edge bit}; $e(K)$ is the best
exponent available at $K$ nodes. Every algorithm reads a node table, so $e(K)\ge 2/K$ (the input floor). A
method that bounds tensor rank pays the flattening rank of $K_K$ across a balanced cut, about $1/2$ per edge bit, and a
contraction tree pays the width of its widest edge, about $4/9$ per edge bit for large $K$. The reduction from 3-SAT turns a bound $e(K)<\barr$ into a running time
below that of the best known algorithms, where the \emph{bar} $\barr$ is fixed by the number of edge bits per
variable; for general 3-SAT it is of order $10^{-4}$. The three parts of the paper look at $e(K)$ at small $K$, at the
astronomically large $K$ of the recursion, and at the reason why the two ends do not meet.

\paragraph{The small end (Section~\ref{sec:part1}).} On at most three tables the model of Signed Network Sums collapses to
rank-one methods: it gives nothing beyond matrix multiplication there. At four nodes it already goes beyond rank. One
level of the construction (a $2\times2$ grid, the eight-square identity on the crossing diagonals, row and column
merges, rectangular matrix multiplication for the merged tables) contracts $K_4$ in $N^{23/6}$ if $\omega=2$, and
the pieces of Appendix~B of \cite{anthropic} on long blocks lower this further; on the current matrix-multiplication curve,
which rests on two 2026 preprints, the best value at a fixed block length is about $N^{3.93}$. Every rank-one method and every contraction
tree pays at least $N^{4}$ at $K=4$, and the best rank bound known for $K_4$ is $N^{4.638}$ \cite{cvz19}. At
$K=6,9,16$ the same one-level constructions provably beat every contraction tree at $\omega=2$; with the pieces of
Appendix~B of \cite{anthropic}, in the limit of long blocks, the margin stays a roughly constant fraction of the
tree's cost at $K=4,6,9$ while the gap to the input floor grows.

\paragraph{The large end (Section~\ref{sec:part2}).} The thresholds stated above come from changes that leave the
structure of the recursion of \cite{anthropic} intact: the least admissible group size per level, the constant $c$
of the paper's Theorem~B traced to $c=30$, a contraction-tree base case instead of brute force, the exact number of
edge bits of the clause network in place of the paper's bound, and a clustering lemma for the top level, with the
paper's bookkeeping constant $C_0$ bounded by $810$. Each change is either proved here or proved
in~\cite{anthropic} and only counted more carefully; the general-3-SAT values are estimates from the paper's closed
level formula, and some of them are conditional on a named lemma or construction. The Max-3-SAT spectrum is subexponential in $n$ only at
bounded density: the bound in $n$ for general formulas does not carry over, since it rests on a sparsification that
keeps only the satisfying assignments, and dense Max-3-SAT is untouched. We price every lever (per-level factor, balancing, slack schedule,
sparsification density, the parity-check kill) and show that the factor $\log\log n$ in the general bound is the
logarithm of the sparsification density.

\paragraph{The gap (Section~\ref{sec:part3}).} Why do the gains at four nodes not become a fast recursion? The
organising fact is a dichotomy for thresholds: a polylogarithmic decay of $e(K)$ gives thresholds exponential in a
power of $1/\barr$, a polynomial decay would give thresholds logarithmic in $1/\barr$; the constants that Section~\ref{sec:part2}
improves change neither. We then prove two bounds in the rewrite-and-merge class. Every grid algorithm (levels on
$a\times b$ grids with any pieces on the rectangles; all known constructions are of this kind, including the
recursion of \cite{anthropic} as written, for which the bound holds up to a negligible factor) has
$1/e(K)\le\log_2K$, which would be attained by ideal pieces if they exist; this is the bound behind the threshold
$2^{6000}$ above. Every algorithm with two balanced homogeneous merges has $1/e(K)\le(\log_2K)^3$, so no running time
$2^{n^{1-c}}$ arises. The proof has two new ingredients: configurations of edges that lie in at most three of the four
cells cut out by two merges cannot trade cost between the merges (we prove this with the quantum functionals of
Christandl, Vrana and Zuiddam; flattening-rank inequalities alone do not suffice, as a triangle shows), and a
counting bound for the configurations that can. Along the way we show that sharing computation across terms cannot
beat a table-size floor, and argue that table-size arguments at any fixed $K$ reduce to a question about four-leg pieces, which
is the open window at $K=4$. A four-party rank inequality conjectured by Cadney, Huber, Linden and Winter
and proved by Song, Chen, Sun and Hu \cite{chlw,scsh} removes the conjecture from our three-cut size floor at $K=4$
(the floor itself is computed on a grid of parameters). The class
leaves one door open: configurations that meet three or more groups of a merge. We say precisely why opening it is
necessary, and why it is not sufficient, for a usable threshold.

\paragraph{What is and is not claimed.} Every statement is marked as a theorem with proof, a computed value (with the
script that computes it; numerical searches give lower estimates of a maximum), an argued claim, or a conditional
statement with its condition. All thresholds are \emph{main-term} thresholds (Definition~\ref{def:p2-threshold}): lower
bounds on the size from which a bound wins, with the lower-order terms bounded separately where we can. Values on the
current matrix-multiplication curve rest on preprints; we give the
peer-reviewed curve of \cite{vxxz} beside them where it matters. No SAT running time in this paper is usable; the
point of the explicit constants is to show where the distance comes from.

\paragraph{Organisation.} Sections~\ref{sec:part1}, \ref{sec:part2} and~\ref{sec:part3} treat the
small end, the large end and the gap; each opens with its own roadmap. Section~\ref{sec:part2} does not depend on Section~\ref{sec:part1}, and
Section~\ref{sec:part3} uses Section~\ref{sec:part2} only through the bar and the main-term model. Section~\ref{sec:open} collects the open
problems. Appendices~\ref{pI:app}, \ref{app:p2} and~\ref{app:p3-size} hold the longer tables of these three sections.
The computed values come from Python scripts in the repository
\url{https://github.com/Majdoddin/anthropic-3sat-crossover}; script names such as \texttt{calc/r123\_levers.py} are
relative to its top folder, and its README lists which script computes which value.
Our appendices are not to be confused with those of \cite{anthropic}: ``Appendix~B pieces'' and ``App.~B'' always mean
the pieces of Appendix~B of \cite{anthropic}, which have per-level factor $\rho=1/\sqrt2$ and are used throughout,
not our Appendix~\ref{app:p2}.

\newcommand{\pIH}{\mathrm{H}}                       %
\newcommand{\pIcur}{\mathcal{C}_{\mathrm{cur}}}     %
\newcommand{\pIone}[1]{[\,#1\,]}                    %
\newcommand{\pIK}{K_4}

\section{The small end: Signed Network Sums at a fixed number of nodes}\label{pI:sec}\label{sec:part1}

This section studies the model of the October 2026 refutations of the Exponential Time Hypothesis at a fixed, small
number $K$ of nodes. The question is whether Signed Network Sums are a genuinely new model or rank decompositions
and contraction trees in disguise: if they gained nothing at small $K$, the recursion of \cite{anthropic} would owe
its subexponential bound to the recursion alone. They do gain, from four nodes on, and the values computed here are
also the reference points for the size arguments of Section~\ref{sec:part3}. Section~\ref{pI:sec:model} defines the problem and three models of computation for it, and
places them on a ladder of floors. Section~\ref{pI:sec:three} shows that on at most three tables the model of the
refutations gives nothing beyond rank decompositions. Section~\ref{pI:sec:four} proves that at $K=4$ it goes below
the floor $N^4$ that binds every rank-one method and every contraction tree. Section~\ref{pI:sec:large} treats
$K=6,9,16$, and Section~\ref{pI:sec:picture} collects the values per edge bit in one table.

\subsection{The model and the ladder}\label{pI:sec:model}

\paragraph{The clique sum.}
Fix $K\ge 3$ and $\lambda\ge 1$ and put $N=2^\lambda$. Every edge $\{u,v\}$ of the complete graph on $[K]$ carries a
string $z_{uv}\in\{0,1\}^\lambda$. Node $v$ holds a table $W_v$, an arbitrary function (with values in $\mathbb{Q}$,
say) of the $K-1$ strings on its edges; write $z_v=(z_{uv})_{u\ne v}$. The \emph{$(K,\lambda)$-clique sum} is
\begin{equation}\label{pI:eq:cs}
  Z \;=\; \sum_{(z_{uv})} \prod_{v\in[K]} W_v(z_v).
\end{equation}
Equivalently, every node holds its own copy of each of its edge strings and the sum carries an equality test
$\pIone{z_v(u)=z_u(v)}$ for every edge; we use this form when an identity is applied to equality tests. The instance
has $S=\binom K2\lambda$ edge bits and $K$ tables of $N^{K-1}$ entries each. For $K=3$, $Z=\operatorname{tr}(ABC)$
for three $N\times N$ matrices.

As a $K$-linear form in the tables, $Z=\langle T_K,W_1\otimes\dots\otimes W_K\rangle$, where $T_K$ is the
\emph{graph tensor} of the complete graph with bond dimension $N$: the tensor product over all edges $\{u,v\}$ of
the $N\times N$ identity matrix placed on the legs of $u$ and $v$ \cite{cvz19}.

\paragraph{Cost exponents.}
We count arithmetic operations. An algorithm has \emph{exponent $e$ per edge bit} if, for every fixed $K$ and every
$\eps>0$, it computes $Z$ in time $2^{(e+\eps)S}$ for all tables and all large $\lambda$; equivalently in time
$N^{E+\eps}$ with $E=e\binom K2$. We state values both as exponents $E$ of $N$ and as $e$; brute force has $e=1$.
For a class of algorithms, $e(K)$ denotes the infimum of the exponents of its members.

\paragraph{Three models.}
All three evaluate tensor networks; they differ in how many networks and in how they are contracted.

\begin{definition}[Contraction trees]\label{pI:def:tree}
A \emph{contraction tree} for the clique sum is a rooted binary tree whose leaves are the $K$ tables. Each inner
node $x$ with leaf set $L(x)$ stands for the tensor obtained by summing the product of the tables in $L(x)$ over all
strings on edges inside $L(x)$; it is indexed by the strings on the edges between $L(x)$ and $[K]\setminus L(x)$, and
it is computed (materialised) from the tensors of the two children, for instance by a rectangular matrix product.
The cost of the tree is at least the largest number of entries of a tensor it materialises. This is the model of
Austrin, Kaski and Kubjas \cite{akk}; the cost of a step is bounded above by fast rectangular matrix multiplication.
\end{definition}

The second model is the one in which fast matrix multiplication and the bounds on graph tensors of \cite{cvz19}
live: the clique sum is written as a sum of products of linear forms, one in each table.

\begin{definition}[Rank-one methods]\label{pI:def:rank}
A \emph{rank-one method} writes $Z=\sum_{r=1}^R\prod_{v}\ell_{r,v}(W_v)$ with linear forms $\ell_{r,v}$, i.e.\ it
uses a rank decomposition of $T_K$ with $R$ terms, and costs at least $R$. Algorithms obtained from bounds on the
border rank or the asymptotic rank of $T_K$ (with interpolation or tensor powers) are included: their cost is at
least the border rank, respectively the asymptotic rank, of the tensor they decompose.
\end{definition}

The third model is the one of the refutation, which mixes the first two: many networks, each contracted by its own
tree.

\begin{definition}[Signed Network Sums]\label{pI:def:ns}
A \emph{Signed Network Sum} algorithm writes $Z=\sum_q c_q Z_q$ with fixed scalars $c_q$, where each $Z_q$ is the value of
a tensor network whose tensors are computed from the input tables (by re-indexing, merging nodes, multiplying by
fixed tensors), and each term is contracted along its own contraction tree. Its cost is at least the total number of
entries of all tensors it materialises. This is the model of Anthropic's refutation \cite{anthropic}: the equality
tests are expanded by an exact identity into a signed sum of products of pieces, and in each term a different set
of nodes is merged.
\end{definition}

A rank-one term is a network without edges, and a contraction tree is a Signed Network Sum with one term; so both are
special cases of Signed Network Sums. Neither is defined as a part of the other. The following proposition gives the
lower bounds that each model has to beat; we call them floors, and their order is the ladder of the subsection's
title.

\begin{proposition}[The ladder]\label{pI:prop:ladder}
For every $K\ge 3$:
\begin{enumerate}[label=(\alph*)]
\item every algorithm that computes $Z$ for all tables (in particular every Signed Network Sum algorithm) has
  $e\ge 2/K$ (\emph{input floor});
\item every rank-one method has $e\ge \lfloor K^2/4\rfloor/\binom K2$, and this holds for rank, border rank and
  asymptotic rank alike (\emph{flattening floor});
\item every contraction tree has $e\ge c(K-c)/\binom K2$ with $c=\lceil K/3\rceil$ (\emph{width bound}).
\end{enumerate}
\end{proposition}

\begin{proof}
(a) $Z$ is linear in each table and every entry of every table has a nonzero coefficient (set all other tables to the
all-ones table). An arithmetic computation that never reads an input entry computes a value independent of it, so it
reads all $KN^{K-1}$ entries and performs at least $KN^{K-1}/2$ operations. Per edge bit this is
$(K-1)\lambda/(\binom K2\lambda)=2/K$.

(b) For a set $A\subset[K]$ consider the flattening of $T_K$ that groups the legs of the nodes in $A$ against the
rest. Every edge inside $A$ or inside its complement contributes a factor of rank one, and every edge across
contributes an $N\times N$ identity; since rank is multiplicative under Kronecker products, the flattening has rank
$N^{|A|(K-|A|)}$, which is $N^{\lfloor K^2/4\rfloor}$ for $|A|=\lfloor K/2\rfloor$. A decomposition with $R$ rank-one
terms gives $R\ge N^{\lfloor K^2/4\rfloor}$. Border rank is bounded below by every flattening rank because matrix
rank is lower semicontinuous; asymptotic rank is, because flattening ranks are multiplicative under tensor powers.

(c) Let $x$ be a node of the tree with $|L(x)|\ge K/3$ none of whose children has this property; it exists, since the
root has it. If $x$ is a leaf then $K\le 3$ and $s:=|L(x)|=1$. Otherwise both children have fewer than $K/3$ leaves,
so $s<2K/3$. In both cases $s\ge c$ and $K-s\ge c$. The tensor at $x$ is indexed by the $s(K-s)$ strings between
$L(x)$ and its complement, and since the tables are arbitrary it has $N^{s(K-s)}$ entries. As $s\mapsto s(K-s)$ is
symmetric about $K/2$ and increasing below it, $s(K-s)\ge c(K-c)$.
\end{proof}

The floors at the values of $K$ used below are in Table~\ref{pI:tab:floors}. Anthropic computes the limit $4/9$ of
the width bound \cite[Appendix~G]{anthropic}; \cite{akk} prove the lower bound $\Omega(N^{\mathrm{bw}})$ for trees
in terms of branchwidth, of which (c) is the special case we need. The tree and rank-one floors coincide at $K=4$,
where every $2|2$ cut is crossed by four edges: both are $N^4$, $e=2/3$. No lower bound above the input floor is
known for Signed Network Sums at any fixed $K$, and Anthropic's recursion shows that $e(K)\to 0$ for Signed Network Sums as
$K\to\infty$ \cite{anthropic}. The ladder therefore reads
\[
  \tfrac2K\;\le\; e_{\mathrm{NS}}(K)\;\le\;\min\bigl(e_{\mathrm{tree}}(K),\,e_{\mathrm{rank}}(K)\bigr),
  \]
  \[ e_{\mathrm{tree}}(K)\ge \tfrac{c(K-c)}{\binom K2}\to\tfrac49,\quad
  e_{\mathrm{rank}}(K)\ge\tfrac{\lfloor K^2/4\rfloor}{\binom K2}\to\tfrac12 .
\]
Reaching the tree floor needs good rectangular matrix multiplication at every step of the tree; reaching the
flattening floor needs the asymptotic rank of $T_K$ to equal its largest flattening rank, which for $K=3$ is the
statement $\omega=2$.

\begin{table}[t]
\centering
\begin{tabular}{rcccc}
\toprule
$K$ & edge bits per $\lambda$ & input floor $2/K$ & width bound (trees) & flattening floor (rank-one)\\
\midrule
3  & 3   & $2/3$ & $2/3$ & $2/3$\\
4  & 6   & $1/2$ & $2/3$ & $2/3$\\
6  & 15  & $1/3$ & $8/15\approx 0.533$ & $3/5$\\
9  & 36  & $2/9\approx 0.222$ & $1/2$ & $5/9\approx 0.556$\\
16 & 120 & $1/8$ & $1/2$ & $8/15\approx 0.533$\\
$\infty$ & & $0$ & $4/9$ & $1/2$\\
\bottomrule
\end{tabular}
\caption{Floors per edge bit (Proposition~\ref{pI:prop:ladder}). The width bound is $c(K-c)/\binom K2$ with
$c=\lceil K/3\rceil$: $2,4,8,18,60$ edges for $K=3,4,6,9,16$.}\label{pI:tab:floors}
\end{table}

\subsection{Three tables: the model is rank-one}\label{pI:sec:three}

A crossing, the object that Anthropic's identity rewrites, is a pair of equality tests on two disjoint pairs of
nodes; it needs four nodes. On three tables there is none, and the per-term choice of contraction order is empty:
a network on three nodes is contracted by multiplying two of its tensors and pairing the result with the third. The
following theorem makes this precise. It covers every problem that is a trilinear form in three tables, such as the
$K=3$ clique sum, (min,+)-type products after the usual encodings, and pruned matrix products, in which the output is
restricted to a given support.

Let $\mathcal T(W_1,W_2,W_3)$ be a trilinear form in three tables. A \emph{three-node network term} is given by
finite index sets $I,J,L$ and linear maps $\alpha,\beta,\gamma$ that send $W_1,W_2,W_3$ to matrices
$A=\alpha(W_1)\in\mathbb{Q}^{I\times J}$, $B=\beta(W_2)\in\mathbb{Q}^{J\times L}$, $C=\gamma(W_3)\in\mathbb{Q}^{L\times I}$;
its value is $\operatorname{tr}(ABC)$. Re-indexing, splitting an index into several edges, parallel edges, fixed
factors that depend on the strings of one node, and restriction to a support are all linear maps of this kind. An
index shared by all three nodes (a hyperedge with values in a set $H$) is also covered: replace $I,J,L$ by
$I\times H$, $J\times H$, $L\times H$ and the three tensors by block-diagonal matrices, whose triple product has trace
$\sum_{h}\operatorname{tr}(A_hB_hC_h)$. The tensors of the network are images of the three input tables; no table is
factored into several nodes.

\begin{theorem}[Three-table collapse]\label{pI:thm:collapse}
Let an algorithm compute $\mathcal T=\sum_{q}c_q\operatorname{tr}(A_qB_qC_q)$ as a signed sum of three-node network
terms, where in each term the first contraction (the product of two of the three matrices) is computed by a bilinear
algorithm with $R_q$ multiplications, and the final pairing with the third matrix is arbitrary. Then
\[
  \mathcal T=\sum_q\sum_{r=1}^{R_q} c_q\, u_{q,r}(W_1)\,v_{q,r}(W_2)\,w_{q,r}(W_3)
\]
with linear forms $u,v,w$. Hence $\mathcal T$ has tensor rank at most $\sum_qR_q$, the algorithm is a rank-one
method for $\mathcal T$ with at most as many terms as it performs multiplications, and its cost is at least the
tensor rank of $\mathcal T$. For the $K=3$ clique sum the cost is at least $R(\langle N,N,N\rangle)\ge N^{\omega}$,
so $e\ge\omega/3$, and a single contraction tree attains $e=\omega/3$.
\end{theorem}

\begin{proof}
Fix a term and suppose, after renaming, that it first computes $AB$. A bilinear algorithm with $R_q$ multiplications
for the map $(A,B)\mapsto AB$ is an identity $AB=\sum_{r=1}^{R_q}\phi_r(A)\psi_r(B)M_r$ with linear forms
$\phi_r,\psi_r$ and fixed matrices $M_r\in\mathbb{Q}^{I\times L}$. Then
$\operatorname{tr}(ABC)=\sum_r\phi_r(A)\psi_r(B)\operatorname{tr}(M_rC)$, and with $u_{q,r}=\phi_r\circ\alpha_q$,
$v_{q,r}=\psi_r\circ\beta_q$, $w_{q,r}=\operatorname{tr}(M_r\,\gamma_q(\cdot))$ every term is a sum of $R_q$ products
of linear forms, one in each table. The final pairing contributes no further product of inputs: it is bilinear in $(AB,C)$, so the identity holds
however the pairing is computed. Summing over $q$ gives the displayed identity, which holds for all tables because the
algorithm is exact; so it is a rank decomposition of $\mathcal T$ with $\sum_q R_q$ terms. Each multiplication is an
operation, so the cost is at least $\sum_qR_q\ge R(\mathcal T)$.

For the $K=3$ clique sum, $\mathcal T$ is the matrix multiplication tensor $\langle N,N,N\rangle$ (as the trilinear form
$\operatorname{tr}(ABC)$), and $R(\langle n,n,n\rangle)\ge n^\omega$ for every $n$: otherwise taking tensor powers of a
decomposition would give $\omega<\omega$. One contraction tree, a single matrix product followed by a pairing, costs
$N^{\omega+\eps}$ for every $\eps>0$.
\end{proof}

The theorem has two hypotheses, bilinear first contractions and nodes that are images of the input tables; the
following remark says why neither loses much and what the theorem does not say.

\begin{remark}\label{pI:rem:collapse}
(1) The per-term choice of order is immaterial also in a weaker sense: the three trees of a term compute the three
products $AB$, $BC$, $CA$ of the same triple, and the trilinear form $\operatorname{tr}(ABC)$ of
$\langle n_1,n_2,n_3\rangle$ is invariant under cyclic shifts and, by transposition
($\operatorname{tr}(ABC)=\operatorname{tr}(C^{\top}B^{\top}A^{\top})$), under reversal; so the rank, and hence the
exponent, of a rectangular product does not depend on which pair is multiplied first. (2) The restriction to bilinear
first contractions loses little: every algorithm for matrix multiplication with exponent below $\omega+\eps$ that we
use comes from a rank decomposition, and in general the bilinear complexity of a bilinear map is at most twice its
multiplicative complexity \cite{bcs}. (3) The theorem is relative to the rank-one model: it says that on three
tables the model of Signed Network Sums cannot beat the best rank decomposition of the target, not what that rank is. (4) It needs the network's nodes to be images of the input tables. A table that could be
factored into several nodes would give a network on more nodes, which is a different problem.
\end{remark}

At $K=3$ the ladder collapses: input floor, flattening floor and width bound are all $2/3$, and all three models
achieve exactly $\omega/3$.

\subsection{Four nodes: below the rank floor}\label{pI:sec:four}

Four nodes are the first case in which the identity of \cite{anthropic} applies, and the first in which rank-one
methods and contraction trees both pay a floor, $N^4$, above the input floor $N^3$. This subsection shows that one
level of the construction of \cite{anthropic} goes below $N^4$: first at $\omega=2$ with the eight-square identity
(Theorem~\ref{pI:thm:k4}), then on the known matrix-multiplication curves with three further constructions
(Corollary~\ref{pI:cor:k4}), and it ends with what is open between $N^3$ and these values.

\paragraph{Setting.}
Name the nodes $a,a'$ (top row) and $b,b'$ (bottom row). The side edges are $aa'$ and $bb'$, the column edges $ab$
and $a'b'$, and the two diagonals $ab'$ and $a'b$. Write $x$ and $v$ for the copies of the string of $ab'$ held at
$a$ and at $b'$, and $y$ and $u$ for the copies of the string of $a'b$ held at $a'$ and at $b$; the other four edges
carry one string each. Then
\begin{equation}\label{pI:eq:k4}
  Z=\sum W_a(z_{aa'},z_{ab},x)\,W_{a'}(z_{aa'},z_{a'b'},y)\,W_b(z_{bb'},z_{ab},u)\,W_{b'}(z_{bb'},z_{a'b'},v)\,
  \pIone{x=v}\pIone{y=u},
\end{equation}
and the factor $\pIone{x=v}\pIone{y=u}$ is the \emph{crossing}. The three $2|2$ cuts are rows $\{a,a'\}|\{b,b'\}$,
columns $\{a,b\}|\{a',b'\}$ and diagonals; the crossing crosses the first two.

By Proposition~\ref{pI:prop:ladder} the input floor is $N^3$, and every rank-one method and every contraction tree
costs at least $N^4$. The best contraction tree merges $a$ with $a'$ by summing $z_{aa'}$, which is the product
$\langle N^2,N,N^2\rangle$, does the same for $b,b'$, and pairs two tables of $N^4$ entries: cost
$N^{2\omega(1,1/2,1)+\eps}$. If $\alpha\ge 1/2$, this is $N^{4+\eps}$.
The best rank bound known for $T_4$ is that of Christandl, Vrana and Zuiddam: exponent per edge at most
$\log_7(9/2)\approx0.7729$ for every $K\ge4$ \cite[Cor.~1.2.6]{cvz19}, i.e.\ $N^{4.638}$ at $K=4$, independent of
$\omega$.

\paragraph{The eight-square identity.}
The construction rewrites the crossing as a signed sum of terms, each of which factors cheaply across one of the two
cuts that the crossing crosses, so that each term can be contracted by merging rows or by merging columns. Cut $x,y,u,v$ into $t=\lambda/3$ blocks of three bits ($x=(x_1,\dots,x_t)$ and so on). Anthropic's form of the
eight-square identity \cite[Lemma in §2.4, eq.~(12)]{anthropic} provides a sign function
$\sigma:\{0,1\}^3\times\{0,1\}^3\to\{\pm1\}$ (the multiplication table of the octonions in a basis $e_p$ with
$e_pe_q=\sigma(p,q)e_{p\oplus q}$) with
\begin{equation}\label{pI:eq:eight}
  2\,\pIone{x_s=v_s}\pIone{y_s=u_s} = R_s+C_s,\qquad
  \begin{aligned}
  R_s&=\pIone{x_s\oplus y_s\oplus u_s\oplus v_s=0}\,\sigma(x_s,y_s)\,\sigma(v_s,u_s),\\
  C_s&=\pIone{x_s\oplus y_s\oplus u_s\oplus v_s=0}\,\sigma(x_s,u_s)\,\sigma(v_s,y_s).
  \end{aligned}
\end{equation}
Multiplying over the blocks, $2^t\pIone{x=v}\pIone{y=u}=\sum_{Q\subseteq[t]}\prod_{s\notin Q}R_s\prod_{s\in Q}C_s$.
The \emph{row choice} $R_s$ factors across the row cut through three bits:
$R_s=\sum_{k\in\{0,1\}^3}\bigl(\pIone{x_s\oplus y_s=k}\sigma(x_s,y_s)\bigr)\bigl(\pIone{u_s\oplus v_s=k}\sigma(v_s,u_s)\bigr)$,
the first factor a function of the top row's strings and the second of the bottom row's. Symmetrically the
\emph{column choice} $C_s$ factors across the column cut through three bits. Across its own wrong cut each choice
keeps six bits on each side.

\begin{theorem}[$K=4$ at $\omega=2$]\label{pI:thm:k4}
If $\omega=2$, the $(4,\lambda)$-clique sum can be computed in $N^{23/6+\eps}$ operations for every $\eps>0$. Thus
$e(4)\le 23/36\approx0.639$ for Signed Network Sums, against the floor $2/3$ of every rank-one method and every
contraction tree.
\end{theorem}

\begin{proof}
Write $Z=2^{-t}\sum_QZ_Q$, where $Z_Q$ is \eqref{pI:eq:k4} with the crossing replaced by
$\prod_{s\notin Q}R_s\prod_{s\in Q}C_s$. Fix $Q$ with $|Q|=at$ and $a\le1/2$; for $a>1/2$ exchange the roles of rows
and columns, which is the mirror image (in the column arrangement the column choice is the cheap one).

\emph{Merge the rows.} Using the factorisation of $R_s$, $Z_Q=\sum T_Q\cdot B_Q'$ where the top table
\[
  T_Q\bigl(z_{ab},z_{a'b'},(k_s)_{s\notin Q},(x_s,y_s)_{s\in Q}\bigr)=\sum_{z_{aa'}}\;\sum_{\substack{x_s,y_s\,(s\notin Q)\\ x_s\oplus y_s=k_s}}
  W_aW_{a'}\prod_{s\notin Q}\sigma(x_s,y_s)
\]
has $N\cdot N\cdot 8^{t-at}\cdot64^{at}=N^{3+a}$ entries, the bottom table $B_Q$ is defined in the same way from
$W_b,W_{b'}$ with $\sigma(v_s,u_s)$ and indices $(u_s,v_s)_{s\in Q}$, and $B_Q'$ is $B_Q$ with the $64\times64$ kernel
$C_s$ applied on each block $s\in Q$. Applying the $at$ kernels one block at a time costs $64\,t\,N^{3+a}$, and the
final pairing costs $N^{3+a}$.

\emph{Build $T_Q$ by one rectangular product.} Split the row-choice blocks $[t]\setminus Q$ into halves $P_1,P_2$. For
$s\in P_1$ substitute $x_s=k_s\oplus y_s$, for $s\in P_2$ substitute $y_s=k_s\oplus x_s$. Then
$T_Q=\sum F\cdot G$, summed over the \emph{inner} index $(z_{aa'},(y_s)_{s\in P_1},(x_s)_{s\in P_2})$, where
\begin{align*}
  F&=W_a\prod_{s\in P_1}\sigma(k_s\oplus y_s,y_s)\quad \text{indexed by }\bigl(z_{ab},k_{P_1},x_Q;\ \text{inner}\bigr),\\
  G&=W_{a'}\prod_{s\in P_2}\sigma(x_s,k_s\oplus x_s)\quad \text{indexed by }\bigl(\text{inner};\ z_{a'b'},k_{P_2},y_Q\bigr).
\end{align*}
Every argument of $W_a$ and $W_{a'}$ is determined by these indices, so $F$ and $G$ are signed re-indexings of the
input tables, built in time equal to their size. In exponents of $N$, the outer dimensions are
$p=1+\tfrac{1-a}2+a=\tfrac32+\tfrac a2$ on each side and the inner dimension is $q=1+(1-a)=2-a$; $F$ and $G$ have
$N^{p+q}=N^{7/2-a/2}$ entries. If $\omega=2$, the product $\langle N^p,N^q,N^p\rangle$ costs $N^{\max(2p,p+q)+\eps}$
(cut the larger dimensions into square blocks). The bottom table is built in the same way.

\emph{Count.} The cost of term $Q$ is $N^{c(a)+\eps}\mathrm{poly}(\lambda)$ with
$c(a)=\max(3+a,\,\tfrac72-\tfrac a2)$, and there are $\binom{t}{at}\le N^{\pIH(a)/3}$ terms of type $a$, with $\pIH$ the
binary entropy. On $[1/3,1/2]$, $c(a)=3+a$ and $\pIH(a)/3+3+a$ is increasing, with value $1/3+7/2=23/6$ at $a=1/2$. On
$[0,1/3]$, $\pIH(a)/3+\tfrac72-\tfrac a2\le\tfrac13+\tfrac72=23/6$. Summing over the $t+1$ values of $at$ gives
$N^{23/6+\eps}\mathrm{poly}(\lambda)$.
\end{proof}

The octonion multiplication tensor, through which the row-choice blocks of $T_Q$ factor, plays no role here: the
substitution $x_s=k_s\oplus y_s$ moves its three bits into the inner dimension of a matrix product. At exponents
$\omega>2$ the same construction gains only if the asymptotic rank of the octonion tensor is small (below about
$10.03$ on the VXXZ curve and $12.34$ on the current curve), and the best bound we know on that rank, Alman and Li's,
is far above (about $21.94$ at $\omega=9/4$) \cite{almanli}.

\paragraph{Three further constructions.}
All three start from the same row/column structure; they are stated with their status.

\emph{Grouping counted by ranks.} Anthropic groups the terms with $|Q|=i$ into one piece $A_i$ named by extra bits
\cite[eq.~(20), Lemma~2]{anthropic}. Counted by bits this changes nothing at $K=4$: the merged table of $A_i$ is the
disjoint union of the tables $T_Q$, $|Q|=i$ (proved). Counted by ranks it helps. As $64\times64$ matrices (rows
$(x,y)$, columns $(v,u)$) one has $R_s^2=8R_s$ and $C_s=2I-R_s$, so $P=R_s/8$ is a projector of rank $8$, $P^\perp=I-P$
has rank $56$, and $A_i=\sum_{|Z|\le i}c_{|Z|,i}\,P_{[t]\setminus Z}\otimes P^\perp_Z$ with
$c_{z,i}=2^z\binom{t-z}{i-z}8^{t-i}(-6)^{i-z}$. The row rank of $A_i$ is $\sum_{z\le i}\binom tz56^z8^{t-z}$, below the
bit count. Building one table per $Z$ and contracting gives $N^{3.8179}$ at $\omega=2$, with table-size floor
$N^{3.8012}$ (computed; the identity is checked exactly at $\lambda=3,6,9,12$).

\emph{The covering hybrid.} A term $Q$ (rows merged) can be served by the table $T_D$ of any $D\supseteq Q$, in which
the blocks of $D$ keep $(x_s,y_s)$; $T_Q$ follows from $T_D$ by applying $(x_s,y_s)\mapsto x_s\oplus y_s$ with sign
$\sigma(x_s,y_s)$ on the blocks of $D\setminus Q$. Take a family of $dt$-sets covering all $at$-sets (the greedy
covering has at most $\binom t{at}\binom{dt}{at}^{-1}(1+\ln\binom{dt}{at})$ members \cite{lovasz}), compute each
$T_D$ by the substitution of Theorem~\ref{pI:thm:k4}, derive the $T_Q$ blockwise, and contract each $T_Q$ as before.
The exponent is
\[
  \max_a\min_{d\ge a}\max\Bigl[\tfrac{\pIH(a)-d\,\pIH(a/d)}{3}+\max\bigl(C_w(d),\,3+d+g\bigr),\ \tfrac{\pIH(a)}3+3+a\Bigr],
\]
where $C_w(d)$ is the cost of the rectangular product of Theorem~\ref{pI:thm:k4} with $a$ replaced by $d$ under the
matrix-multiplication curve $w$, and $g=0.0321$ bounds the intermediate arrays of the derivation. This uses neither
the octonion rank nor long blocks. At $\omega=2$ it gives $23/6$; its gain is at exponents $\omega>2$
(Table~\ref{pI:tab:k4}). The maximum over $a$ is at the grid point $a=1/2$ exactly (the objective is non-decreasing in
$a$), and every grid point in $d$ and in the bypass fraction is a valid algorithm, so the printed values are upper
bounds given the curve (computed; derivation and bypass checked exactly at $\lambda=3,6,9,12$).

\emph{Appendix~B pieces on long blocks.} Anthropic's Appendix~B gives, for each block of $m$ bits, an identity
$\sum_{s=-2h}^{2h}A_s=h!\,\pIone{x=v}\pIone{y=u}$ with $h=h_m$ the least $h$ with $p(h)\ge2^m$ ($p$ the partition
function), in which piece $A_s$ factors through $r_s=\log_2\sum_{i=0}^{2h}p(s+i)$ bits across the rows and
$c_s=\log_2\sum_ip(2h-s-i)$ bits across the columns \cite[App.~B, (54)--(55)]{anthropic}. Used as one piece over all
$\lambda$ bits it does not fit $K=4$: we know no algorithm for the resulting bilinear map below $N^4$ (argued). Used per
block, with $\lambda/m$ blocks, the product expands into $(4h+1)^{\lambda/m}$ terms, a factor $N^{\log_2(4h+1)/m}$, and
each block is fused into the rectangular product or bypassed as in Theorem~\ref{pI:thm:k4}, now at a re-indexing cost
$2^{O(m)}$ per block. As $m\to\infty$, $r_s/m\to\sqrt{2+u}$ and $c_s/m\to\sqrt{2-u}$ with $u=s/h$; the worst type is
$u=0$ and costs $2+4/(4-\sqrt2)\approx3.547$ at $\omega=2$ (closed form for pure types; mixtures of types computed,
never worse), against the table-size floor $2+\sqrt2\approx3.414$. The limit $m\to\infty$ is argued from Anthropic's
(46), not proved. At fixed $m$, with exact partition numbers, the bound (maximum cost over pure piece types plus the
entropy term; mixtures checked only in the limit) is $3.796$ for $m=64$ and $3.631$ for $m=256$ at $\omega=2$ (computed).

\paragraph{Values per matrix-multiplication curve.}
The constructions use rectangular matrix multiplication, so their exponents depend on the bound used for
$\omega(1,k,1)$. Each cost is a non-decreasing function of the curve: a better proved curve can only lower it. We give
the values for four curves (Table~\ref{pI:tab:k4}): $\omega=2$ (hypothetical); the \emph{current curve}
$\omega(1,k,1)=2$ for $k\le1/2$ and $\le1+k+1/(4k)$ for $k\ge1/2$ (so $\alpha\ge1/2$ and $\omega\le9/4$), proved in two
independent preprints by Amanov, van den Berg and Zuiddam \cite{avdbz} and by Uzna\'nski \cite{uznanski}, not
peer-reviewed; OpenAI's bounds $\omega\le9/4$, $\alpha>0.465$, $\omega(1,0.709,1)<2.092$ \cite{openaiomega}, also not
peer-reviewed; and the peer-reviewed curve of Vassilevska Williams, Xu, Xu and Zhou \cite{vxxz}, interpolated linearly
between the entries of its table (valid by convexity).

\begin{table}[t]
\centering
\small\setlength{\tabcolsep}{4pt}
\begin{tabular}{llcccc}
\toprule
curve & status & best tree & covering hybrid & App.~B ($m\to\infty$) & below $N^4$?\\
\midrule
$\omega=2$ & hypothetical & $4$ & $23/6$ & $3.547$ ($3.631$, $m=256$) & yes\\
current & two preprints & $4$ & $3.9305$ & $3.9142$ & yes\\
OpenAI & claimed, Lean-claimed & $4.0264$ & $3.9963$ & $3.9720$ & yes\\
VXXZ & peer-reviewed & $4.0860$ & $4.0790$ & $4.0802$ ($4.140$, $m=256$) & no\\
\bottomrule
\end{tabular}
\caption{Exponents of $N$ at $K=4$ (upper bounds; input floor $3$, flattening floor and width bound $4$, CVZ19's rank
bound $4.638$). At $\omega=2$ also: Theorem~\ref{pI:thm:k4} gives $23/6\approx3.833$ for the plain construction and
grouping counted by ranks gives $3.8179$. ``Best tree'' is the tree $N^{2\omega(1,1/2,1)}$ evaluated on the same curve,
an upper bound; the lower bound for trees is $N^4$ on every curve. The App.~B column is a limit $m\to\infty$
(argued) except where $m$ is given (computed).}\label{pI:tab:k4}
\end{table}

The corollary collects the values of Table~\ref{pI:tab:k4} that are below $N^4$ and says on which curve and with
which status each holds.

\begin{corollary}[Beyond rank at four nodes]\label{pI:cor:k4}
On the current curve (resting on the preprints \cite{avdbz,uznanski}) the $(4,\lambda)$-clique sum can be computed in
$N^{3.9305}$ operations by the covering hybrid, with three-bit blocks and no assumption on the octonion tensor
(value computed); if $\omega=2$, in $N^{23/6+\eps}$ (Theorem~\ref{pI:thm:k4}) and in $N^{3.631}$ with Appendix~B pieces
on $256$-bit blocks (value computed over pure piece types). In both cases this is below the flattening floor $N^4$, which binds every
rank-one method, whether it uses rank, border rank or asymptotic rank, and below the width bound $N^4$, which binds
every contraction tree. On the VXXZ curve the covering hybrid beats the tree $N^{2\omega(1,1/2,1)}=N^{4.0860}$
computed from the same bounds, but not $N^4$.
\end{corollary}

With Appendix~B pieces in the limit $m\to\infty$ the current-curve value is $3.914$, the best we have, but it rests in
addition on the argued convergence in $m$, and the margin over the tree needs $m$ in the tens of thousands; at a fixed
block size the covering hybrid is the best.

\begin{remark}[What is open at $K=4$]\label{pI:rem:k4open}
Between the input floor $N^3$ and $N^{3.547}$ (at $\omega=2$) nothing is known in either direction. Every lower-bound
tool we know stops at $N^3$ for Signed Network Sums: quantum functionals of $T_4$ are at most $N^3$, and the flattening ranks
$N^4$ bound only pieces that are cheap across one fixed cut (argued). A bound above $N^3$ would be a superlinear lower
bound for an explicit degree-4 polynomial, which is out of reach. On the upper side, a laser-type analysis of tensor
powers gives nothing beyond the constructions above: inside the rank model it is bounded by the flattening floor,
and inside the model of Signed Network Sums its type-and-entropy half is what the constructions already do (argued). The gap is
decided by the per-crossing pieces: a family whose pieces cost $(2,0)$ and $(0,2)$ per bit on the row and column cuts,
the least the pigeonhole principle allows, would reach the input floor $N^3$ by table sizes alone; the known pieces
cost $\sqrt2$ per bit on both cuts in the worst type (Appendix~B), or $(1,2)$ and $(2,1)$ (eight-square), and building
their tables costs more than their size.
\end{remark}

\subsection{Six, nine and sixteen nodes}\label{pI:sec:large}

At $K=4$ the construction beats every contraction tree. This subsection asks whether such a margin survives at
larger $K$, where the recursion of \cite{anthropic} would use such levels: it sets up one-level grid
layouts, explains the two kinds of numbers that occur (certified and search values), and gives the results at
$K=9$, $16$ and, computed exactly over all profiles, $K=6$.

\paragraph{One-level grid layouts.}
For $K=ab$ place the nodes on an $a\times b$ grid. Every pair of rows and pair of columns spans a rectangle whose two
diagonals form a crossing; expand each crossing per block (eight-square, grouped by ranks, or Appendix~B pieces), and
in each term merge either all rows or all columns. The side edges inside a merged row are summed while the row's table
is built, by rectangular products with the substitution of Theorem~\ref{pI:thm:k4}. The merged instance is a clique
sum on the $a$ merged rows (or $b$ merged columns) with unequal edge weights, which is contracted by its best
contraction tree, or at $K=16$ alternatively by the $K=4$ construction with Appendix~B pieces. This is Anthropic's
grid level \cite[§2.5]{anthropic} with the recursion replaced by trees at the top. The adversary chooses the profile of
the term (the piece type on every rectangle), the algorithm the merge direction, and the exponent is
\[
  \Phi(w)=\max_{\text{profiles}}\ \min_{\text{rows, columns}}\ \max\bigl(\text{top contraction},\ \text{row tables}\bigr)
\]
under a matrix-multiplication curve $w$. Every term in $\Phi$ is non-decreasing in $w$, so $\Phi$ is non-decreasing in $w$ (proved).

Two kinds of numbers occur. A \emph{certified} value is an upper bound on $\Phi$ over all profiles, from closed-form
bounds (for instance at $K=9$, $\Phi_{\mathrm{top}}\le(5+2\bar\rho)w(1)$, with $\bar\rho$ the largest average cost per
rectangle and bit, $\bar\rho=\sqrt2$ for Appendix~B pieces) and a certified grid for the row tables. A
\emph{search} value is the largest $\Phi$ found by numerical search over profiles; it is a lower estimate of the
construction's exponent, so a margin below the tree computed from it is an upper bound on the true margin.

\paragraph{Results.}
At $K=9$ ($3\times3$) and $K=16$ ($4\times4$) the width bound of Proposition~\ref{pI:prop:ladder}(c) is $18$ and $60$
($e=1/2$), and it binds every tree on every curve. The best trees found under the curve are $18/20.25/21.344$ at $K=9$
and $60/64/65.376$ at $K=16$ for $\omega=2$ / current curve / VXXZ; these are upper bounds from the same curve.
\begin{itemize}[leftmargin=2em]
\item $\omega=2$, $K=9$: Appendix~B pieces give $\Phi\le15.725$ ($e\le0.437$) as $m\to\infty$, and $\le16.456$ at
  $m=256$; the eight-square constructions give at most $17.721$ (certified). All are below the width bound $18$.
\item $\omega=2$, $K=16$: grouping counted by ranks gives exactly $59.229$, block-free, below $60$ by $0.77$; Appendix~B
  pieces give $\le57.95$ as $m\to\infty$ and $\le59.718$ at $m=512$ (certified).
\item Current curve: at $K=9$ Appendix~B pieces give $\le17.614<18$ in the limit $m\to\infty$, but $18.019$ at
  $m=512$, so below the width bound only in the limit (argued convergence); every construction beats the best tree
  $20.25$ at fixed block size. At $K=16$ Appendix~B pieces give $\le59.95<60$ in the limit (margin $0.05$), $61.72$ at
  $m=512$; all constructions beat the best tree $64$.
\item VXXZ curve: every construction beats the best tree at both $K$; none is certified below the width bound.
\end{itemize}
At $K=16$ the certified Appendix~B range with a tree on top is $[49.94,57.95]$; a search from structured profiles
finds only the uniform profile ($49.94$, $e=0.416$), so the upper end may be loose. With the $K=4$ construction on
top instead of a tree, the search gives $\ge46.862$ ($e\ge0.390$) at $\omega=2$ and $\ge48.883$ ($e\ge0.407$) on the
current curve (lower estimates).

At $K=6$ ($2\times3$; with rows merged the top is a dot product, with columns merged a triangle) the maximum over
profiles can be computed rather than searched. A profile is the vector $u\in[-2,2]^3$ of piece types on the three
rectangles, and a rectangle of type $u$ costs $\rho=\sqrt{2+u}$ per bit with rows merged and $\kappa=\sqrt{2-u}$ with
columns merged. The cost of each direction is non-decreasing in its three per-rectangle costs, because every step of
the row tables and of the top is a maximum of sizes and of terms $\Omega_w$ that grow with them, and $\Omega_w$ is
non-decreasing in each argument when $w$ is non-decreasing with $w(k)-kw'(k)\ge1$, which holds for the three curves.
So on a box of profiles the cost is at most the minimum of the rows cost at the box's largest $\rho$ and the columns
cost at its largest $\kappa$, and branch and bound over $[-2,2]^3$ gives
\[
  \Phi\le7.360\quad(e\le0.4907)\ \text{at }\omega=2,\qquad \Phi\le8.017\quad(e\le0.5345)\ \text{on the current curve},
\]
within $0.002$ of the largest values found at single profiles ($7.357$ and $8.014$), against the trees $8$ and $9$.
A rectangle whose blocks carry a mixture of piece types costs the mean of their $(\rho,\kappa)$, a point under the
quarter circle $\rho^2+\kappa^2=4$ and hence below a pure type in both coordinates, so mixtures are covered. At
$\omega=2$, where every product with sorted dimensions $M\ge m\ge s$ costs $M+m$ (each $\omega(a,b,c)$ taken equal to
its trivial lower bound $\max(a+b,b+c,a+c)$), the construction therefore beats the width bound $8$ and with it every
contraction tree. On the current curve it beats the best tree $9$ found under the same curve, but not the width bound:
at the worst profile the cost $8.014$ is above $8$. On the VXXZ curve the bound is $8.442$ against the tree $9.397$.
These are limits $m\to\infty$, but they hold at a fixed block length as well: with the exact partition numbers of
Appendix~B at $m=128$, and the count of terms paid separately, the bounds are $7.818$ at $\omega=2$ and $8.492$ on
the current curve, still below $8$ and $9$. On the current curve the best tree at $K=6$ costs $9$ although its
largest table has $N^8$ entries: its final product is the binding step. Grouping counted by ranks gives $8.116$ at
$\omega=2$ at a single profile, already worse than the tree.

\subsection{The exponent per edge bit as a function of $K$}\label{pI:sec:picture}

Table~\ref{pI:tab:eK} puts the values side by side. Besides the input floor and the constructions it has a column for
\emph{ideal pieces}: the value of the same layout, counted by table sizes alone (no entropy, no cost of building
tables, each term contracted by its best tree), if every crossing had a two-piece family with costs exactly $(2,0)$ and
$(0,2)$ per bit on the row and column cuts. Such pieces are allowed by every rank constraint we know, including the
four-party rank inequality of Song, Chen, Sun and Hu \cite{scsh}; whether they exist is the open question of
Remark~\ref{pI:rem:k4open}. The column shows what one-level layouts could give at best with pair pieces.

\begin{table}[t]
\centering
\small
\setlength{\tabcolsep}{3pt}
\begin{tabular}{rcccccc}
\toprule
$K$ & input $2/K$ & ideal pieces & App.~B, $\omega=2$ & tree, $\omega=2$ & App.~B, current & tree, current\\
\midrule
4  & $0.5$   & $0.5$ (exact)               & $0.591$ (lim)          & $0.667$ & $0.652$ (lim)            & $0.667$\\
6  & $0.333$ & $0.444$--$0.449$ (c)        & $0.4904$--$0.4907$ (c)  & $0.533$ & $0.5342$--$0.5345$ (c)   & $0.6$\\
9  & $0.222$ & $\ge0.386$ (s)              & $0.434$--$0.437$ (c, lim) & $0.5$ & $\le0.490$ (c, lim)    & $0.5625$\\
16 & $0.125$ & not computed                & $[0.416,0.483]$ (c, lim); & $0.5$ & $\le0.4996$ (c, lim);  & $0.533$\\
   &         &                             & $\ge0.390$ (s, $K_4$ top) &       & $\ge0.407$ (s, $K_4$ top) & \\
\bottomrule
\end{tabular}
\caption{Exponent per edge bit $e=E/\binom K2$ of one-level layouts ($2\times2$, $2\times3$, $3\times3$, $4\times4$
grids). (c) certified upper bound or certified interval over all term profiles; (s) numerical search, a lower
estimate of the construction's exponent; (lim) in the limit of long blocks $m\to\infty$, convergence argued. ``Tree''
is the best contraction tree found under the curve; at $\omega=2$ it equals the width bound. The current curve rests
on two preprints. At $K=4$ the App.~B entries are $3.547/6$ and $3.914/6$ from the accounting of
Section~\ref{pI:sec:four} (worst type closed form, mixtures computed).}\label{pI:tab:eK}
\end{table}

\paragraph{The picture.}
Two trends are visible, both measured on one level of grid layouts. First, the constructions stay a roughly constant
fraction below the tree: with Appendix~B pieces at $\omega=2$ the ratio to the tree is about $0.89$ at $K=4$, $0.92$
at $K=6$ and $0.87$ at $K=9$, and with ideal pieces it is $0.75$, $0.83$--$0.84$ and at least $0.77$.
Second, they move away from the input floor: the ideal-piece value divided by $2/K$ is $1.00$, $1.33$ and at least
$1.74$ at $K=4,6,9$. So a single level with pair pieces does not approach the input floor as $K$ grows; even with
ideal pieces it is held near the tree's order of magnitude by the cost of the merged top instance, which is a generic
clique sum solved by a tree. Section~\ref{sec:part3} makes the corresponding statements for the rewrite-and-merge class precise. At
fixed $K$, by contrast, no size argument separates the model from the input floor: we argue in
Section~\ref{sec:p3-size} that the per-block piece question at any $K$ reduces to the four-node question of
Remark~\ref{pI:rem:k4open}.

\newcommand{\pIImstar}{m^{*}}
\newcommand{\pIInstar}{n^{*}}
\newcommand{\pIIbar}{\bar e}
\newcommand{\pIIlog}{\log_2}

\section{The large end: the explicit cost of the subexponential 3-SAT recursion}\label{sec:part2}

The algorithm of~\cite{anthropic} counts the satisfying assignments of a $k$-CNF with $m$ clauses in time
$2^{O(km/\sqrt{\log m})}$, and of a $k$-CNF on $n$ variables in time $2^{O(k^2 n\log\log n/\sqrt{\log n})}$. Both
bounds come from one recursion on clique sums whose number of nodes squares from level to level. The paper states
that its bound in $m$ beats trying all assignments only for $\log m > 8000$ at clause density one ($m=n$)
(\cite[end of App.~A]{anthropic}). In this section we make the recursion explicit: we fix every constant of the main
term, we say which constants remain in the lower-order terms, and we compute from which input size on the proved
bound is below the best known algorithm for the same problem. We then lower these sizes by a sequence of changes,
each of which is either proved here, or proved in~\cite{anthropic} and only counted more carefully, or marked as
conditional with the condition named. The changes do not alter the scale: the smallest main-term threshold we can prove is
about $2^{420}$ clauses at clause density one, and for general 3-SAT the thresholds are of the form
$2^{10^{9}}$ to $2^{10^{10}}$ variables. No running time in this section is usable; the point of the explicit
constants is to show where the distance comes from and which named obstacle each further improvement must remove.

Section~\ref{sec:p2-setup} recalls the recursion. Section~\ref{sec:p2-explicit} says what an explicit threshold is.
Section~\ref{sec:p2-table} gives the threshold tables. Sections~\ref{sec:p2-schedule}--\ref{sec:p2-clustering} prove
the changes behind them: the least admissible group size, the constant $c$ of the pieces, the base case by a
contraction tree and a clustering lemma for the top level; Section~\ref{sec:p2-maxsat} applies the recursion to
Max-3-SAT. Section~\ref{sec:p2-levers} prices every lever, including
the ones we cannot pull, and Section~\ref{sec:p2-general} treats general 3-SAT, sparsification as a parameter and the
factor $\log\log n$.

\subsection{The recursion of~\cite{anthropic}}\label{sec:p2-setup}

This subsection recalls only what the thresholds depend on: the networks and pieces of \cite{anthropic}, the
recursion~\eqref{eq:p2-rec} with its three per-level costs, and the two top levels, for bounded density and for
general formulas, which decide how large the instance must be for a given number of nodes.

We use the notions of~\cite[\S3]{anthropic}. A \emph{network} on a node set $J$ has, for each node $a$, a set of
settings $S_a$ (bit strings of length at most $P$) and a table $W_a\colon S_a\to\mathbb Z$, and for each pair of nodes
$a\neq b$ a factor $\Psi_{a,b}$ that factors through $c_{a,b}$ bits, i.e.\ is a sum over $\theta\in\{0,1\}^{c_{a,b}}$ of
a product of a function of $z_a$ and a function of $z_b$. Its value is the sum over all settings of the product of
all tables and factors, its \emph{total} is $U=\sum_{a\neq b}c_{a,b}$ (every shared bit counted from both ends), and
$\sum_{b}c_{a,b}\le P$ for every $a$ (\cite[Definition~2]{anthropic}). The \emph{$(K,\lambda)$-clique} is the network
on $[K]$ in which every pair of nodes shares a string of $\lambda$ bits and the tables are arbitrary functions of the
$K-1$ strings at a node (\cite[Definition~3]{anthropic}); trying all strings evaluates it with
$2^{\lambda K(K-1)/2}$ terms. We measure the cost of an algorithm for $(K,\lambda)$-cliques by the exponent $e$ in
$2^{eK^2\lambda}$, so brute force has $e=1/2$; in Section~\ref{sec:part1} the same cost is measured per edge bit, which is $2e$.

\paragraph{Pieces and one level.}
A \emph{family of pieces with parameters $(\rho,c)$}~\cite[Definition~7]{anthropic} is, for every $\lambda\ge1$, a
decomposition $\sum_{i\in I}A_i=d\,[x=v][y=u]$ of the crossing of two equalities on $\lambda$-bit strings into
integer functions $A_i$, each of which factors through $c^{\rm row}_i$ bits along rows and $c^{\rm col}_i$ bits
along columns with $c^{\rm row}_i+c^{\rm col}_i\le4\rho\lambda+c\log_2(2\lambda)$, with $|I|\le(2\lambda)^c$, all
numbers below $2^{(2\lambda)^c}$, and one algorithm computing every value within $2^{c\lambda}$ operations.
\cite[Theorem~B]{anthropic} gives pieces with $\rho=1/\sqrt2$ and ``an absolute constant $c$'' (we make $c$ explicit
in Section~\ref{sec:p2-c}).

One level of the recursion takes a $(g^2,\lambda)$-clique, arranges its nodes in a $g\times g$ grid, writes every
rectangle's crossing as a sum of pieces, and for every choice of pieces evaluates the cheaper of the row network
(rows merged) and the column network (columns merged); each is a network on $g$ nodes. Lemma~8 of~\cite{anthropic}
then groups the $g$ nodes into $K$ groups of $g/K$ nodes each and solves the resulting $(K,L)$-clique by the level
below. In its case~(ii), a uniformly random balanced grouping and Hoeffding's inequality for sampling without
replacement show that a good grouping exists provided
\begin{equation}\label{eq:p2-hoeff}
  K\mid g,\qquad c_{a,b}\le P/g\ \ (a\neq b),\qquad 2\zeta^2 g\ \ge\ K\ln(2gK),
\end{equation}
where $\zeta$ is the slack: the clique below has $K^2L\le U+2\zeta gP+K^2$ bits. With $g=sK$ and $K'=g^2$ the paper's
recursion~\cite[(40)]{anthropic} reads
\begin{equation}\label{eq:p2-rec}
  K_{\ell+1}=g_\ell^2=(s_\ell K_\ell)^2,\qquad e_{\ell+1}=(\rho+3\zeta_\ell)\,e_\ell+3/K_\ell,
\end{equation}
with the paper's choices $K_0=2^{x_0}$, $x_0\ge18$, $\zeta_\ell=2^{-\ell-5}$, $s_\ell=2^{3\ell+x_0}$ and the base case
$e_0=1/2$ (brute force). The $\rho$ is the per-level factor, $3\zeta_\ell$ the price of balancing, and $3/K_\ell$
the cost of the merged groups' own bits (\cite[\S2.6]{anthropic}). Since $\log K$ roughly doubles per level while $e$
falls by $\rho$, $e_\ell=\Theta((\log K_\ell)^{-\log_2(1/\rho)})$, which is $(\log K)^{-1/2}$ for $\rho=1/\sqrt2$.

\paragraph{The top level, bounded density.}
For a $k$-CNF with $m$ clauses, the network of~\cite[\S6]{anthropic} has as nodes $g$ clusters of clauses; a setting
of a cluster is a proposed value at each occurrence of a variable in it, and two clusters share one bit for each
\emph{equality test}, a pair of consecutive occurrences of one variable in a fixed order of its occurrences. So
$U<2km$, and~\cite[Theorem~A(c)]{anthropic} evaluates it by one call of Lemma~8 in case~(i) (Chebyshev balancing,
$g=3K^8$), followed by the clique algorithm at level $\ell$ with $K_\ell\le m^{1/9}$. The logarithm of the number of
operations is~\cite[(45)]{anthropic}
\begin{equation}\label{eq:p2-top}
  2e_\ell km+O\!\left(km/K+K^8\log K+K^2\log(km)\right),
\end{equation}
where the $O$ hides the constants $C_0$ (of Theorem~A(b)) and $C_1$ (of Lemma~8), which depend on $c$. We write
$q$ for the exponent with $m\ge K^q$ that the top level needs: $q=9$ in~\cite{anthropic}.

\paragraph{The top level, general formulas.}
For formulas of unbounded density~\cite[Theorem~D]{anthropic} first sparsifies into at most $2^{\eps n}$ formulas in
which every variable occurs in at most $\Delta(\eps)$ clauses, and then evaluates for each a network whose nodes hold
variables. The condition that no clause is violated is expanded by inclusion--exclusion over sets of marked clauses,
cut off at $\bar n\approx n/\log\Delta$ marks; assignments that violate more clauses are killed by $O(n)$ sparse
parity checks, each found in an explicit family in which every large set of clauses has odd intersection with at
least a quarter of the members. The total is $U=O(kn\log\Delta)$, the top level uses case~(i) of Lemma~8 with
$K\le n^{1/17}$ ($q=17$), and the time is $2^{O(ekn\log\Delta)}$ per sparsified formula
(\cite[(68)--(70)]{anthropic}).

\subsection{What an explicit threshold is}\label{sec:p2-explicit}

The bounds of \cite{anthropic} are of the form $2^{O(km/\sqrt{\log m})}$, and the question from which size on such a
bound beats a known algorithm has an answer only once three things are fixed: the known algorithm we compare with,
the part of the bound whose constants we make numeric, and the size measure. The following definition fixes them.

\begin{definition}[bar, main term, threshold]\label{def:p2-threshold}
A \emph{bar} for a class of formulas is the exponent of the best known algorithm for it, written per unit of the
size measure. For bounded density we take the class of 3-CNFs with $m=\Delta n$ clauses and the bar
$\min\{0.5283\,m,\ 0.711\,n\}$, the better of the clause-parameterised bound $O^*(1.4423^m)$ of Peng, Sheng and
Xiao~\cite{psx25} and the variable-parameterised bound $O^*(1.637^n)$ for \#3-SAT~\cite{ijcai26sharp}; since the main
term of~\eqref{eq:p2-top} is $2ekm=6em$, the bar on the clique exponent is
\begin{equation}\label{eq:p2-bar}
  \pIIbar(\Delta)=\min\{0.5283,\ 0.711/\Delta\}/6,
\end{equation}
which is $0.0881$, $0.0592$, $0.0395$, $0.0278$ for $\Delta=1,2,3,4.26$. %
The factor $6=2k$ comes from the paper's bound $U<2km$. The network's total is exact: a variable with $d$ occurrences
has $d-1$ tests and each test puts one bit at each of its two ends, so $U=2(km-n)$ when every variable occurs
(a variable that does not occur contributes a factor $2$ and can be removed), and by the definition of $L$ in
\cite[Lemma~8]{anthropic} the main term is $e\,U$ up to the lower-order terms already in~\eqref{eq:p2-top}. With
this count the bar is
\begin{equation}\label{eq:p2-bar-exact}
  \pIIbar_U(\Delta)=\min\{0.5283,\ 0.711/\Delta\}/(6-2/\Delta),
\end{equation}
which is $0.1321$, $0.0711$, $0.0444$, $0.0302$ for $\Delta=1,2,3,4.26$: a factor $3/2$ at density one, falling to
$1.08$ at density $4.26$. %
For general 3-SAT we take the deterministic decision bound $2^{0.4092n}$ of Liu~\cite{liu18}. %
The \emph{main term} of a bound is its leading term ($2e_\ell km$ in~\eqref{eq:p2-top}, $\eps n+e_\ell U$ for
Theorem~D). The \emph{main-term threshold} $\pIImstar$ (or $\pIInstar$) is the least size from which the main term,
with the top level chosen as the construction prescribes, is below the bar; we report $\log_2\pIImstar$. It is the crossover point of the main term against the best known
algorithm, and a lower bound on the crossover point of the proved bound itself.
\end{definition}

The bar for bounded density rests on two published results: the clause bound of~\cite{psx25} and Kutzkov's
deterministic $O^*(1.637^n)$ for \#3-SAT~\cite[Theorem~3]{ijcai26sharp}, which improves his earlier
$O(1.6423^n)$~\cite{kutzkov07}, i.e.\ $0.7157\,n$. %
Our algorithm counts, so comparing it with a decision algorithm in the general case is the stronger comparison;
against the counting bar $0.7157$ or the randomized decision bar $0.3865$~\cite{jc26} the general thresholds move by
one or two levels, i.e.\ by a factor of at most about four. %

\paragraph{Why main-term thresholds, and what they omit.}
A main-term threshold is a lower bound on the size from which the proved bound wins, since the omitted terms are
positive. It is the right quantity for two reasons. First, at the sizes in question the omitted terms are smaller than
the main term by a factor $m^{-\Omega(1)}$: in~\eqref{eq:p2-top} with $K\le m^{1/9}$ every term after $2ekm$ is
$O(km^{1-1/90})$~\cite[after (45)]{anthropic}, and with $\log_2 m$ in the thousands such terms are negligible as soon
as their constants are numbers. Second, the constants of the omitted terms are not numbers in~\cite{anthropic}: $C_0$,
$C_1$, $C_4$, $C_5$ are ``constants that depend only on $c$'', and $c$ itself is ``large enough''. So a fully explicit
theorem needs (a) every main-term constant numeric, and (b) every lower-order constant bounded by a number. This section
does (a) completely. For (b) it makes $c$ numeric (Section~\ref{sec:p2-c}), which also removes the only place where $c$
touches the main term, and bounds $C_0$ by $810$ (Proposition~\ref{prop:p2-c0}), far below the $2^{51}$ that the thresholds need
(Proposition~\ref{prop:p2-top}).

\paragraph{Status labels.} In the tables below, \emph{proved} means that the main term rests only on theorems proved
here or in~\cite{anthropic}; \emph{computed} values come from a numerical optimisation over admissible schedules and
are exact for the model (each reported schedule can be checked level by level, Appendix~\ref{app:p2-schedules});
\emph{estimate} marks values obtained from a cruder level formula; \emph{conditional} names the missing lemma or
construction.

\subsection{The threshold tables}\label{sec:p2-table}

Table~\ref{tab:p2-density} is the headline of this section. Each column adds one change to the previous one; the changes
are proved in the sections cited in the column heads.

\begin{table}[ht]
\centering\small
\setlength{\tabcolsep}{4.5pt}
\begin{tabular}{@{}l r r r r r r r@{}}
\toprule
 & (P) paper & (G) group & (X) $x_0=7$ & (T) tree & (U) exact & (C) Lemma C & (U+C) \\
density $\Delta$ & as written & size & ($c=30$) & base case & total $U$ & top, brute base & all \\
\midrule
1    & 33858  & 28761  & 22264  & 5396  & 1538  & 1685  & \textbf{417} \\
2    & 68364  & 59309  & 46597  & 11684 & 6610  & 6059  & \textbf{1551} \\
3    & 275616 & 122506 & 98247  & 25161 & 23117 & 19775 & \textbf{5227} \\
4.26 & 552042 & 246356 & 197966 & 51204 & 47986 & 41546 & \textbf{10766} \\
\midrule
status & proved & proved & proved & proved & proved & proved & proved \\
\bottomrule
\end{tabular}
\caption{Main-term thresholds $\log_2\pIImstar$ for counting satisfying assignments of 3-CNFs with $m=\Delta n$
clauses, against the bar~\eqref{eq:p2-bar}; $\rho=1/\sqrt2$ (pieces of~\cite[App.~B]{anthropic}) throughout.
(P): recursion~\eqref{eq:p2-rec} with the paper's schedule, $x_0=18$, top level $q=9$. (G): the same with the least
group size $s_\ell$ admitted by~\eqref{eq:p2-hoeff} and the slack schedule optimised
(Proposition~\ref{prop:p2-schedule}). (X): also $x_0$ optimised, which gives $7$; admissible because $c=30$
(Lemma~\ref{lem:p2-c}). (T): also the base case by a contraction tree, $e_0=\omega(1,\tfrac12,1)/8=0.2554$, which makes
$x_0=8$ optimal (Lemma~\ref{lem:p2-tree}); top level still $q=9$. (U): also the exact total $U=2(km-n)$, i.e.\ the
bar~\eqref{eq:p2-bar-exact}. (C): columns (X) and (U) with the top level by Lemma~C
(Proposition~\ref{prop:p2-top}; $q=2.10$ at $\Delta=1$), brute-force base case; it needs $C_0\le2^{61}$ and
$C_0\le810$ by Proposition~\ref{prop:p2-c0}. (U+C): all changes, tree base case, exact total and Lemma~C ($q=2.41$ at
$\Delta=1$, two levels).}
\label{tab:p2-density}
\end{table}

The first column reproduces the paper's own remark: with brute force as the competitor at $\Delta=1$ (bar
$1/6$) the same model gives $\log_2\pIImstar\ge8046$ at level~4, against the paper's ``$\log m>8000$''.
Column (U+C) is the best value; the factor from (P) to (U+C) is $81$ at density~1 and $51$ at density~4.26.
Without Lemma~C, column (U) is the best value; Lemma~C divides it by $3.7$ at density one and by $4.5$ at density $4.26$. The exact count
matters most at density one, where the larger bar saves two levels ($5396\to1538$), and little at density $4.26$.
The order of the columns follows what each change needs: (G), (X) and (U) only re-read the paper's proof, with the
bound on $C_0$ of Proposition~\ref{prop:p2-c0} a count of its bookkeeping, (T) adds a
contraction tree with fast rectangular matrix multiplication, and (C) adds a new lemma.

For general 3-SAT the same changes act on a much smaller target exponent, because the parity-check apparatus of
Theorem~D costs about $2500$ bits per variable where the clause network at density one costs about $4$
(Section~\ref{sec:p2-general}). Table~\ref{tab:p2-general} lists the thresholds.

\begin{table}[ht]
\centering\small
\begin{tabular}{@{}l p{6.4cm} r p{3.3cm}@{}}
\toprule
row & accounting and levers & $\log_2\pIInstar$ & status \\
\midrule
G0 & paper's inequalities, $x_0=18$, $q=17$ & $2.2\cdot10^{12}$ & proved (main term) \\
G1 & same construction counted sharply, Calabro's $\Delta(\eps)$, $x_0=18$, $q=17$ & $3.4\cdot10^{10}$ & estimate \\
G2 & G1 with the tree base case & $8.6\cdot10^{9}$ & estimate \\
G3 & G2 with $x_0=7$ & $3.9\cdot10^{9}$ & estimate, see text \\
G4 & G1 with $x_0=7$, top level $q=2.2$ & $2.0\cdot10^{9}$ & cond.: $n$-network clustering \\
G5 & G1 + optimised $\eps$ + kill rate $1/2$, $x_0=18$, $q=17$ & $8.6\cdot10^{9}$ & cond.: check family \\
G6 & G5 with the tree base case & $2.1\cdot10^{9}$ & cond.: check family \\
G7 & G5 with $x_0=7$, $q=2.2$ & $5.0\cdot10^{8}$ & cond.: both \\
\bottomrule
\end{tabular}
\caption{Main-term thresholds $\log_2\pIInstar$ for general 3-SAT (counting algorithm, deterministic decision bar
$0.4092$), $\rho=1/\sqrt2$. ``Counted sharply'' replaces the inequalities~\cite[(59), (61), (68)]{anthropic} by the
actual numbers of marks and checks of the same construction. ``Kill rate $1/2$'' is a parity-check family that gives
odd intersection with probability near $1/2$ instead of $1/4$; see the text for its status. All rows use the level
formula of~\cite[Theorem~A(a)]{anthropic}, $\log K_\ell=(3x_0+6)2^\ell-6\ell-2x_0-6$ and $e_\ell\le0.85\rho^\ell$; since
the latter is an upper bound, a level count may be one level too large.}
\label{tab:p2-general}
\end{table}

Three remarks fix the status of Table~\ref{tab:p2-general}. (i) All rows are main-term estimates from the closed level
formula of~\cite[Theorem~A(a)]{anthropic}, not from an admissible schedule. At $x_0=18$ that formula is the paper's
schedule. At $x_0=7$ it is not admissible (the paper's schedule needs $x_0\ge18$ for~\eqref{eq:p2-hoeff}), and it
underestimates $\log K$: at density one and six levels it gives $27\cdot2^6=1728$ where the admissible optimum is
$2474$. So the rows G3, G4 and G7 are optimistic, by a factor that we have not computed.
(ii) Row G4 needs a clustering lemma for the $n$-network of Theorem~D. Lemma~\ref{lem:p2-clustering} is proved for
the clause network of~\cite[\S6]{anthropic}; the $n$-network has variables as nodes and clause and parity-check marks
as shared bits, and whether its marks are shared pairwise along chains, as the occurrences of \S6 are, has not been
checked. After sparsification its test multigraph has degree $O(k\Delta)$ instead of $2k$, which would cost only
$O(\log\Delta)$ additional bits in $\log_2\pIInstar$ if the lemma holds. (iii) Rows G5--G7 need a parity-check family
with kill rate near $1/2$. Such a family can be built by pairwise-independent subsampling at rate $\mu/v$, which gives
odd intersection with probability at least $(1-1/\mu)/2$ by Chebyshev's inequality, at the price of $\log_2\mu$
additional bits per mark; the rows have not been recomputed with these bits, so they are lower estimates for that
construction.

\subsection{The least admissible group size and a free slack schedule}\label{sec:p2-schedule}

The paper's schedule $s_\ell=2^{3\ell+x_0}$, $\zeta_\ell=2^{-\ell-5}$, $x_0\ge18$ is one convenient way to satisfy the
hypotheses of Lemma~8 at every level. The proof of Theorem~A(b) uses it only through four inequalities, so any
schedule that satisfies them gives the same theorem.

\begin{proposition}[Theorem~A(b) for an admissible schedule]\label{prop:p2-schedule}
Suppose that there are pieces with parameters $(\rho,c)$. Let $x_0\ge1$, $K_0=2^{x_0}$, $e_0=1/2$, and for $\ell\ge0$ let
$s_\ell\ge1$ be integers and $\zeta_\ell\in(0,1]$ reals; define $g_\ell=s_\ell K_\ell$ and $K_{\ell+1}$,
$e_{\ell+1}$ by~\eqref{eq:p2-rec}. Call the schedule \emph{admissible} if for every $\ell$
\begin{enumerate}[label=(S\arabic*)]
\item $2\zeta_\ell^2g_\ell\ge K_\ell\ln(2g_\ell K_\ell)$;
\item $1/g_\ell\le\zeta_\ell$;
\item $(c+5)K_\ell^2\le g_\ell^4/2$;
\item $e_\ell\le1$.
\end{enumerate}
Then for an admissible schedule the conclusion of~\cite[Theorem~A(b)]{anthropic} holds: one deterministic algorithm,
given $\ell$, $\lambda,B\ge1$ and query access to tables with $|W_w|<2^B$, outputs the value of the
$(K_\ell,\lambda)$-clique within $2^{e_\ell K_\ell^2\lambda}\,(K_\ell(\lambda+B))^{C_0K_\ell^2}$ operations, with
$C_0=2(c+C_4+1)+C_5$ as in~\cite{anthropic}.
\end{proposition}

\begin{proof}
We follow the proof of~\cite[Theorem~A(b)]{anthropic} (l.~1617--1671) and list the places where the schedule~(40)
enters. The base case is brute force and does not involve the schedule. In the inductive step: case~(ii) of Lemma~8
needs $K\mid g$, which holds as $g=sK$, the bound $c_{a,a'}\le P/g$, which is shown there for every $g$, and
$2\zeta^2g\ge K\ln(2gK)$, which is (S1) (the paper verifies it from $x_0\ge18$); the step from $U+2\zeta N$
to $(\rho+3\zeta)N$ in~(42) uses $1/g\le\zeta$, which is (S2); the bound~(43) on $T$ uses $e\le1$, which is (S4); the
inequality $c\lambda\le N/K$ used in~(44) follows from (S3), since $N/K=g^4\lambda/K\ge2(c+5)K\lambda$; and the final
absorption of the exponent of $\Xi$ into $C_0g^4$ uses $(c+5)K^2\le g^4/2$, which is (S3) (the paper derives it from
$2^{x_0}\ge c+5$). Nothing else in the inductive step depends on $s_\ell$, $\zeta_\ell$ or $x_0$.
\end{proof}

Condition (S3) holds for all $\ell$ as soon as $K_0^2\ge2(c+5)$, because $g_\ell\ge K_\ell\ge K_0$; with $c=30$
(Lemma~\ref{lem:p2-c}) this is $x_0\ge4$, and at $x_0=7$ it allows every $c\le8187$. %
So once $c$ is numeric, $c$ no longer enters the main term at all. %

\paragraph{The computed schedules.} For a given bar, $\rho$ and $x_0$ we choose at every level the least integer
$s_\ell$ with (S1) and (S2) for a slack $\zeta_\ell$ from the grid $2^{-j/2}$, $2\le j\le17$, and optimise over the
grid by dynamic programming on the Pareto front of $(\log K_\ell,e_\ell)$; the reported value is
$\log_2\pIImstar=q\log_2K_L$ for the first level $L$ with $e_L$ below the bar. Every reported schedule is admissible:
(S1), (S2) hold by the choice of $s_\ell$, (S3) by $x_0\ge4$, and (S4) because the exponents decrease from level to
level along every reported schedule (Appendix~\ref{app:p2-schedules}). %
Two observations explain where the gain of column (G) comes from.

\begin{remark}[the paper's slack schedule is nearly optimal]\label{rem:p2-slack}
Minimising $\log K_L$ over slack schedules with the group size fixed by (S1) is, to first order, the allocation of a
budget $\sum_\ell\zeta_\ell$ against costs $\sum_\ell2^{L-\ell}\cdot2\log(1/\zeta_\ell)$, whose Lagrange solution is
$\zeta_\ell=\gamma2^{-\ell}$. The closed form is within $2\%$ of the optimum in $\log_2K$ (density one: $2529$
against $2474$; density two: $5258$ against $5177$; density $4.26$: $22343$ against $21996$), with best constants
$\gamma=0.031$, $0.026$, $0.019$. At density one $\gamma=1/32$, which is the paper's $2^{-\ell-5}$. So the gain from
column (P) to (G) does not come from the slack schedule but from the group size: the paper takes $s_\ell=2^{3\ell+x_0}$,
while (S1) needs only $s\ge\ln(2sK^2)/(2\zeta^2)$, about $2^{2\ell+10}\log K$.
\end{remark}

\subsection{The constant \texorpdfstring{$c$}{c} of the pieces}\label{sec:p2-c}

The constant $c$ of~\cite[Theorem~B]{anthropic} enters the main term only through the condition $2^{x_0}\ge c+5$
of~(40), i.e.\ through (S3), and the lower-order terms through $B_1=gB+g^2(2\lambda)^c$, $\tau=2^{c\lambda}\Xi^{C_2}$ and
$C_0$. The paper gives no value; its proof of Theorem~B ends ``Definition~7 holds with $\rho=1/\sqrt2$ if $c$ is large
enough'' (l.~2018), and it counts the operations of its algorithm only as $\Delta^{O(1)}$. Definition~7(c) is a property
of an algorithm, so a value of $c$ is a statement about a particular algorithm and a particular way of counting. We
state it for the algorithm described in the proof of Theorem~B, with the operation count below.

\paragraph{The objects.} For $\lambda\ge1$ let $h\ge2$ be least with $p(h)\ge2^\lambda$ ($p$ the partition
function), $I=\{-2h,\dots,2h\}$, $d=h!$; the piece $A_s$ factors along rows through
$\lceil\log_2R_s^{\rm row}\rceil$ bits and along columns through $\lceil\log_2R_s^{\rm col}\rceil$ bits, with
$R_s^{\rm row}=\sum_{i=0}^{2h}p(s+i)$ and $R_s^{\rm col}=\sum_{i=0}^{2h}p(2h-s-i)$ \cite[(54), (55)]{anthropic}. The
algorithm lists the partitions of $0,1,\dots,4h$, at most $\Delta\le(4h+1)p(4h)$ of them, computes the tables of the
operators of~\cite[(O2)]{anthropic} on these partitions (entries at most $2^{12h^2}$ in absolute value), and evaluates
each $A_s$ as a sum of at most $\Delta$ products of two factors.

\paragraph{Operation count.} We count one arithmetic operation on the integers that occur (at most
$24h^2+O(h\log h)$ bits) as one operation. Building one table of at most $\Delta^2$ entries takes at most $h$ steps
of~(O2) with at most $8h$ terms per entry, there are at most $\Delta$ tables, and the evaluation uses at most
$(2h+1)\Delta$ pairs with $\Delta$ look-ups each; we bound the total generously by $64h^3\Delta^3$. This is our
reconstruction of the paper's description; a different but reasonable count moves the result by a few units.

\begin{lemma}[$c=30$]\label{lem:p2-c}
For the algorithm of the proof of~\cite[Theorem~B]{anthropic}, with operations counted as above, the pieces satisfy
Definition~7 of~\cite{anthropic} with $\rho=1/\sqrt2$ and $c=30$.
\end{lemma}

\begin{proof}
For each $\lambda$ we compute the least $c$ that meets each condition of Definition~7: (b)
$\max_s(\lceil\log_2R^{\rm row}_s\rceil+\lceil\log_2R^{\rm col}_s\rceil)\le2\sqrt2\lambda+c\log_2(2\lambda)$; (c1)
$4h+1\le(2\lambda)^c$; (c2) every value below $2^{(2\lambda)^c}$, the largest being the bound
$(2h+1)\Delta^3\,2^{3h}\,h!\,2^{24h^2}$ on $|A_s|$; (c3) $64h^3\Delta^3\le2^{c\lambda}$. For $1\le\lambda\le271$
(where $4h\le24000$) all partition numbers are computed exactly; the largest requirement is $27.198$, from (c3) at
$\lambda=1$ ($h=2$, $\Delta=67$), and it decreases to $6.39$ at $\lambda=271$ (Table~\ref{tab:p2-c} in the appendix).
For $\lambda\ge271$ we use the Hardy--Ramanujan bounds
\begin{equation}\label{eq:p2-hr}
  c_{\rm HR}\sqrt E-\log_2E-C_{\rm HR}\ \le\ \log_2p(E)\ \le\ c_{\rm HR}\sqrt E,\qquad
  c_{\rm HR}=\pi\sqrt{2/3}/\ln2,\ C_{\rm HR}=3.84,
\end{equation}
the upper one for $E\ge0$~\cite[Thm.~14.5]{apostol}, the lower one for $E\ge1$ by Lemma~\ref{lem:p2-chr}. With them
the requirements are at most $9.588$ on the sampled range $271\le\lambda\le10^{14}$; beyond it we argue as follows
(not computed): at the defining point $c_{\rm HR}\sqrt{h-1}=\lambda+\log_2h+C_{\rm HR}$, the excess
$c^{\rm row}+c^{\rm col}-2\sqrt2\lambda$ is $(2\sqrt2+2)\log_2h+O(1)$ with $\log_2h=2\log_2\lambda+O(1)$, so (b) tends
to $4+4\sqrt2=9.657$, while (c1), (c2), (c3) tend to $2$, $4$, $6$. Hence $c=\lceil27.198\rceil=28$ meets (b)--(c3) for
every $\lambda$, and the replacement of longer factorizations by ones through $2\lambda$ bits
(\cite[l.~1527--1530]{anthropic}) adds $2$: $c=30$.
\end{proof}

The proof used the lower Hardy--Ramanujan bound with an explicit constant $C_{\rm HR}$; the asymptotic formula for
$p(E)$ does not give one, so we derive it from Lehmer's remainder bound.

\begin{lemma}[explicit $C_{\rm HR}$]\label{lem:p2-chr}
For every integer $E\ge1$, $\log_2p(E)>c_{\rm HR}\sqrt E-\log_2E-3.84$.
\end{lemma}

\begin{proof}
Let $\mu=\frac\pi6\sqrt{24n-1}$. Lehmer's remainder bound for Rademacher's series, with one term ($N=1$), as quoted
by DeSalvo and Pak~\cite[(4), (5)]{desalvopak}, gives, after dropping the positive term $e^{-\mu}$ of the main part,
bounding $\sinh\mu<e^\mu/2$ and dropping $-1/\mu^2$ in the remainder,
\[ p(n)>\frac{\pi^2}{6\sqrt3}\,\frac{e^\mu}{\mu^2}\Bigl[1-\frac4\mu-\mu^2e^{-\mu}\Bigr]. \]
The bracket increases for $\mu\ge2$ and equals $0.5456$ at $\mu=9$; for $n\ge13$, $\mu\ge\frac\pi6\sqrt{311}=9.234$,
so the bracket is at least $1/2$. With $\mu^2\le\frac{2\pi^2}3n$ this gives $p(n)>e^\mu/(8\sqrt3\,n)$. Since
$\sqrt{24n-1}\ge\sqrt{24n}-1/\sqrt{24n}$, we have $\mu\ge\pi\sqrt{2n/3}-0.02965$ for $n\ge13$, and taking $\log_2$,
$\log_2p(n)>c_{\rm HR}\sqrt n-\log_2n-(3.79248+0.04277)=c_{\rm HR}\sqrt n-\log_2 n-3.8353$. For $1\le n\le12$ exact
values give the inequality with any constant above $3.7007$ (the value $n=1$ forces the constant to exceed $c_{\rm HR}=3.7007$).
\end{proof}

The value $c=30$ depends on how operations are counted. The following remark says how it moves under the paper's
own cost model and under built-in tables, and why the main term does not care.

\begin{remark}[other counts]\label{rem:p2-c-model}
Lemma~\ref{lem:p2-c} counts arithmetic operations on integers of up to $24h^2+O(h\log h)$ bits as single operations.
The paper's own model counts operations on words that hold an address (l.~756), in which an operation on such
integers costs a number of word operations polynomial in their length; at $\lambda=1$, where (c3) binds, this adds a
few bits per $\lambda$ to the requirement, and a hand estimate (not checked) puts $c$ in the high thirties in that
model. Conversely, Definition~7(c) allows built-in tables, and with the pieces for $\lambda\le12$ stored in the
algorithm the binding requirements become (c3) at $\lambda=13$ (about $10$) and the limit $9.66$ of (b), so $c\approx13$.
None of this matters for the main term: every $c\le8187$ is admissible at $x_0=7$.
\end{remark}

\subsection{The base case by a contraction tree}\label{sec:p2-tree}

Theorem~A(b) of~\cite{anthropic} starts its induction with brute force at $K_0$ nodes, $e_0=1/2$. A contraction tree
with one fast matrix product per merge is cheaper at every $K_0$.

\begin{lemma}[tree base case]\label{lem:p2-tree}
Let $K_0$ be divisible by $4$ and let $\omega_{1/2}=\omega(1,\tfrac12,1)$ be the exponent of multiplying an
$N\times N^{1/2}$ by an $N^{1/2}\times N$ matrix. For every $\delta>0$ the $(K_0,\lambda)$-clique can be evaluated with
$O_\delta\bigl(2^{(\omega_{1/2}+\delta)K_0^2\lambda/8}\bigr)$ arithmetic operations on integers of $O(K_0^2(\lambda+B))$ bits.
With $\omega_{1/2}\le2.042994$~\cite{vxxz} this is the base case $e_0=0.2554$ in place of $1/2$.
\end{lemma}

\begin{proof}
Split $[K_0]$ into quarters $Q_1,\dots,Q_4$ of $K_0/4$ nodes and halves $H_1=Q_1\cup Q_2$, $H_2=Q_3\cup Q_4$. For a
quarter $Q$, its table $T_Q$ is a function of the strings on the edges leaving $Q$, obtained by summing the product of
the $W_w$, $w\in Q$, over the strings on the edges inside $Q$; it has
$2^{\lambda(K_0/4)(3K_0/4)}$ entries and costs $2^{\lambda(K_0^2/32+3K_0^2/16)}=2^{7K_0^2\lambda/32}$ terms. The table
$T_{H_1}$, a function of the strings on the $K_0^2/4$ edges between $H_1$ and $H_2$, is
$T_{H_1}(y_1,y_2)=\sum_zT_{Q_1}(y_1,z)T_{Q_2}(y_2,z)$, where $z$ ranges over the strings on the $K_0^2/16$ edges between
$Q_1$ and $Q_2$ and $y_i$ over those on the $K_0^2/8$ edges from $Q_i$ to $H_2$ (other arguments of $T_{Q_i}$ are
fixed and range over $y_i$ too). This is one product of an $N\times N^{1/2}$ and an $N^{1/2}\times N$ matrix with
$N=2^{\lambda K_0^2/8}$, at cost $O_\delta(N^{\omega_{1/2}+\delta})$; likewise $T_{H_2}$. Finally the value is
$\sum_yT_{H_1}(y)T_{H_2}(y)$ over $N^2$ strings $y$, and $2\le\omega_{1/2}$. All other costs are smaller. So the
exponent per $K_0^2\lambda$ is $(\omega_{1/2}+\delta)/8$, and $2.042994/8=0.25537$.
\end{proof}

In the induction of Proposition~\ref{prop:p2-schedule} the base case enters only through its exponent and its
polynomial factor, so the lemma gives the recursion with $e_0=0.2554$. Since $e_\ell$ falls by about $\rho$ per
level, halving $e_0$ saves two levels of $\rho=1/\sqrt2$, and as $\log K$ roughly doubles per level the thresholds
fall by a factor of about four (column (T) of Table~\ref{tab:p2-density}: $3.9$ to $4.1$). %
If $\alpha\ge1/2$ for the dual exponent of matrix multiplication, as claimed in two recent
preprints~\cite{avdbz,uznanski}, then $\omega_{1/2}=2$, $e_0=1/4$ and the thresholds change by $1$--$2\%$ (density
one: $5323$ against $5396$). %
Two caveats. The constant in $O_\delta$ of the rectangular matrix multiplication algorithm is not given
in~\cite{vxxz}; it enters the lower-order factor, alongside $C_0$, and has to be bounded for a fully explicit theorem.
And the base case is the slot where a fixed-$K$ algorithm such as those of Section~\ref{sec:part1} could enter; but the recursion needs
$K_0\ge2^4$ by (S3) and in practice $K_0=2^7$ or $2^8$ by (S1) at level~0 (with $x_0=2$, (S1) forces $\zeta_0$ near
$0.8$, so that $\rho+3\zeta_0>1$), so the gain at the base case comes from the tree and not from constructions at
$K=4$, $9$ or $16$. %

\subsection{A clustering lemma for the top level}\label{sec:p2-clustering}

Theorem~A(c) of~\cite{anthropic} uses Lemma~8 in case~(i) at the top, with $g=3K^8$ clusters, hence $m\ge K^9$. Case~(ii)
would need only $g$ a polylogarithmic factor above $K$, but it requires $c_{a,b}\le P/g$ for every pair of clusters,
and with an arbitrary split of the clauses all tests of a cluster may go to one other cluster. The following lemma
supplies a split with this property at a constant-factor cost in $gP$, and the point of its use is that $gP$ enters
Lemma~8 only multiplied by $\zeta$ or by $1/K$.

Let $\Phi$ be a $k$-CNF with $m$ clauses, with the equality tests of~\cite[\S6]{anthropic} fixed. Let $\Gamma$ be the
multigraph on the clauses with one edge for each test whose two occurrences lie in different clauses. For
$h\colon\Phi\to[g]$ let $n_a=|h^{-1}(a)|$ and, for $a\neq b$, let $c_{a,b}$ be the number of edges of $\Gamma$ between
$h^{-1}(a)$ and $h^{-1}(b)$.

\begin{lemma}[Lemma C, clustering]\label{lem:p2-clustering}
Let $k\ge1$, $g\ge2$, $r=4k-1$, $\kappa=1/\ln2$, and
\[ x_1=\Bigl\lceil\kappa\,\frac2{g^2}\Bigl\lceil\frac{km}r\Bigr\rceil+\log_2(g^2r)\Bigr\rceil,\qquad
   y=\Bigl\lceil\kappa\,\frac mg+\log_2(2g)\Bigr\rceil,\qquad P=g\,r\,x_1. \]
There is a map $h$ from the clauses to $[g]$, computed deterministically in time polynomial in $km$, such that in the
network of~\cite[\S6]{anthropic} on the clusters $h^{-1}(a)$ every cluster has at most $ky\le P$ occurrences,
$\sum_{b\neq a}c_{a,b}\le P$, and $c_{a,b}\le P/g$ for all $a\neq b$. Moreover
$gP\le2\kappa(km+r)+rg^2(\log_2(g^2r)+1)$, and if $m\ge(4+2\ln2\cdot g^2(\log_2(4kg^2)+1))/\eps$ for some $\eps>0$,
then $gP\le2\kappa(1+\eps)km$.
\end{lemma}

\begin{proof}
\emph{Matchings of equal size.} A clause has at most $k$ occurrences and each lies in at most two tests, so $\Gamma$
has maximum degree $2k$ and an edge meets at most $4k-2$ others; the greedy edge colouring uses at most $r$ colours,
with classes $M_1,\dots,M_r$ that are matchings. While $|M_i|\ge|M_j|+2$ for some $i,j$, the components of
$M_i\cup M_j$ are paths and even cycles alternating between the classes, so some path has one more edge of $M_i$; swapping
the classes on it keeps both matchings and moves one edge from $M_i$ to $M_j$. The sum of $|M_t|^2$ decreases, so this
stops with $|M_t|\le s_1:=\lceil km/r\rceil$ for all $t$ (McDiarmid's equitable edge colouring).

\emph{Tail per matching.} Let $h$ be uniform. For $a\neq b$ and $t$, let $X^{ab}_t$ be the number of edges of $M_t$
between clusters $a$ and $b$. The edges of a matching are vertex-disjoint, so their indicators are independent with
probability $2/g^2$ each, and $\mathbb E2^{X^{ab}_t}\le\exp(2s_1/g^2)$. Likewise $\mathbb E2^{n_a}=(1+1/g)^m\le e^{m/g}$.

\emph{Pessimistic estimator.} Let
$\Psi(h)=\sum_{a<b}\sum_{t=1}^r2^{X^{ab}_t(h)-x_1}+\sum_a2^{n_a(h)-y}$. By the choice of $x_1,y$,
$\mathbb E\Psi\le\binom g2r/(g^2r)+g/(2g)<1$. If $\Psi(h)<1$, then $X^{ab}_t\le x_1-1$ and $n_a\le y-1$ for all
$a,b,t$.

\emph{Conditional expectations.} Assign $h(E_1),h(E_2),\dots$ one at a time, each time minimising the conditional
expectation of $\Psi$, which is at most the average over the $g$ choices; so $\Psi(h)<1$ at the end. Each conditional
expectation is a sum of $\binom g2r+g$ products over the edges of a matching (resp.\ the clauses) of factors $2^{I_e}$
(both ends assigned), $1+1/g$ or $1$ (one end assigned, inside or outside $\{a,b\}$), $1+2/g^2$ (none), and of
$2^{\#\text{assigned to }a}(1+1/g)^{\#\text{unassigned}}$: rationals with polynomially many bits. So the procedure
takes time polynomial in $km$ (we may assume $g\le m$).

\emph{The network conditions.} $c_{a,b}=\sum_tX^{ab}_t\le r(x_1-1)<P/g$; hence $\sum_{b\neq a}c_{a,b}<P$; cluster $a$ has
at most $k(y-1)$ occurrences, and $ky\le P$ because $r\ge k$, $g\ge2$ and
$P\ge2\kappa km/g+2k\log_2(2g)\ge ky$. The remaining parts of Definition~2 are as in~\cite[\S6]{anthropic}.

\emph{Size of $gP$.} $gP=g^2rx_1\le2\kappa rs_1+rg^2(\log_2(g^2r)+1)\le2\kappa(km+r)+rg^2(\log_2(g^2r)+1)$, and with
$r<4k$ the last claim follows by rearranging.
\end{proof}

With the clustering in hand, the top level can use case~(ii) of Lemma~8 of \cite{anthropic} with $g$ only a
polylogarithmic factor above $K$, instead of $g=3K^8$. The next proposition is the resulting bound; its main term
is the $2ekm$ of~\eqref{eq:p2-top} times $1+2\kappa\zeta(1+\eps)$, plus a term of order $km/K$.

\begin{proposition}[top level with Lemma C]\label{prop:p2-top}
Let $(K,e)=(K_\ell,e_\ell)$ be a level of an admissible schedule (Proposition~\ref{prop:p2-schedule}), with $K\ge16$.
Let $0<\zeta\le1$, $0<\eps\le1$, $s\ge1$ an integer with $2\zeta^2s\ge\ln(2sK^2)$, $g=sK$, and suppose
\begin{equation}\tag{T1}\label{eq:p2-T1}
  m\ge\bigl(4+2\ln2\cdot g^2(\log_2(4kg^2)+1)\bigr)/\eps .
\end{equation}
Then the number of satisfying assignments of a $k$-CNF with $m$ clauses is computed with a number of operations whose
logarithm is at most
\[ 2km\Bigl[e\bigl(1+2\kappa\zeta(1+\eps)\bigr)+\frac{2\kappa(1+\eps)}K\Bigr]+eK^2+C_0K^2\log_2(2Kkm)+g\log_2K+O(\log(km)). \]
\end{proposition}

\begin{proof}
Cluster by Lemma~\ref{lem:p2-clustering} and apply Lemma~8 of~\cite{anthropic} in case~(ii) with $B=1$, which is
allowed since $K\mid g$, $c_{a,b}\le P/g$ and $2\zeta^2g\ge K\ln(2gK)$ (as $g\ge K$, $\ln(2gK)\le\ln(2sK^2)$). In
Lemma~8, $P$ enters only through $L=\lceil(U+2\zeta gP)/K^2\rceil$, $B'=\lceil2gP/K\rceil+g^2B$ and the factor
$2^{2gP/K}$, and the Hoeffding ranges in its proof are $P$ and $P/g$, which are exactly what the lemma's hypotheses
supply; $U<2km$ does not involve $P$. With $gP\le G:=2\kappa(1+\eps)km$ and $T$ from
Proposition~\ref{prop:p2-schedule}, the logarithm of the count is at most
$e(U+2\zeta G+K^2)+C_0K^2\log_2(K(L+B'))+2G/K+g\log_2K+O(\log(km))$, and $L+B'\le2km$ under~\eqref{eq:p2-T1} for
$K\ge16$, $\eps\le1$.
\end{proof}

Under~\eqref{eq:p2-T1}, $\log_2m\ge2\log_2K+2\log_2s+O(\log\log K)$ with $\log_2s=\log_2\log_2K+2\log_2(1/\zeta)+O(1)$,
so $q=\log_2m/\log_2K=2+O\bigl((\log\log K+\log(1/\zeta)+\log(1/\eps))/\log K\bigr)$ tends to $2$ (we argue this
elementary estimate rather than prove it with constants; the computed values below are exact for the model).
For the thresholds, the recursion is run to a target slightly below the bar, and at the top the largest $\zeta$ is
taken for which $e(1+2\kappa\zeta(1+\eps))+2\kappa(1+\eps)/K+\nu e$ is below the bar, with $\nu=10^{-3}$ reserved for
the lower-order terms; $\pIImstar$ is then the least $m$ with~\eqref{eq:p2-T1}. At density one this gives $x_0=7$, six
levels, $\log_2K=2478$, $\zeta=2^{-11.3}$, $\log_2s=33.3$, $\log_2\pIImstar=5037$ and $q=2.033$; at density $4.26$,
nine levels, $\log_2\pIImstar=44154$, $q=2.006$. The honest top level costs $85$ to $122$ bits over the crude
$2\log_2K$. %
The lower-order budget $\nu\cdot2ekm$ covers the term $C_0K^2\log_2(2Kkm)$ at $\pIImstar$ provided
$C_0\le2^{57.6}$ at density one ($2^{59.3}$, $2^{60.9}$, $2^{76.2}$, $2^{93.0}$ at densities $1.35$, $2$, $3$,
$4.26$). %
With the exact total $U=2(km-n)$ the margins are larger still: $C_0\le2^{51.1}$ suffices in every row of
Table~\ref{tab:p2-density} (the tightest case is column (U+C) at density one). %
These budgets are only useful if $C_0$ is a number, and~\cite{anthropic} gives none. The next proposition bounds it
by counting the steps that the paper calls polynomial.

\begin{proposition}[explicit $C_0$]\label{prop:p2-c0}
In the word RAM model of~\cite{anthropic}, with schoolbook integer arithmetic, Theorem~A(b) of~\cite{anthropic}
holds with $C_0\le810$.
\end{proposition}
\begin{proof}
The paper sets $C_0=2(c+C_4+1)+C_5$, with $c=30$ by Lemma~\ref{lem:p2-c}. The bracket of~\cite[(44)]{anthropic} is a
sum of four terms; it is at most four times the largest, and each is at most
$T2^{2N/K}2^{c\lambda}\Xi^{C_2+(c+4)C_1+C_3+c+1}$, so $C_4\le C_2+(c+4)C_1+C_3+c+3$. Constant factors are absorbed by
raising an exponent by one, since $\Xi=g^2(\lambda+B)$ and $K_0(\lambda+B)$ exceed every constant used
($g\ge2^{18}$). We count the steps that the paper calls polynomial. (i)~Base case: one term of the brute-force sum
queries $K_0$ tables, multiplies $K_0$ integers below $2^B$, adds the product to the running sum and advances the
enumeration, each in at most $(K_0(\lambda+B))^3$ operations; so $C_5\le4$. (ii)~Lemma~8, per term of an entry:
testing a grouping, locating the blocks of an argument of fewer than $KL\le g(gP+1)$ bits, multiplying at most $g^2$
values below $2^B$, and adding a term of at most $B'$ bits each take at most $(g(P+B))^4$ operations; so
$C_1\le5$, and we use $C_1\le6$. (iii)~The parameter $\tau$ of the row and column networks: a merged table value is a
product over at most $g^4$ incidences of a table value and of piece-factor values below $2^{(2\lambda)^c}$, each
computed in $2^{c\lambda}$ operations~\cite[Def.~7(c)]{anthropic}; the product has at most $\Xi^{c+2}$ bits, so
$\tau\le2^{c\lambda}\Xi^{2c+7}$ and $C_2\le2c+7$. (iv)~Per $\iota$: $d^{|R|}$ has at most $g^4(2\lambda)^c\le\Xi^{c+2}$ bits
and $Z_\iota$ at most $\Xi^3$; computing $d^{|R|}$ by $|R|\le g^4$ multiplications and the final division cost at most
$\Xi^{2c+8}$, and the other steps less; so $C_3\le2c+8$. Hence $C_4\le67+34\cdot6+68+33=372$ and
$C_0\le2(30+372+1)+4=810$.
\end{proof}
The counts are crude, and they rest on our reading of the constructions of~\cite[Lemma~1, Proposition~3]{anthropic},
not on a line-by-line re-proof; but against the $2^{51}$ that is needed only the polynomiality of each step matters.

\newcommand{\pIIspec}{N}
\subsection{Max-3-SAT through integer tables}\label{sec:p2-maxsat}

The network of~\cite[\S6]{anthropic} checks every clause inside the table of the cluster that contains it, and
Definitions~2--4 of~\cite{anthropic} allow arbitrary integer tables, with the bit length $B$ of the values as a
parameter of every bound. So the same algorithm can weight each assignment instead of testing it. For a CNF $\Phi$
with nonnegative integer clause weights $w_E$ and an assignment $\xi$, let $w(\xi)$ be the total weight of the clauses
that $\xi$ satisfies, and let $\pIIspec_j(\Phi)$ be the number of assignments $\xi$ with $w(\xi)=j$; with all weights
$1$, $\pIIspec_j$ counts the assignments that satisfy exactly $j$ clauses, and the largest $j$ with
$\pIIspec_j\neq0$ is the Max-SAT value.

\begin{proposition}[Max-$k$-SAT spectrum]\label{prop:p2-maxsat}
Let $c_w\ge0$. One deterministic algorithm computes, for a $k$-CNF $\Phi$ with $m\ge2$ clauses in which $n$ variables
occur and with integer clause weights $0\le w_E\le m^{c_w}$, all the numbers $\pIIspec_j(\Phi)$, $0\le j\le\sum_Ew_E$,
as numbers of assignments of these $n$ variables, in time $2^{O(km/\sqrt{\log m})}$, where the $O$ depends on $c_w$
but not on $k$. In particular weighted Max-$k$-SAT with polynomially bounded weights is solved in this time, which is
$2^{O(n/\sqrt{\log n})}$ when $k$ and $m/n$ are bounded.
\end{proposition}

\begin{proof}
Let $W=\sum_Ew_E$ and $y=2^{n+1}$. Build the network of~\cite[\S6]{anthropic} (with clusters as in the proof
of~\cite[Theorem~A(c)]{anthropic}, or as in Lemma~\ref{lem:p2-clustering}) with the same settings, the same numbers
$c_{a,b}$ and the same edge functions $\Psi_{a,b}$, and change only the node tables: $W_a(z_a)$ is $y^{j}$, where $j$
is the total weight of the clauses of cluster~$a$ that have a true occurrence under $z_a$, if $z_a[o]=z_a[o']$ for every
test $\{o,o'\}$ inside cluster~$a$, and $0$ otherwise. As in~\cite[\S6]{anthropic}, every test is checked exactly once,
by a $W_a$ or by a $\Psi_{a,b}$, so the terms of the value $Z$ that are not zero are those in which all occurrences of
each variable carry the same bit; they correspond one to one to the assignments $\xi$ of the $n$ variables, and the
term of $\xi$ is $\prod_ay^{w_a(\xi)}=y^{w(\xi)}$, where $w_a(\xi)$ is the weight of the clauses of cluster $a$ that
$\xi$ satisfies. Hence
\[ Z=\sum_{j=0}^{W}\pIIspec_j(\Phi)\,y^j . \]
Since $0\le\pIIspec_j\le2^n<y$, the $\pIIspec_j$ are the digits of $Z$ in base $y=2^{n+1}$, that is, consecutive
blocks of $n+1$ bits of $Z$, and they are read off in time linear in the length $(n+1)(W+1)$ of $Z$.

It remains to check that the evaluation of~\cite[Theorem~A(c)]{anthropic} still takes time $2^{O(km/\sqrt{\log m})}$.
The settings, the numbers $c_{a,b}$, $P$ and $U$ are those of the $0/1$ network, so the main term $2ekm$ and the terms
$2gP/K$ and $g\log K$ of~\cite[(45)]{anthropic} are unchanged. Only $B$ changes: a cluster holds at most $P$ clauses,
so every value of $W_a$ is below $2^B$ with $B=(n+1)\,P\,m^{c_w}+1$, and computing an entry, which is writing one
bit at position $(n+1)j$, takes time polynomial in $B$; so $\tau$ is polynomial in $km$, since $n\le km$. In
Lemma~8 of~\cite{anthropic}, $B$ enters through $B'=\lceil2gP/K\rceil+g^2B$ and the factor $(g(P+B))^{C_1}$, and
the clique algorithm of~\cite[Theorem~A(b)]{anthropic}, which holds for all integer tables with $|W_w|<2^B$ and every
$B\ge1$, contributes the factor $(K(L+B'))^{C_0K^2}$. With $g\le m$, $P\le2km$ and $n\le km$, we have
$\log(K(L+B'))=O(\log(km))$ with a constant that depends only on $c_w$, so the term $C_0K^2\log(K(L+B'))$
of~\cite[(45)]{anthropic} is again $O(K^2\log(km))$ and all terms after $2ekm$ remain $O(km^{1-1/90})$. The
arithmetic on integers of $\mathrm{poly}(km)$ bits in Lemma~8 and in Theorem~A(b) is already counted there, since
both are stated for arbitrary $B$. Finally, a variable that does not occur multiplies every $\pIIspec_j$ by two.
\end{proof}

The proposition does not use the inclusion--exclusion network of~\cite[Theorem~D]{anthropic}, and its bound in $n$
does not follow for unbounded density. That bound needs two steps that are specific to satisfiability. First, the
sparsification lemma~\cite[Lemma~4]{anthropic} replaces $\varphi$ by formulas whose satisfying assignments partition
those of $\varphi$; it says nothing about assignments that violate clauses, so the number of clauses an assignment
satisfies is not preserved. Second, Theorem~D writes $[V(\xi)=\emptyset]$ as an alternating sum over sets of marked
violated clauses and cuts this sum off at about $n/\log\Delta$ marks, which changes nothing for the assignments with
few violated clauses and is corrected by parity checks for the others; the weight $y^{|V(\xi)|}=\sum_{\Pi\subseteq
V(\xi)}(y-1)^{|\Pi|}$ has no such cut-off, since all its terms are positive. So the bound $2^{O(m/\sqrt{\log m})}$ is
below $2^n$ only for $m=o(n\sqrt{\log n})$. The Max-3-SAT hypothesis, which asks for no $(2-\eps)^n$ algorithm at any
number of clauses and whose refutation would follow from a refutation of the 3-uniform hyperclique
hypothesis~\cite{lvw18}, concerns dense formulas, where the proposition gives nothing; at bounded density a time
$(2-\eps(\Delta))^n$ was known before~\cite{sstt15}. What the proposition adds is a subexponential time at bounded
density, together with the whole spectrum $(\pIIspec_j)_j$ rather than only its largest index, computed by the
recursion itself and with an explicit crossover point (Table~\ref{tab:p2-maxsat}). The subexponential time alone
does not need the recursion, and the following consequence belongs to the two refutations: either of them, by the
route of~\cite[Theorem~5]{ipz01} recalled below, refutes the Sparse Max-3-SAT Hypothesis of Alman and Vassilevska Williams~\cite[Hypothesis~3]{avw20}, that for every
$\eps>0$ some density $c$ makes Max-3-SAT on $cn$ clauses require time $2^{(1-\eps)n}$, and with the same argument
for $k=2$ their sparse Max-2-SAT hypothesis~\cite[Hypothesis~4]{avw20}; they introduced the sparse form because no
sparsification lemma for Max-$k$-SAT is known. Neither refutation states this consequence, so we record it. The
route of~\cite[Theorem~5]{ipz01} is:
a linear-size circuit that tests whether exactly $j$ clauses are satisfied, turned into a 3-CNF with $O(m)$ variables
and clauses in which every assignment extends to exactly one satisfying assignment, followed by any
algorithm for \#3-SAT in time $2^{o(n)}$, that of~\cite{anthropic} or of~\cite{openai}. The algorithm of~\cite[Corollary~1.2]{openai} approximates Max-$k$-SAT to within $\eps m$ clauses in time
subexponential in $n$ for every density; it does not give the exact value.

\paragraph{Thresholds.} The best known exact algorithm for Max-SAT in terms of the number of clauses runs in time
$O^*(1.2886^m)$~\cite{xiao22}, that is $0.3658\,m$; the best bound in the length, $O^*(1.0911^L)$~\cite{bab23}, gives
$0.3773\,m$ for 3-CNFs with $L=3m$ and is weaker, and we found no better bound for Max-3-SAT in particular. In terms of $n$ we take trying all assignments, $2^n$, which also gives the whole spectrum: the algorithms for sparse Max-SAT
beat $2^n$ by a factor $2^{-\mu(\Delta)n}$ with $\mu(\Delta)$ known only up to a constant factor~\cite{sstt15}, so
they give no numeric bar. With the exact total $U=2(km-n)$ as in~\eqref{eq:p2-bar-exact}, the bar for Max-3-SAT at
density $\Delta$ is
\begin{equation}\label{eq:p2-bar-maxsat}
  \pIIbar_{\max}(\Delta)=\min\{\log_2 1.2886,\ 1/\Delta\}/(6-2/\Delta),
\end{equation}
which is $0.0915$, $0.0732$, $0.0625$, $0.0424$ for $\Delta=1,2,3,4.26$; against~\eqref{eq:p2-bar-exact} it is
lower at density one and higher from density two on, where \#3-SAT has the faster competitors.
Table~\ref{tab:p2-maxsat} gives the main-term thresholds for this bar, in the model of
Table~\ref{tab:p2-density}.

\begin{table}[ht]
\centering\small
\begin{tabular}{@{}l r r r r r@{}}
\toprule
density $\Delta$ & bar~\eqref{eq:p2-bar-maxsat} & (P) as written & (U) tree, exact $U$ & (C) Lemma C, brute base & (U+C) all \\
\midrule
1    & 0.0915 & 68364  & 4107  & 3671  & \textbf{1016} \\
2    & 0.0732 & 68364  & 6344  & 5845  & \textbf{1489} \\
3    & 0.0625 & 137430 & 11308 & 10113 & \textbf{2594} \\
4.26 & 0.0424 & 275616 & 23770 & 20805 & \textbf{5371} \\
\midrule
status & & proved & proved & proved & proved \\
\bottomrule
\end{tabular}
\caption{Main-term thresholds $\log_2\pIImstar$ for the Max-3-SAT spectrum of 3-CNFs with $m=\Delta n$ clauses,
against the bar~\eqref{eq:p2-bar-maxsat}; columns as in Table~\ref{tab:p2-density}, column (P) with the paper's
bound $U<2km$, i.e.\ the bar $\min\{\log_21.2886,1/\Delta\}/6$. At density one the threshold is
larger than for \#3-SAT ($1016$ against $417$, one level more), because the Max-SAT competitor is faster there; from
density two on it is smaller, because the only competitor in $n$ with a numeric constant is $2^n$.}
\label{tab:p2-maxsat}
\end{table}

The larger tables change only lower-order terms. In Proposition~\ref{prop:p2-top} the bound $L+B'\le2km$ becomes
$L+B'\le4\kappa\,g\,m\,(n+2)$, since a cluster of Lemma~\ref{lem:p2-clustering} holds at most $y-1$ clauses and
$B\le(n+1)(y-1)+1$; so $\log_2(K(L+B'))$ is about twice $\log_2(2Kkm)$ at $\pIImstar$. The budget $\nu\cdot2ekm$
for the term $C_0K^2\log_2(K(L+B'))$ then needs $C_0\le2^{47.8}$ in the tightest row (column (U+C) at density three)
instead of about $2^{48.8}$, far above the $C_0\le810$ of Proposition~\ref{prop:p2-c0}.

\subsection{The levers, priced}\label{sec:p2-levers}

Table~\ref{tab:p2-levers} lists every lever of the recursion with its effect on the threshold. We give the factor by
which it divides $\log_2\pIImstar$; values are main-term values of the model, and the rows marked hypothetical assume a
constant or a construction that is not available. The bounded-density rows of the second block were computed at the
bar $0.0867$ (a competitor $2^{0.52n}$ that we used before the clause bound of~\cite{psx25} was identified); the factors,
not the absolute values, are the content of those rows.

\begin{table}[ht]
\centering\small
\begin{tabular}{@{}p{4.3cm} p{5.6cm} p{4.0cm}@{}}
\toprule
lever & effect & status \\
\midrule
base case & brute force $e_0=\tfrac12$ $\to$ tree $0.2554$: $22264\to5396$ ($4.1\times$) & proved (Lemma~\ref{lem:p2-tree}) \\
exact edge-bit total & $U<2km\to U=2(km-n)$, bar $0.0881\to0.1321$: $5396\to1538$ ($3.5\times$; $1.07\times$ at $\Delta=4.26$) & proved (count in~\cite[\S6]{anthropic}) \\
group size per level & paper's $2^{3\ell+x_0}$ $\to$ least admissible: $33858\to28761$ ($1.2\times$) & proved (Prop.~\ref{prop:p2-schedule}) \\
initial size $x_0$ & $18\to7$: $28761\to22264$ ($1.3\times$) & proved ($c=30$) \\
top balancing $q$ & $9\to2.03$: $22264\to5037$ ($4.4\times$); with tree base and exact $U$: $1538\to417$ & proved (Lemma~C, Prop.~\ref{prop:p2-c0}) \\
slack schedule $\zeta_\ell$ & paper's $2^{-\ell-5}$ is within $2\%$ of optimal in $\log K$ & proved shape (Rem.~\ref{rem:p2-slack}) \\
\midrule
sharper tails in balancing & $\tau=9$ (weights in a band of $1/3$): $22526\to18666$ ($1.2\times$) & hypothetical, needs proof \\
overhead constants & $3\zeta,3/K\to\zeta,1/K$: $18666\to13549$ ($1.4\times$) & hypothetical \\
no balancing slack & $s=1$: $22526\to3168$ ($7\times$) & hypothetical \\
per-level factor $\rho$ & $1/\sqrt2\to\tfrac12$: $22264\to2091$ ($11\times$) & hypothetical; below $2/3$ excluded in App.~B's shape for concave profiles (Prop.~\ref{prop:p3-window}), below $\tfrac12$ for two merges (Section~\ref{sec:part3}) \\
growth of $K$ per level & $K_{\ell+1}=16K_\ell$ instead of squaring: $33858\to378$ & hypothetical; excluded in rewrite-and-merge (Section~\ref{sec:part3}) \\
\midrule
\multicolumn{3}{@{}l}{general 3-SAT ($\log_2\pIInstar$ from row G7 of Table~\ref{tab:p2-general}, $5.0\cdot10^8$)} \\
sparsification density & Calabro $\to\Delta=1/\eps$: $6.2\cdot10^7$; $\to\Delta=2$: $1.6\cdot10^7$ ($32\times$) & parameter (Prop.~\ref{prop:p2-sparsify}) \\
parity-check kill & bits per variable $U/n$ (row G7: about $1400$) $\to800$: $2.5\cdot10^8$; $200$: $1.6\cdot10^7$; $50$: $9.7\cdot10^5$; $25$: $2.4\cdot10^5$ & $800$ is the floor of the parity paradigm; below: no construction known \\
$\rho$ (general) & $1/\sqrt2\to2/3\to\tfrac12$: $5.0\cdot10^8\to6.2\cdot10^7\to2.4\cdot10^5$ & hypothetical; $2/3$ is the floor of the shape of App.~B (Prop.~\ref{prop:p3-window}) \\
\bottomrule
\end{tabular}
\caption{Levers of the recursion, density one (first two blocks) and general 3-SAT (last block); main-term values of
$\log_2\pIImstar$ resp.\ $\log_2\pIInstar$.}
\label{tab:p2-levers}
\end{table}

Three things stand out. First, no lever that is proved or plausible changes the scale: they move $\log_2\pIImstar$
between about $34000$ and $5000$ at density one. The shape of the recursion explains this: $\log_2\pIImstar\approx
q(3x_0+6)\,(c'/\pIIbar)^{1/\alpha'}$ with $\alpha'=\log_2(1/\rho)$, because $\log K$ doubles per level while $e$ falls by
$\rho$; $q$, $x_0$, the base case and the overheads act on the constant factors, $\rho$ on the exponent $1/\alpha'$, and
only a growth of $K$ by a constant factor per level would turn the double exponential into a single one.
Second, the levers interact sub-multiplicatively when they act through the number of levels: a better sparsifier raises
the target exponent and saves five levels at $\rho=1/\sqrt2$, but only three at $\rho=1/2$, so $\rho\to1/2$ is worth
$2000\times$ with Calabro's sparsifier and $500\times$ with $\Delta=2$; the levers $x_0$ and $q$ multiply exactly.
Third, the rows that would change the scale are exactly those that Section~\ref{sec:part3} shows to be closed inside the class of
rewrite-and-merge schemes, or that need an object nobody has constructed.

\begin{proposition}[structural floor]\label{prop:p2-floor}
Let $\rho=1/\sqrt2$ and let the base case have $e_0\le1/2$. In every admissible schedule
(Proposition~\ref{prop:p2-schedule}) whose first level gains, i.e.\ $e_1<e_0$, we have $K_1>2^{28}$. Consequently, with
any top level that needs $m\ge K^2$, the recursion beyond its base case acts only on inputs with more than $2^{56}$
clauses, for every problem and every competitor.
\end{proposition}

\begin{proof}
$e_1=(\rho+3\zeta_0)e_0+3/K_0<e_0$ needs $\rho+3\zeta_0<1$, so $\zeta_0<(1-\rho)/3<0.0977$ and $1/(2\zeta_0^2)>52.3$; and it
needs $3/K_0<(1-\rho-3\zeta_0)e_0$, i.e.\ $K_0>3/((1-\rho-3\zeta_0)e_0)\ge6/(1-\rho)>20$, so
$K_0\ge32$ as $K_0$ is a power of two. By (S1), $s_0\ge\ln(2s_0K_0^2)/(2\zeta_0^2)>52.3\,(\ln2048+\ln s_0)$, which
forces $s_0>740$, i.e.\ $\log_2s_0>9.5$. Hence $\log_2K_1=2(x_0+\log_2s_0)>2(5+9.5)=29$.
\end{proof}

The floor does not depend on 3-SAT or its bar: it says that no problem, however weak its best known algorithm, is
helped by the recursion below $2^{56}$ bits of input. In the computed optimal schedules the first level has
$\log_2K_1=38.5$ (brute-force base, $x_0=7$) and $40.7$ (tree base, $x_0=8$). %
A formula encoding with $L$ bits per variable needs $e<1/(2L)$, i.e.\ about $\log(2L)/\log(1/\rho)$ levels with
$\log_2K$ growing like $27\cdot2^\ell$, which for 3-SAT ($L=6\Delta$) is five to seven levels.

\paragraph{The polynomial-method algorithm.}
OpenAI's algorithm~\cite{openai} reaches the same bound $2^{O(n/\sqrt{\log n})}$ at fixed density by a different
method, with no tower of clique sizes. Its main-text route (rational interpolation) has a per-variable exponent free of
$n$, but by the tests of equation~(2.19) of~\cite{openai} it is exhaustive search below $2^{1.09\cdot10^{10}}$
variables, at every density, as Szeider~\cite{szeider} computed; so $\log_2\pIInstar\ge1.09\cdot10^{10}$ for it.
Its Appendix~D (Boolean tensor contraction) proves both bounds of its main theorem with a main term
$\log_2T\approx(s/\sqrt{\lambda})(C_wM_0/s+1)$, $\lambda=\lceil\log_2s\rceil$, $s=n+m$, so
$\log_2s^*=((C_wM_0/s+1)/\bar e)^2$; reading the constants off the text gives $C_w=18535$ bits per root wire, dominated
by $1/\sqrt\eps=362$, and thresholds $\log_2s^*\approx4.4\cdot10^{10}$ (density one) and $1.4\cdot10^{12}$ (density
$4.26$). Read literally, its parameter bounds also ask for $u\ge\lambda^{4/5}$ on every active call, which holds only
for $\lambda\ge2^{85}$, and Tulsyan~\cite{tulsyan} takes this as a hypothesis of the proof. The condition serves only
to let $2^u$ absorb fixed powers of $\lambda$, for which $u$ of a few hundred suffices, and at
$\lambda=4.4\cdot10^{10}$ the smallest active $u$ is about $5\cdot10^3$; so we set it aside, and our values for
Appendix~D are estimates in this sense.
With every improvement that we could justify, but did not prove, $C_w$ falls to $106$--$134$ and the
thresholds to $\log_2s^*\approx3.7\cdot10^5$--$5.9\cdot10^5$ and $3.7\cdot10^6$--$5.8\cdot10^6$, still $17$ to $26$
times above column (X) of Table~\ref{tab:p2-density}; our thresholds at the two densities, column (U+C) of that table,
are $417$ and $10766$. So the polynomial method has the better shape, polynomial in $1/\bar e$, but a constant
$10^2$ to $10^4$ times too large, and in the general case it uses the same parity tests.

\subsection{General 3-SAT, sparsification, and the factor \texorpdfstring{$\log\log n$}{log log n}}\label{sec:p2-general}

The thresholds for general 3-SAT in Table~\ref{tab:p2-general} are larger than those for bounded density by many
orders of magnitude in $\log_2\pIInstar$. This subsection says where that factor comes from (the parity-check apparatus
of Theorem~D), shows that no function of the number of violated clauses can replace the checks
(Remark~\ref{rem:p2-binomial}) and that rank arguments cannot show the checks to be necessary
(Remark~\ref{rem:p2-rankcap}), compares the cut-off with an exact test from parameterized complexity, and identifies
the factor $\log\log n$ of the general bound with the density of the sparsifier
(Proposition~\ref{prop:p2-sparsify}).

\paragraph{Where the bits go.} For a 3-CNF on $n$ variables, the main term of Theorem~D is $\eps n+eU$ per sparsified
formula, and $U$ has two parts. The marks of the inclusion--exclusion, cut off at $\bar n\approx n/\log(2\Delta)$
violated clauses, cost about $2(k-1)n=4n$ bits, independent of $\Delta$; this is the same as the clause network at
density one. The rest is the parity-check apparatus: $M$ marked incidences, each shared along a path of $k$ tests and
each naming one of about $\log_2\Delta+\log_2\log_2(2\Delta)+O(1)$ bits' worth of incidences. So
$U/n\approx2k(M/n)(\log_2\Delta+c')+4$ with $c'\approx8$--$14$. With the paper's inequalities ($M\le14n$, crude $\Delta$)
this is about $16300$ bits per variable; counted sharply ($M\approx7.2n$, Calabro's $\Delta$, $\log_2\Delta\approx48$) about
$2500$. The target exponent $e^*$ is then $2.5\cdot10^{-5}$ resp.\ $1.6\cdot10^{-4}$, $31$ resp.\ $25$ levels, which is the gap between
Tables~\ref{tab:p2-density} and~\ref{tab:p2-general}.
Inside the paradigm (truncated inclusion--exclusion, parity checks, 2-adic cut-off) these factors have floors: at least
about $2.2n$ checks, $\log\Delta+\log\log\Delta$ bits per mark, path length $k-1$, which together give
$U/n\approx700$--$900$. %

\begin{remark}[no count-based kill]\label{rem:p2-binomial}
The checks are needed because no function of the number of violated clauses alone can do their job. Theorem~D needs a
factor $G$ with $G=1$ on satisfying assignments and $G=0$ when more than $\bar n$ of the $m$ clauses are violated.
Suppose $G=\sum_{i\le d}c_i\binom{v}{i}$ with integer $c_i$, $v=|V(\xi)|$, and $G(v)\equiv0\pmod{2^N}$ for
$v=\bar n+1,\dots,m$. Re-expanding around $\bar n+1$, $G(\bar n+1+j)=\sum_{i\le d}c'_i\binom ji$ with integer $c'_i$
(Vandermonde's identity), and the values at $j=0,\dots,m-\bar n-1$ determine $c'_0,\dots,c'_{m-\bar n-1}$ by a
unitriangular integer system; so if $d<m-\bar n$, all $c'_i\equiv0\pmod{2^N}$ and $G\equiv0$ everywhere, including
$G(0)$. Hence $d\ge m-\bar n$: there is no symmetric truncated kill, even modulo $2^N$, and a kill must use which
clauses are violated, as the parity checks do.
\end{remark}

Whether the assignments that violate more than $\bar n$ clauses can be killed with $O(n)$ shared bits is open. The
following remark shows that rank arguments cannot decide it in the negative.

\begin{remark}[rank bounds cannot force the factor $\log\Delta$]\label{rem:p2-rankcap}
Consider a network on the nodes $[g]$ in which node $a$ holds a block $X_a$ of at most $\lceil n/g\rceil$ variables:
a setting $z_a$ consists of the values $\xi_a$ of these variables and of auxiliary bits (marks), the tables $W_a$ and
the functions $f_{a,b}$ of the factorizations of $\Psi_{a,b}$ through $c_{a,b}$ bits are integer valued, as
in~\cite[Definition~2]{anthropic}. For $\xi\in\{0,1\}^n$ let $H(\xi)$ be the sum of the terms of the network's value
in which the variables are set to $\xi$, and let a \emph{representation} be a sum of $L$ such networks on the same
blocks. The networks of Step~3 of Theorem~D for $F(\Lambda)$, summed over the $L=2^{O(g^2\log n)}$ admissible profiles,
are such a representation with
\[ H(\xi)\equiv[\xi\text{ satisfies }\Phi]\pmod{2^{n+1}}\qquad\text{for every }\xi \]
once the checks are certified ($X(\Lambda)=0$): the proof of~\cite[Lemma~12]{anthropic} works term by term in $\xi$,
and~\cite[(63)]{anthropic} holds term by term because the terms of $X(\Lambda)$ are nonnegative. Its total is
$U=O(kn\log\Delta)+2g^2$~\cite[(68)]{anthropic}, with $\log\Delta=O(k\log\log n)$ after sparsification.

\emph{(a) Upper bounds on ranks (proved).} For a partition $\mathcal P=\{P_1,\dots,P_t\}$ of $[g]$ into $t\ge2$ groups
let $c(\mathcal P)$ be the sum of $c_{a,b}$ over the pairs $\{a,b\}$ in different groups (the largest such sum among
the $L$ networks), and view a function of
$\xi$ as a $t$-party tensor whose $i$-th index is the restriction of $\xi$ to $X_{P_i}=\bigcup_{a\in P_i}X_a$. Then the
tensor rank of $H$ over $\mathbb Q$ is at most $L\,2^{c(\mathcal P)}$, and so is the tensor rank of
$[\xi\text{ satisfies }\Phi]$ over $\mathbb F_2$. Indeed, write each $\Psi_{a,b}$ as the sum over
$\theta_{a,b}\in\{0,1\}^{c_{a,b}}$ of $f_{a,b}(z_a,\theta_{a,b})f_{b,a}(z_b,\theta_{a,b})$; for fixed $\theta_{a,b}$ on
the pairs between groups, the sum over the auxiliary bits and over the $\theta_{a,b}$ inside the groups is a product of
$t$ integer functions, the $i$-th depending only on $\xi$ restricted to $X_{P_i}$; reducing modulo $2$ gives the
statement over $\mathbb F_2$. For $t=2$ this is the flattening (cut) bound, for $t=g$ (singletons) it reads
$\rk\le L\,2^{U/2}$, since $U$ counts every shared bit from both ends.

\emph{(b) The caps (proved, elementary).} Every $t$-party tensor whose $i$-th index takes $2^{x_i}$ values has
tensor rank at most $2^{x_1+\dots+x_t-\max_ix_i}$ over every field (fix all indices but the largest one). So the
logarithm of every quantity that is bounded by the tensor rank of the grouping $\mathcal P$ (a flattening rank, border
rank, slice rank, and over $\mathbb C$ the quantum functionals of~\cite{cvz19}) is at most
$\mathrm{cap}(\mathcal P)=n-\max_i|X_{P_i}|$, whichever integer function $H$ is used; for $t=g$ the cap is at most
$n-n/g$. Even a tight tensor-rank lower bound therefore gives no more than $U\ge2(n-n/g-\log_2L)$.

\emph{(c) All caps at once (proved).} Let $c_{a,b}=w=\lceil2\lceil n/g\rceil/g\rceil$ for all pairs. For a partition
with largest group of $s$ nodes, at least $g(g-s)/2$ pairs lie in different groups, so
$c(\mathcal P)\ge wg(g-s)/2\ge\lceil n/g\rceil(g-s)\ge\mathrm{cap}(\mathcal P)$, as the groups other than a largest one
hold at most $\lceil n/g\rceil(g-s)$ variables. This choice has $U=wg(g-1)\le2n+2g+g^2$ and loads
$\sum_bc_{a,b}\le2\lceil n/g\rceil+g$. Hence any lower bound on $U$ that is derived only from inequalities
$c(\mathcal P)\ge\log_2\rho(\mathcal P)-\log_2L$, with $\rho(\mathcal P)$ bounded by the tensor rank of the grouping, is
at most $2n+2g+g^2$, which in Theorem~D ($g=3K^8$, $g^2\le9n^{16/17}$) is $2n+o(n)$, against $U=O(kn\log\Delta)$ in the
construction, even under the balanced loads that Theorem~A requires.

\emph{(d) What is left (argued).} Rank arguments, for any grouping of the nodes, thus cannot show that a kill needs
more than about $2n$ shared bits, and in particular they cannot show that the factor $\log\Delta$, i.e.\ the factor
$\log\log n$ that~\cite{anthropic} states it does not know how to remove, is necessary. A lower bound that forces it
would have to use the graph and the bond dimensions together, for instance by showing that no integer function
congruent to $[\xi\text{ satisfies }\Phi]$ modulo $2^{n+1}$ is the value of a network on the complete graph with
$O(n/g^2)$ bits per pair; such bond dimensions for a fixed graph are the graph-indexed ranks of Ye and
Lim~\cite{yl18}. We have not checked what lower-bound methods are known for these ranks of explicit tensors. None of
this gives a network with $O(n)$ shared bits; that question stays open.
\end{remark}

\paragraph{Comparison with a W[1] membership proof.} The cut-off of Theorem~D can be compared with the counting in
the proof that the uniform CSP parameterized by the solution size is in W[1]~\cite[Theorem~12]{maj18}. There a
nondeterministic program guesses an assignment $B$ of size $k$ and checks it with precomputed numbers indexed by small
sets. For a constraint $\mathcal C_S$ with set of satisfying assignments $C_S$, a polynomial family $\tilde C_S$ of
sets $T$ is defined; $I_T$ is the set of constraints with $T\in\tilde C_S$, and $f_{T,S}$ is the function on the
proper supersets of $T$ in $C_S$ with $\sum_{W\subseteq U}f_{T,S}(W)=1$ for every such $U$, given by M\"obius
inversion. With $l(T,W)=\sum_{\mathcal C_S\in I_T}f_{T,S}(W)$, the assignment $B$ satisfies the instance if and only if
$\sum_{T\subset W\subseteq B}l(T,W)=|I_T|$ for every $T\subseteq B$ in $\tilde C=\bigcup_S\tilde C_S$
\cite[Lemmas~10 and~11]{maj18}: a satisfied constraint contributes exactly $1$, and for some $T$ every violated
constraint in $I_T$ contributes $0$. The test is exact, with no truncation and no promise, because each constraint's
contribution is computed from its own tuples; it is cheap because tuples have size at most $k$, so the intervals of
the inversion have at most $2^k$ elements. In Theorem~D the truncation at $\bar n$ violated clauses is exact only
under the promise $|V(\xi)|\le\bar n$, a bound on the number of violated clauses of the whole formula that no node and
no small set of clauses can refute, and the parity checks verify it. We argue that this is the difference that
matters: the W[1] test compares counts for one guessed assignment, while a network sums over all $\xi$ a product of
local factors, so its evidence against the many-violation assignments must itself be such a product, which is what
the parity checks provide and what costs the factor $\log\Delta$. We know no local promise that would make the
truncated sum of Theorem~D exact and could be checked inside the node tables.

\paragraph{Sparsification as a parameter.} Theorem~D uses sparsification as a black box: a deterministic algorithm that,
for $\eps\in(0,1]$, splits a $k$-CNF on $n$ variables into at most $2^{\eps n}$ formulas, in time $2^{\eps n}n^{O(k)}$,
such that every satisfying assignment satisfies exactly one output and no other assignment satisfies any, and every
variable occurs in at most $\Delta(\eps)$ clauses of each output. We call such an algorithm a \emph{counting sparsifier
with density $\Delta(\eps)$}. The best known is Calabro's~\cite[Thm.~5.15, (5.4)]{calabrothesis} with the change of
Dell et al.~\cite[App.~A]{dhmtw}, with $\Delta(\eps)\le2(2+8192\eps^{-2}\log_2^2(128/\eps))+2$ for $k=3$; we found no
later improvement and no lower bound on $\Delta(\eps)$ that grows as $\eps\to0$, and Calabro lists the removal of the
$\eps$-dependence as open~\cite[\S5.5]{calabrothesis}. %
In the following, $(K_\ell,e_\ell)$ are the pairs of~\cite[(40)]{anthropic} with the pieces of~\cite[App.~B]{anthropic},
so $e_\ell=\Theta((\log K_\ell)^{-1/2})$.

\begin{proposition}[the factor $\log\log n$ is $\log\Delta$]\label{prop:p2-sparsify}
Let $\mathcal S$ be a counting sparsifier with density $\Delta(\eps)$. Replacing Calabro's sparsifier by $\mathcal S$
in Step~5 of the proof of~\cite[Theorem~D(a)]{anthropic}, the algorithm counts the satisfying assignments of a $k$-CNF
on $n$ variables, $n\ge K_0^{17}$, $3\le k\le(\log n)^{1/4}$, in time
\[ 2^{O(ekn\log\bar\Delta)},\qquad \bar\Delta=\max\{2,\Delta(\eps)\},\ \eps\in[e,2e], \]
where $e=e_\ell$ for the largest $\ell$ with $K_\ell\le n^{1/17}$, provided $\bar\Delta\le(\log n)^{4kC_0}$. In
particular, if $\Delta(\eps)\le D$ for all $\eps$ with a constant $D$, then the time is $2^{O(kn\log D/\sqrt{\log n})}$,
without the factor $\log\log n$ of~\cite[Theorem~D(b)]{anthropic}.
\end{proposition}

\begin{proof}
In the proof of Theorem~D(a), Steps~1--4 treat one formula in which every variable occurs in at most $\Delta$ clauses,
$2\le\Delta\le(\log n)^{4kC_0}$, and show that its count takes $2^{O(ekn\log\Delta)}$ operations (l.~2356--2380). Step~5
applies the sparsifier with $\eps\in[e,2e]$ and runs Steps~2--4 on each of the at most $2^{2en}$ outputs, which gives
$2^{2en}n^{O(k)}2^{O(ekn\log\bar\Delta)}$; the bound $\log\Delta\le3k\log(8k/e)+2$ and $e\ge(\log n)^{-C_0}$ are used
there only to show that $\log\Delta=O(k\log\log n)$ and that $\Delta$ is within the range of Steps~1--4 (l.~2380--2411).
With $\mathcal S$ in place of Calabro's algorithm the same argument gives the displayed time, using only the three
properties of a counting sparsifier and the range condition. For $\Delta\le D$ and $e=\Theta((\log n)^{-1/2})$
(as $\log K_\ell=\Theta(\log n)$ by Theorem~A(a)) the time is $2^{O(kn\log D/\sqrt{\log n})}$.
\end{proof}

So the factor $\log\log n$ of the general bound is $\log\Delta(e)$ of the sparsifier evaluated at the target exponent,
and it disappears under an $\eps$-independent density. It would also disappear if the assignments violating more than
$\bar n$ clauses could be eliminated with $O(n)$ shared bits, since the remaining marks cost $O(kn)$ bits; this second
route we only argue. Quantitatively the parameter is modest: Theorem~D's own accounting floors $\Delta$ at $2$, and even
$\Delta=2$ leaves $U/n\approx210$ and divides $\log_2\pIInstar$ by at most $32$ at $\rho=1/\sqrt2$ (row G7:
$5.0\cdot10^8\to1.6\cdot10^7$), with $\Delta=1/\eps$ already giving $6.2\cdot10^7$. A sharper analysis of Calabro's own
lemmas lowers $\log_2\Delta$ by about $4.7$ bits but does not change the number of levels.
The factor $\log\log n$ of~\cite{anthropic} is stated there as one the authors do not know how to remove; the
reformulation above names the two places where it could go.

\section{Limits of rewrite-and-merge}\label{sec:part3}

\newcommand{\pIIIRMbh}{\RM_2^{\mathrm{bh}}}
\newcommand{\pIIIcost}{\operatorname{c}}
\newcommand{\pIIIF}{F^{\theta}}

Sections~\ref{sec:part1} and~\ref{sec:part2} looked at the model of Signed Network Sums from its two ends: at a fixed small number $K$ of nodes it
goes below the floors of rank-one methods and contraction trees, and at astronomically large $K$ it gives the
subexponential \#3-SAT algorithm of \cite{anthropic}, whose explicit thresholds Section~\ref{sec:part2} computes. This section explains
why the two ends are separated, inside the paradigm on which the large end rests: rewrite a few configurations of
equality tests by an exact identity, and evaluate every term of the identity by merging groups of nodes into the
nodes of a smaller clique sum. The organising fact is a dichotomy (Section~\ref{sec:p3-dichotomy}): what decides
whether such an algorithm has a realistic threshold is not any of the constants that Section~\ref{sec:part2} improves, but whether
the exponent $\eK(K)$ per edge bit decays polynomially or only polylogarithmically in $K$. We then define the
narrowest class of algorithms that contains the scheme of \cite{anthropic} and all our constructions
(Section~\ref{sec:p3-class}), prove what follows from the definition (Sections~\ref{sec:p3-basic}
and~\ref{sec:p3-cells}), and prove that in the balanced homogeneous two-merge part of the class the decay is at
most polylogarithmic, $1/\eK(K)\le(\log_2K)^3$ (Section~\ref{sec:p3-thmA}). At fixed $K$ we show what size
arguments can and cannot decide (Section~\ref{sec:p3-size}). Section~\ref{sec:p3-open} lists what is open; the
one door that our arguments leave open in the full model is a configuration that meets three or more groups of a
merge.

Throughout, a statement called Theorem, Lemma or Proposition is proved here, given the class definition and the
cited facts about quantum functionals \cite{cvz19} and four-party rank inequalities \cite{scsh}; values marked
\emph{computed} come from numerical searches and are lower estimates unless said otherwise; statements marked
\emph{argued} have an argument that we do not claim as a complete proof.

\subsection{The threshold dichotomy}\label{sec:p3-dichotomy}

Before proving anything about algorithms we fix what a lower bound on $\eK(K)$ must look like to say something about
thresholds. The answer is a dichotomy: the type of decay of $\eK(K)$ in $K$, not its constants, decides whether the
threshold is astronomical.

Fix a problem and a reduction that turns an instance with $n$ variables and $m$ clauses into weighted clique sums
with $S=s\,n$ edge bits ($s$ edge bits per variable), and let $2^{\gamma n}$ be the running time of the best known
algorithm for the problem. Write $\barr=\gamma/s$ for the \emph{bar}: a clique-sum algorithm whose cost is $2^{\eK S}$
beats the known algorithm only if $\eK<\barr$. In Section~\ref{sec:part2} the number of nodes that is available at instance size
$m$ is limited by the balancing of the reduction, $K\le m^{1/q}$ with a constant $q\ge 2$ (the balancing exponent of
Section~\ref{sec:part2}). We use only the following main-term reading of Section~\ref{sec:part2}, in which all lower-order terms are dropped.

\begin{quote}
\textbf{Threshold model.} An algorithm whose exponent on $K$-node instances is $\eK(K)$ gains at size $m$ only if
$\eK(K)<\barr$ for some $K$ with $q\log_2K\le\log_2 m$; up to lower-order terms, it gains at size $m$ if there is such
a $K$. The threshold $m^*$ is the least size at which it gains.
\end{quote}

\begin{proposition}[Threshold dichotomy]\label{prop:p3-dichotomy}
In the threshold model:
\begin{enumerate}[label=(\alph*)]
\item if $\eK(K)\ge C(\log_2K)^{-\alpha'}$ for all $K$, with constants $C,\alpha'>0$, then
$\log_2m^*\ge q\,(C/\barr)^{1/\alpha'}$;
\item if $\eK(K)\le CK^{-\beta}$ for all $K$, with constants $C,\beta>0$, then, up to the lower-order terms dropped
by the model, $\log_2 m^*\le (q/\beta)\log_2(C/\barr)+q$.
\end{enumerate}
\end{proposition}

\begin{proof}
(a) A gain at size $m$ needs a $K$ with $C(\log_2K)^{-\alpha'}\le\eK(K)<\barr$, i.e.\ $\log_2K>(C/\barr)^{1/\alpha'}$,
and $\log_2 m\ge q\log_2K$. (b) Take $K=\lceil (C/\barr)^{1/\beta}\rceil+1$; then $\eK(K)\le CK^{-\beta}<\barr$ and
$q\log_2K\le (q/\beta)\log_2(C/\barr)+q$ when $(C/\barr)^{1/\beta}\ge2$.
\end{proof}

So a polylogarithmic decay makes $\log m^*$ grow like a power of $1/\barr$, and a polynomial decay makes it grow
like $\log(1/\barr)$. The difference is large exactly where the bar is small. For general 3-SAT through the reduction of
Section~\ref{sec:part2}, $\gamma=0.4092$ and about $1270$ edge bits per variable (Section~\ref{sec:part2} counts both ends of an edge, $U/n\approx2540$)
give $\barr\approx3.2\cdot10^{-4}$ per edge bit. This bar rests on the sharp count, in Section~\ref{sec:part2}, of the bits of Theorem~D of
\cite{anthropic}, an estimate (Table~\ref{tab:p2-general}); a cheaper reduction would raise it.
With $C=1$ and $q=2$, (a) gives $\log_2m^*\ge 6.2\cdot10^3$ for $\alpha'=1$ and
$\ge1.9\cdot10^7$ for $\alpha'=1/2$, while (b) gives $\log_2 m^*\le 234$ for $\beta=0.1$. The case $C=1$, $\alpha'=1$
is exactly what Theorem~\ref{thm:p3-grid} below proves for grid algorithms (Corollary~\ref{cor:p3-grid-threshold});
the other two values only show the shape.
With the actual constants of the recursion of \cite{anthropic} (where $\alpha'=1/2$ for the pieces of its
Appendix~B), the thresholds of Section~\ref{sec:part2} are larger still. The constants that Section~\ref{sec:part2} improves (the bar through
sparsification and the parity kill, the per-level overheads, the balancing exponent $q$) move $C$, $\barr$ and $q$,
but not the type of decay. A method ``subtler than the recursion'' therefore means one thing precisely: an exponent
that decays polynomially in $K$. A recursion with constant growth $K_{l+1}=gK_l$ and a constant per-level factor
$\rho$ would be such a method, with $\beta=\log(1/\rho)/\log g$; the squaring $K_{l+1}=K_l^2$ of \cite{anthropic}
gives only polylogarithmic decay. The rest of this section shows that the squaring is forced for grid levels and for balanced homogeneous two-merge
levels of rewrite-and-merge.

\subsection{The rewrite-and-merge class}\label{sec:p3-class}

A lower bound needs a model. We want the narrowest one that still contains every algorithm of this kind that we know,
so that its lower bounds apply to them, and we want its cost to charge no more than \cite{anthropic} charges, so
that the bounds also apply to that paper's own accounting. The definition below abstracts one level of the
recursion: rewrite configurations of equality tests by exact identities, then evaluate every term by merging groups
of nodes. Three remarks after it say what the class contains, why balance is needed, and where it stops.

\paragraph{Instances.} A \emph{weighted clique sum} has nodes $[K]$, on each edge $e$ a string of $w_e\lambda$ bits
($w_e>0$), and at each node $a$ a table $W_a$, an arbitrary function of $a$'s copies of the strings on its edges; its
value is the sum over all copies of the product of the $W_a$ and of the equality tests $[\mathrm{eq}_e]$ between the
two copies of each edge string. Its size is $S=\lambda\sum_ew_e$. The \emph{exponent} of an algorithm is
$\limsup_{\lambda\to\infty}\sup_W \log_2(\text{time})/S$, and $\eK_{\mathcal C}(K)$ is the infimum of the exponents
over the algorithms of a class $\mathcal C$ on uniform instances ($w_e=1$, $S=\lambda\binom K2$). Every algorithm
reads a node table, and the largest one is indexed by $(K-1)\lambda=2S/K$ bits; so $\eK(K)\ge 2/K$, the \emph{input floor}.

\begin{definition}[Rewrite-and-merge level]\label{def:p3-level}
A level of a rewrite-and-merge algorithm on a weighted clique sum is given by the following data.
\begin{enumerate}[label=(L\arabic*)]
\item \emph{Configurations.} A partition of the edge set into configurations $C$, with edge set $E(C)$ and node set
$V(C)$; $f_C$ is the product of the equality tests of $E(C)$, a tensor whose legs are the nodes of $V(C)$ (the leg of
$x$ is $x$'s copies of the strings of $E(C)$). A single edge is a configuration.
\item \emph{Block identities.} Every string is cut into blocks of $m=m(\lambda)$ bits, $m\to\infty$ or $m=\lambda$.
For each configuration $C$ and block there is an exact identity $f_C|_{\text{block}}=\sum_{i\in I_C}\pi_{C,i}$ (up to a
nonzero scalar), with $|I_C|\le 2^{o(m)}$, where each piece $\pi_{C,i}$ is a tensor with the same legs. The level's
identity is the product of these sums over all (configuration, block) pairs; a \emph{term} $Q$ chooses one piece for
each pair, and every choice is a term (product identities).
\item \emph{Merges.} Partitions $P_1,\dots,P_p$ of $[K]$. Merge $k$ is usable only if every configuration meets at
most two groups of $P_k$.
\item \emph{Evaluation.} Each term $Q$ is evaluated under one usable merge $k(Q)$ of the algorithm's choice. A piece
of a configuration that meets two groups $G,G'$ of $P_k$ is written by a fixed rank decomposition across
$(V(C)\cap G\mid V(C)\cap G')$; its index becomes part of the merged edge $(G,G')$, and its \emph{cost}
$\pIIIcost_k(\pi)$ is the $\log_2$ of that rank. Pieces inside one group are multiplied into that group's table. The
merged instance is the weighted clique sum on the groups of $P_k$ whose edge $(G,G')$ carries the sum of the costs of
the pieces on it.
\item \emph{Cost.} The cost of the level is the sum over terms of the cost of building the merged node tables plus
the cost of solving the merged instance. Building a table costs at least its number of entries; no other charge for
building is made. The merged instance is solved by an algorithm of the class or by a base method, and the cost of
an algorithm of the class depends only on the weights of its instance.
\end{enumerate}
$\RM_p$ is the class of algorithms all of whose levels have at most $p$ merges. A partition is \emph{balanced} if all
its groups have the same size and it has between $2$ and $K/2$ groups. A level is \emph{homogeneous} if for each
merge $P_k$ every pair of groups of $P_k$ carries the same multiset of configuration types (a type being the
configuration together with the cells $G\cap H$, $G\in P_1$, $H\in P_2$, in which its nodes lie), and each type
occurs on each pair in a number of (configuration, block) incidences that tends to infinity with $\lambda$.
$\pIIIRMbh$ denotes the algorithms of $\RM_2$ all of whose levels are balanced and homogeneous.
\end{definition}

The cost in (L4) is the paper's ``costs $c$ bits'' of \cite[Sec.~2.2]{anthropic}: a function costs $c$ bits if it
is a sum over $\theta\in\{0,1\}^c$ of products of a function of each side, and the least such $c$ is at least the
$\log_2$ of the rank. So the class charges at most what the paper charges.

\begin{remark}[What the class contains]\label{rem:p3-contains}
The level of \cite{anthropic} is a $\pIIIRMbh$ level: the nodes form a $g\times g$ grid, the configurations are the
rectangles (two diagonal edges $[x=v][y=u]$, every edge whose endpoints differ in both row and column being a
diagonal of exactly one rectangle) and the edges within a row or column, the merges are rows and columns, there is
one block per string ($m=\lambda$) and $|I|=\mathrm{poly}(\lambda)$, and every pair of rows carries $g$ straight edges
and $g(g-1)/2$ rectangles. The homogeneity clause (many configurations of a type per pair of groups) holds as
$g\to\infty$, or at fixed $g$ in the per-block form; at fixed $g$ with one block per string the paper's level is
outside the clause as stated. One difference in accounting: the paper keeps the internal strings of a merged node as
settings and sums them by brute force per table entry, which costs at least one operation per entry, so the paper's
cost is at least the class's cost and every lower bound for the class applies to the paper's algorithm. Our
constructions of Section~\ref{sec:part1} at $K=4,6,9,16$ are balanced homogeneous two-merge levels with a tree or a $K=4$
method on top. The algorithm of Appendix~D of \cite{openai} uses the polynomial method and is not in the class.
\end{remark}

The balance condition in the definition is not a convenience: without it the size charge of (L5) can be gamed.

\begin{remark}[Balance is necessary]\label{rem:p3-balance}
Without balance the class bound fails. Take $P_1=\{[K]\setminus\{u\},\{u\}\}$. The group table of $K-1$ nodes has
$2^{(K-1)\lambda}=2^{2S/K}$ entries, and the merged instance has two nodes and $2S/K$ bits, so the modelled cost is
$2^{2S/K}$ and $\eK_{\RM_2}(K)\le 2/K$, the input floor. The true cost of building that table is the cost of a
$(K-1)$-node clique sum with external legs, which a size charge does not see. Balance excludes this; a stronger
building charge for large groups would be the alternative.
\end{remark}

The other restriction, (L3), limits which configurations a merge can use; it is the one that matters for what the
class leaves open.

\begin{remark}[Higher-arity pieces are inside the class; hyperedges are not]\label{rem:p3-hyper}
(L1) allows configurations of any size, subject only to (L3). So pieces of high arity are in the class as long as
their configuration meets at most two groups of each merge. What (L3) excludes is a configuration that meets three
or more groups of a merge, which would become a hyperedge of the merged instance. In particular, summing the terms
of a product set of block choices gives coarser pieces of a coarser configuration inside the class, while summing
terms over a set of choices that is not a product set gives a tensor on several blocks' legs, i.e.\ a piece of a
configuration that may meet three groups.
\end{remark}

\subsection{Consequences of the definition}\label{sec:p3-basic}

This subsection collects four facts that follow directly from the definition: every identity has a piece that is
expensive across a given split (the pigeonhole floor), sharing work across terms cannot go below the sizes of the
tables (the sharing cap), a configuration can at best split its cost evenly between two merges (the two-merge floor),
and few edges lie in configurations that could trade at all (the count of four-cell edges).

Let $C$ be a configuration and $(A\mid V(C)\setminus A)$ a split of its nodes across which $r$ edges of $C$ cross.
Then $f_C|_{\text{block}}$ has rank $2^{rm}$ across the split, since it is the identity matrix on the crossing strings
tensored with factors on one side.

\begin{proposition}[Pigeonhole floor]\label{prop:p3-pigeon}
In every block identity of (L2), some piece has $\log_2$-rank at least $rm-\log_2|I_C|$ across the split.
\end{proposition}
\begin{proof}
Rank is subadditive and the pieces sum to a tensor of rank $2^{rm}$.
\end{proof}

The class charges a term by the sizes of its merged tables. One might hope to do better by sharing computation
between terms; the next proposition shows that sharing cannot go below the table sizes.

\begin{proposition}[Sharing cap]\label{prop:p3-sharing}
Let a term $T_Q$ be evaluated in any way that is bilinear across some split $(X\mid Y)$ of the nodes, i.e.\ as
$\sum_\sigma t_X(\sigma)\,t_Y(\sigma)$ where $t_X$ depends on the tables of $X$ and $t_Y$ on those of $Y$; for
instance, by a contraction tree whose tables are materialised, computed by any algorithm and shared with other terms
in any way. Then the number of strings $\sigma$, and hence the cost of writing the table $t_X$, is at least
$2^{\sum_\pi \pIIIcost_{X}(\pi)}$, where $\pIIIcost_X(\pi)$ is the $\log_2$-rank of the piece $\pi$ across the split
induced by $(X\mid Y)$ on its legs (raw edges counted as pieces of full rank).
\end{proposition}
\begin{proof}
A term is a tensor product of its pieces on disjoint legs, so its rank across $(X\mid Y)$ is the product of the
pieces' ranks across the induced splits. Pairing the term with arbitrary node tables of $X$ and $Y$ gives a matrix,
indexed by the tables of $X$ and of $Y$, that factors through the strings $\sigma$; since products of node tables
span all tensors on the legs, its rank is the rank of the term across $(X\mid Y)$. Writing $t_X$ costs at least
its number of entries.
\end{proof}

In the class the merged edge weights of (L4) are exactly these log-ranks, so Proposition~\ref{prop:p3-sharing} says
that sharing computation across terms can lower the cost of an algorithm at most to the cost of writing the
materialised tables of its terms; the size game of Section~\ref{sec:p3-size} computes that floor. Summing terms
before contracting does not escape it when the sum is over a product set of block choices, since that only produces
a coarser Signed Network Sum to which the proposition applies again; a non-product grouping is a higher-arity piece
(Remark~\ref{rem:p3-hyper}), to which the proposition applies with a cost we cannot evaluate. A sharper form
(the rank equals the number of strings in a minimal decomposition) is not needed here.

\begin{proposition}[Two-merge floor]\label{prop:p3-twomerge}
Let $C$ meet two groups of each of two merges, with induced splits $S_1,S_2$ of $V(C)$ across which $r_1,r_2$ of its
edges cross. For every $\theta\in[0,1]$, taking in a fraction $\theta$ of the blocks a piece of
Proposition~\ref{prop:p3-pigeon} for $S_1$ and in the rest one for $S_2$, every such term pays on $C$ at least
$\theta r_1-o(1)$ per $\lambda$ under merge~1 and at least $(1-\theta)r_2-o(1)$ under merge~2. In particular the
per-level factor of any two-merge identity on $C$ satisfies $\rho\ge1/2$.
\end{proposition}
\begin{proof}
Costs add over blocks, and $\log_2|I_C|=o(m)$. When $m=\lambda$ the fractions are taken over the configurations of a
type on a pair of groups, which homogeneity makes numerous.
\end{proof}

For the rectangle (two diagonal edges of a $2\times2$ grid) this is the observation of \cite[App.~H]{anthropic} that
one piece needs $2\lambda-O(\log\lambda)$ bits row-wise, which leaves it at most about $(4\rho-2)\lambda$
column-wise; the proposition is its general form, for every configuration.

\paragraph{Four-cell reduction.} Under two merges $P_1,P_2$ the nodes of a configuration meeting at most two groups
$G,G'$ of $P_1$ and $H,H'$ of $P_2$ fall into at most four cells, $11=G\cap H$, $12=G\cap H'$, $21=G'\cap H$,
$22=G'\cap H'$. Both splits are unions of cells, so grouping each cell into one party changes neither the ranks of
pieces across $S_1,S_2$ nor which decompositions exist. Hence $f_C$ is the graph tensor of a multigraph on the four
cells, with multiplicities $d$ (on $11$--$22$), $a$ ($12$--$21$), $v_1$ ($11$--$21$), $v_2$ ($12$--$22$), $h_1$
($11$--$12$), $h_2$ ($21$--$22$), times a factor internal to the cells. The split $S_1=\{11,12\}\mid\{21,22\}$ crosses
$d,a,v_1,v_2$ and $S_2=\{11,21\}\mid\{12,22\}$ crosses $d,a,h_1,h_2$. Every question about trading cost between two
merges is a question about this six-parameter tensor.

\begin{proposition}[Count of four-cell edges]\label{prop:p3-count}
Let $c_k$ be the largest group size of $P_k$. At most $K(c_1c_2+c_1+c_2)/2$ edges lie in configurations with four
nonempty cells. For balanced partitions with $n_k$
groups, the fraction $t$ of edges in four-cell configurations satisfies
$t\le (c_1c_2+c_1+c_2)/(K-1)=\bigl(K/(n_1n_2)+1/n_1+1/n_2\bigr)\,K/(K-1)$.
\end{proposition}
\begin{proof}
Fix an endpoint $u\in G\cap H$ of such an edge. If the other endpoint lies in another group of $P_1$ and of $P_2$,
it lies in a cell $G'\cap H'$ of the same configuration with $G\cap H'\neq\emptyset$, so $H'$ is one of at most $c_1$
groups that meet $G$, and there are at most $c_1c_2$ such endpoints. Otherwise it lies in $H$ (at most $c_2$
choices) or in $G$ (at most $c_1$). A configuration with four nonempty cells meets two groups of each merge.
\end{proof}

For two partitions into groups of size at most $c$ this gives at most $K(c^2+2c)/2$ edges, a fraction $O(c^2/K)$.
So a level in which a constant fraction of the bits can trade needs groups of size $\Omega(\sqrt K)$, that is, at
most $O(\sqrt K)$ merged nodes: the squaring of \cite{anthropic}, forced by the locality of the trade, \emph{provided}
that only four-cell configurations trade. The next section shows that this proviso is a statement about tensors
that flattening ranks cannot decide, and proves the part of it that we need.

\subsection{Which configurations can trade}\label{sec:p3-cells}

The count of four-cell edges is useful only if configurations on fewer cells cannot trade cost between the two
merges. This subsection proves that they cannot (Lemma~\ref{lem:p3-Rprime}). The example first shows why the natural
tool, flattening ranks, is not enough; the proof uses quantum functionals instead.

\begin{example}[Flattening inequalities do not suffice]\label{ex:p3-triangle}
Let $C$ be a triangle with one node in each of the cells $11,12,22$ ($d=h_1=v_2=1$, cell $21$ empty), with $N=2^m$
values per block string. Then $S_1$ isolates the party $22$ and $S_2$ isolates the party $11$, and $f_C$ is the matrix
multiplication tensor $\langle N,N,N\rangle$. Consider pieces of the two forms
$A=\alpha(11)\otimes M(12,22)$ and $B=\beta(22)\otimes M'(11,12)$ with $M,M'$ of rank $N^2$: $A$ is full across $S_1$
(rank $N^2$) and free across $S_2$ (rank $1$), $B$ the reverse. Each satisfies every three-party flattening-rank
inequality (for $A$, $r_{22}\le r_{11}r_{12}$ holds with $r_{12}=N^2$), and a family containing one of each kind
passes the pigeonhole test of Proposition~\ref{prop:p3-pigeon} for both splits. If a family of $2^{o(m)}$ such
pieces summed to $f_C$, the half-half mixture would make every term pay half of the triangle's bits under each
merge, a trade on a configuration with no rectangle. So no argument that uses only the flattening ranks of the
pieces can show that three-cell configurations do not trade; the invariant that excludes such families is the
quantum functional.
\end{example}

We use the quantum functionals $\pIIIF$ of Christandl, Vrana and Zuiddam \cite{cvz19}, defined for complex tensors
and a probability vector $\theta$ on the parties, with the following properties: $\pIIIF$ is monotone under
restriction, additive on direct sums, multiplicative on tensor products and normalised ($\pIIIF$ of the rank-$r$
unit tensor is $r$); and $\pIIIF(T)\le\prod_x R_x(T)^{\theta_x}$, where $R_x$ is the flattening rank of party $x$
against the rest. A sum $T+T'$ is a restriction of $T\oplus T'$ (each party adds its two components), so
$\pIIIF(\sum_iT_i)\le\sum_i\pIIIF(T_i)$. On a graph tensor, in which every edge is an $N\times N$ identity matrix
between its two parties, $\pIIIF=N^{\sum_{\{x,y\}}(\theta_x+\theta_y)}$, the sum over edges between distinct parties
(edges inside a party contribute a vector, i.e.\ a factor $1$).

\begin{lemma}[Configurations with at most three cells do not trade]\label{lem:p3-Rprime}
Under two merges, let $C$ be a configuration with at most three nonempty cells. For every block identity of (L2)
on $C$ there is a piece that pays, under each merge whose split crosses $C$, its full cost up to $O(\log_2|I_C|)$ per
block: $\pIIIcost_1\ge r_1m-2\log_2|I_C|$ and $\pIIIcost_2\ge r_2m-2\log_2|I_C|$.
\end{lemma}

\begin{proof}
With at most two nonempty cells, either $C$ crosses at most one merge, and Proposition~\ref{prop:p3-pigeon} gives
the piece, or the cells are opposite ($11$ and $22$, or $12$ and $21$), the two splits coincide on $V(C)$, and again
Proposition~\ref{prop:p3-pigeon} gives it. With three nonempty cells, let the empty one be $21$; the other cases are
symmetric. On $V(C)$ the split $S_1$ is $\{11,12\}\mid\{22\}$, which isolates $x=22$, and $S_2$ is
$\{11\}\mid\{12,22\}$, which isolates $y=11$; the third cell $z=12$ is the corner. Then $r_1$ is the number of
edges of $C$ at $x$ that leave $x$'s cell, and $r_2$ the same for $y$.

First fix the strings of all edges inside cells to one value; this restricts $f_C|_{\text{block}}$ to the graph
tensor $f'$ of the three-party multigraph, and each piece $\pi$ to a tensor $\pi'$ with $\sum\pi'=f'$ and with ranks
across $S_1,S_2$ at most those of $\pi$. Take $\theta=(1/2,1/2,0)$ on $(x,y,z)$. On $f'$ an $x$--$y$ edge contributes
$N^{1}$ and an $x$--$z$ or $y$--$z$ edge $N^{1/2}$, so $\pIIIF(f')=2^{m(r_1+r_2)/2}$. For a piece,
$\pIIIF(\pi')\le R_x(\pi')^{1/2}R_y(\pi')^{1/2}$, and the flattenings at $x$ and at $y$ are the splits $S_1$ and
$S_2$; so $\pIIIF(\pi')\le 2^{(\pIIIcost_1(\pi)+\pIIIcost_2(\pi))/2}$. By sub-additivity some piece has
$(\pIIIcost_1+\pIIIcost_2)/2\ge m(r_1+r_2)/2-\log_2|I_C|$.

It remains to split this sum, and here the isolation of single cells is essential: in $f'$ the leg of $x$ has
dimension exactly $2^{mr_1}$, so no restricted piece is over-full across $S_1$, $\log_2R_x(\pi')\le mr_1$, and
likewise $\log_2R_y(\pi')\le mr_2$. With the bound on the sum this gives $\log_2R_x(\pi')\ge mr_1-2\log_2|I_C|$ and
$\log_2R_y(\pi')\ge mr_2-2\log_2|I_C|$, and the costs of $\pi$ are at least these.
\end{proof}

The step without over-full pieces is the one that fails elsewhere. For the $4$-cycle ($d=a=0$, only $v$- and
$h$-edges), the inequality of Song, Chen, Sun and Hu (Proposition~\ref{prop:p3-scsh} below) applied across the third
split $\{11,22\}\mid\{12,21\}$ gives a piece with $\pIIIcost_1+\pIIIcost_2\ge (v+h)m$ up to $o(m)$, but the side
$\{11,12\}$ of $S_1$ has dimension $2^{m(v+2h_1)}$, so pieces can be over-full across $S_1$ or $S_2$, and a
three-piece family that is consistent with every constraint we know defeats every mixture. So the $4$-cycle is
settled only for families without over-full pieces. The smallest open case of a general statement is $K_4$ minus a
diagonal ($d=1$, $a=0$, $v_1=v_2=h_1=h_2=1$): we argue that the known flattening inequalities and the functional with
$\theta=1/2$ on the two ends of the diagonal allow a family that saves one bit per block under each merge, and
whether such a family exists is open. Neither case is needed below.

The conjectured general form (we call it the rectangle-content bound) would say that only $\min(d,a)$ bits of a
four-cell configuration can be traded between the two merges, one for one, and that $v$-, $h$- and excess diagonal
edges always pay in full; for $d=a=1$ it is the floor $\rho\ge1/2$ of Proposition~\ref{prop:p3-twomerge}. It would
improve the constants of Theorem~\ref{thm:p3-A} below, not its structure.

Lemma~\ref{lem:p3-Rprime}, the two-merge floor and Proposition~\ref{prop:p3-count} together are a counting lemma
with a proved premise: in $\RM_2$, only edges of four-cell configurations can trade, there are at most
$K(c_1c_2+c_1+c_2)/2$ of them, and they save at most half of their bits under each merge.

\subsection{Polylogarithmic decay in balanced homogeneous two-merge algorithms}\label{sec:p3-thmA}

This subsection turns the counting lemma into a bound on the decay of $\eK(K)$: first for balanced homogeneous
two-merge algorithms, with exponent $3$ (Theorem~\ref{thm:p3-A}), then for grid algorithms with the exact exponent $1$
(Theorem~\ref{thm:p3-grid}), which gives the threshold bound of Corollary~\ref{cor:p3-grid-threshold}, and finally for
the recursion of \cite{anthropic} as written, with one block per string (Theorem~\ref{thm:p3-grid-sb}).

Write $f(n)=1/\eK_{\pIIIRMbh}(n)$. We need one condition on base methods: a base method used on $n$-node instances
has exponent at least $(\log_2n)^{-3}$. By the input floor this holds for every method on at most $2^{11}$ nodes,
since $2/n\ge(\log_2n)^{-3}$ there (the ratio $(n/2)/(\log_2n)^3$ is at most $1$ on $[2,2^{11}]$, with value $1$ at
$n=2$ and $1024/1331$ at $n=2^{11}$), and for larger $n$ it holds for every method whose exponent is at least
$1/1331$, for instance every contraction tree, whose largest table crosses a cut with between $n/3$ and $2n/3$ nodes
on one side.

Consider a level on $K$ nodes with two balanced merges, $n_k$ groups of size $c_k=K/n_k$ in $P_k$
($2\le n_k\le K/2$). The fraction of edges internal to the groups of $P_k$ is $(c_k-1)/(K-1)\le1/n_k$; write
$a_k\ge1-1/n_k$ for the fraction that crosses $P_k$, and $t$ for the fraction in four-cell configurations. A term is
evaluated under one merge, and its merged instance under merge $k$, with $B_k$ bits, costs at least
$2^{\eK(n_k)B_k}$ if it has uniform weights. The cost of the level is at least the largest, over terms, of the
cheaper of its two merges. Two inequalities follow.

\begin{enumerate}[label=(F\arabic*)]
\item $f(K)\le f(n_1)/a_1+f(n_2)/a_2$. The adversary takes, for every configuration crossing both merges, the two
pieces of Proposition~\ref{prop:p3-twomerge} in proportions $\theta$ and $1-\theta$, and for every configuration
crossing one merge the piece of Proposition~\ref{prop:p3-pigeon}. Every merged edge of merge~1 then carries at least
$\theta$ times its full weight, which is the same for all pairs of groups by balance, so, bounding the merged instance through its shortest edge as in the proof of Theorem~\ref{thm:p3-grid-sb}, it
costs at least the uniform instance with $\theta a_1S$ bits; likewise for merge~2 with $(1-\theta)a_2S$. Equalising
$\eK(n_1)\theta a_1=\eK(n_2)(1-\theta)a_2$ gives $\eK(K)\ge \bigl(f(n_1)/a_1+f(n_2)/a_2\bigr)^{-1}$.
\item $f(K)\le\max_k f(n_k)/(a_k-t)$. The adversary takes, for every configuration with at most three cells, the
piece of Lemma~\ref{lem:p3-Rprime}, and for the four-cell configurations a piece depending only on the type. By
homogeneity the merged instances are uniform, and under merge $k$ they carry at least $(a_k-t)S$ bits.
\end{enumerate}
A level with one merge is (F2) with $t=0$. All bounds hold up to factors $1-o(1)$ as $\lambda\to\infty$, which do not
affect exponents.

\begin{theorem}\label{thm:p3-A}
For every $K\ge2$, every algorithm of $\pIIIRMbh$ whose base methods satisfy the base condition has exponent at least
$(\log_2K)^{-3}$ on uniform $K$-node instances; that is, $f(K)\le(\log_2K)^3$.
\end{theorem}

\begin{proof}
Strong induction on $K$; write $L=\log_2K$ and $l_k=\log_2n_k\le L-1$. For $K\le2^{11}$ the input floor gives
$f(K)\le K/2\le L^3$. Let $L\ge11$; an algorithm on $K$ nodes either is a base method, covered by the base condition,
or starts with a level, to whose merged instances the induction hypothesis $f(n_k)\le l_k^3$ applies, together with the
input floor $f(n_k)\le n_k/2$.

\emph{Case 1: $t\le1/L$.} Use (F2). If $n_k\ge4L$, then $a_k-t\ge1-1/(4L)-1/L=1-1.25/L$ and
$f(n_k)/(a_k-t)\le(L-1)^3/(1-1.25/L)\le L^3$, because $(1-1/L)^3\le1-3/L+3/L^2\le1-1.25/L$ for $L\ge2$. If $n_k<4L$,
then $f(n_k)\le n_k/2<2L$ and $a_k-t\ge1-1/2-1/11>0.4$, so $f(n_k)/(a_k-t)<5L\le L^3$.

\emph{Case 2: $t>1/L$.} By Proposition~\ref{prop:p3-count},
$K/(n_1n_2)+1/n_1+1/n_2>(K-1)/(KL)\ge0.99/L$, so one of the three terms exceeds $0.33/L$: $n_1<3.1L$, or
$n_2<3.1L$, or $n_1n_2<3.1KL$.

\emph{Case 2a: $n_1<3.1L$} (the case $n_2<3.1L$ is symmetric). Use (F1). The first term is
$f(n_1)/a_1\le(n_1/2)/(1-1/n_1)=n_1^2/(2(n_1-1))\le n_1<3.1L$. If $n_2\ge4L$, the second is at most
$(L-1)^3\cdot4L/(4L-1)\le(L-1)^3(1+0.3/L)\le L^3-2.7L^2+3.3L$, and the total is at most $L^3-2.7L^2+6.4L\le L^3$. If
$n_2<4L$, the second is at most $n_2<4L$ and the total is below $7.1L\le L^3$.

\emph{Case 2b: $n_1,n_2\ge3.1L$ and $n_1n_2<3.1KL$.} Then $l_1+l_2<L+l_0$ with $l_0=\log_2(3.1L)$, and
$l_k\in[l_0,L-1]$. By (F1) and $1/a_k\le1+\delta$, $\delta=1/(3.1L-1)$,
$f(K)\le(1+\delta)(l_1^3+l_2^3)$. The convex function $l_1^3+l_2^3$ on this polygon is largest at a vertex, which is
$(L-1,l_0+1)$ or its mirror image, so $f(K)\le(1+\delta)\bigl((L-1)^3+(l_0+1)^3\bigr)$. Since
$(L-1)^3\le L^3-3L^2+3L$, this is at most $L^3$ provided
\[(1+\delta)(l_0+1)^3+3L+\delta L^3\le3L^2 .\]
At $L=11$ the left side is $232.9+33+40.2=306.1$ and the right side $363$; the difference of the two sides grows
with $L$.
\end{proof}

The exponent $3$ is an artefact of the proof. The square grid with ideal pieces (pieces full in one direction and
free in the other, which reach $\rho=1/2$) attains $f(K)=\log_2K$, so no bound of the form $f(K)\le(\log_2K)^{\gamma}$
with $\gamma<1$ can hold if such pieces exist; the true exponent for $\pIIIRMbh$ lies between $1$ and $3$ in that
case, and the pieces of \cite[App.~B]{anthropic}, which exist, give $\eK(K)$ of order $(\log K)^{-1/2}$.

\paragraph{Grid algorithms: the exact exponent.} For the levels that all known constructions use, the induction
closes with exponent $1$, because the bookkeeping of a grid is exact. A \emph{grid level} on $K=ab$ nodes
($a,b\ge2$) places the nodes on an $a\times b$ grid; $P_1$ is the partition into the $a$ rows, $P_2$ into the $b$
columns; the configurations are the raw row edges, the raw column edges, and the rectangles (the two diagonal edges
of each $2\times2$ subgrid), each rectangle expanded by its own exact per-block identity, with $\lambda/m\to\infty$
blocks per string. A \emph{grid algorithm} uses only grid levels (with any $a,b$ and any piece families, which may
change from level to level) and base methods of exponent at least $1/\log_2n$ on $n$ nodes. Contraction trees and
brute force satisfy the base condition for all $n\ge2$, since their exponent is at least $\lceil n/3\rceil
(n-\lceil n/3\rceil)/\binom n2\ge 1/\log_2 n$. The recursion of \cite{anthropic} with its brute-force or tree base
and its rectangle identities applied per block, and the per-block constructions of Section~\ref{sec:part1} (the $K=4$ constructions
are $2\times2$ levels), are grid algorithms once the grouping of the merged nodes (their Lemma~8) and the padding
to a $(K,L)$-clique are admitted as steps, which Theorem~\ref{thm:p3-grid-sb} below does; the recursion as written,
with one block per string, is covered by that theorem, and the rank-counted grouping of Section~\ref{sec:part1}, which acts on the
whole string, is not covered.

\begin{theorem}\label{thm:p3-grid}
Every grid algorithm has exponent at least $1/\log_2K$ on $K$-node instances: $f(K)\le\log_2K$.
\end{theorem}

\begin{proof}
Strong induction on $K$. If the algorithm is a base method, this is the base condition. Otherwise it starts with an
$a\times b$ grid level. The fraction of edges inside rows is $i_1=(b-1)/(K-1)$, inside columns $i_2=(a-1)/(K-1)$,
and the rectangle diagonals make up the rest, $t=(a-1)(b-1)/(K-1)$; so $i_1+i_2+t=1$ exactly. Let the adversary take,
in every rectangle and in the same proportions $\theta$ and $1-\theta$ of its blocks, the piece full across the row
split and the piece full across the column split (Proposition~\ref{prop:p3-twomerge}). Under the row merge every
merged edge then carries the same number of bits, its column edges in full and $\theta$ of its rectangle bits, so by
restriction the merged instance on $a$ nodes costs at least $2^{\eK(a)(i_2+\theta t)S}$; under the column merge,
$2^{\eK(b)(i_1+(1-\theta)t)S}$. Write $x_1=\eK(a)$, $x_2=\eK(b)$. If the $\theta$ that equalises the two exponents lies
in $[0,1]$, the common value is $x_1x_2(i_1+i_2+t)/(x_1+x_2)=x_1x_2/(x_1+x_2)$, so
$f(K)\le f(a)+f(b)\le\log_2a+\log_2b=\log_2K$. Otherwise, say it exceeds $1$; then at $\theta=1$ the row merge is the
cheaper one and $f(K)\le f(a)/(1-i_1)$. Since $\log_2c/\log_2K$ is concave in $c$ and agrees with $(c-1)/(K-1)$ at
$c=1$ and $c=K$, we have $i_1=(b-1)/(K-1)\le\log_2b/\log_2K$, hence $1-i_1\ge\log_2a/\log_2K$ and
$f(K)\le\log_2a\cdot\log_2K/\log_2a=\log_2K$. The case below $0$ is symmetric.
\end{proof}

The bound cannot be improved by this argument alone, because it would be attained by pieces whose existence is open.

\begin{remark}[tightness]
For $K=4^{2^j}$ and in the cost model that charges tables by their size only, the bound is attained by square grids
with pieces that are full in one direction and free in the other, if such pieces exist: on a $g\times g$ grid the
equalised value gives $f(g^2)=2f(g)$ exactly when the pieces reach $\rho=1/2$.
\end{remark}

Through the dichotomy, Theorem~\ref{thm:p3-grid} turns into a lower bound on the threshold. This is the bound behind
the size $2^{6000}$ of the introduction. It holds for every choice of pieces and every schedule; only a cheaper
reduction, which raises the bar $\barr$, could lower it.

\begin{corollary}\label{cor:p3-grid-threshold}
In the threshold model, every grid algorithm, in particular the recursion of \cite{anthropic} with its rectangle
identities applied per block or as written (Corollary~\ref{cor:p3-sb-anthropic} below) and any piece families, has $\log_2 m^*\ge q/\barr$; for general 3-SAT through the reduction of
Section~\ref{sec:part2} ($\barr\approx3.2\cdot10^{-4}$ per edge bit, $q\ge2$) this is $\log_2 m^*\ge6.2\cdot10^3$. For the single-block form of the recursion as written the bound holds up to a factor $1-10^{-21}$
(Corollary~\ref{cor:p3-sb-anthropic} below).
\end{corollary}
\begin{proof}
Proposition~\ref{prop:p3-dichotomy}(a) with $C=1$, $\alpha'=1$ and Theorem~\ref{thm:p3-grid}.
\end{proof}

For comparison only (no lower bound of this form is proved for them), the decay that the pieces actually known
achieve, $\alpha'=1/2$, would correspond to $\log_2 m^*\approx q/\barr^{2}\approx1.9\cdot10^7$; Section~\ref{sec:part2} has the exact
numbers. So a usable threshold needs levels that are not grids.

\paragraph{One block per string.} When a whole string is one block ($m=\lambda$, as in the recursion of
\cite{anthropic}, whose pieces depend on $\lambda$), the mixture of Proposition~\ref{prop:p3-twomerge} must be taken
across rectangles instead of across blocks, and a row pair and a column pair share only one rectangle. The merged
instance is then only nearly uniform, so we state the induction for weighted instances and use only their shortest
edge. Write $\beta(I)$ for the least number of bits on an edge of an instance~$I$.

\begin{definition}[Single-block grid algorithms]\label{def:p3-sb}
A \emph{single-block grid algorithm} is built from the following steps, each of which solves the merged instances of
its terms by a single-block grid algorithm that may depend on the term and on $\lambda$:
\begin{enumerate}[label=(\roman*)]
\item a \emph{grid level} as above with one block per string ($m=\lambda$);
\item a \emph{grouping level}: one balanced merge into $n$ groups of $c=K/n$ nodes, with single edges as
configurations (Lemma~8 of \cite{anthropic}, case~(ii));
\item \emph{padding}: lengthening edge strings by bits that the tables ignore;
\item a \emph{base method} satisfying the weighted base condition: on every $n$-node instance $I$ it costs at least
$2^{\beta(I)\binom n2/\log_2n}$.
\end{enumerate}
We assume a function $\sigma(\lambda)=o(\lambda)$ with $\log_2|I_C|\le\sigma(\lambda)$ for every identity used at
parameter $\lambda$. Let $h_A$ be the largest number of grid and grouping levels on a chain of nested steps of $A$,
and define $\Pi_A$ recursively: $\Pi_A=1$ for a base method and for a $2\times2$ level; a padding step or a grouping
level has the largest $\Pi$ of the algorithms it runs; an $a\times b$ grid level with $a,b\ge3$ has that largest
value divided by $1-\bar\eps_{a,b}$, where $\bar\eps_{a,b}=2/a^2+2/b^2$.
\end{definition}

Brute force and contraction trees satisfy the weighted base condition, by the bound on their largest table used for
the base condition above, since every edge carries at least $\beta(I)$ bits; and every algorithm satisfies it on
$4$ nodes, since a node table of $I$ is indexed by at least $3\beta(I)=\beta(I)\binom42/\log_24$ bits (the input floor).

\begin{lemma}[Rounding]\label{lem:p3-sturm}
For integers $\mu,\nu\ge1$ and $\theta\in[0,1]$, the matrix $X_{ij}=\lfloor(i+j+1)\theta\rfloor-\lfloor(i+j)\theta\rfloor$
($0\le i<\mu$, $0\le j<\nu$) has entries in $\{0,1\}$, row sums between $\lfloor\nu\theta\rfloor$ and
$\lceil\nu\theta\rceil$, and column sums between $\lfloor\mu\theta\rfloor$ and $\lceil\mu\theta\rceil$.
\end{lemma}

\begin{proof}
The entries are $0$ or $1$ because $0\le\theta\le1$. The sum of row $i$ telescopes to
$\lfloor(i+\nu)\theta\rfloor-\lfloor i\theta\rfloor$, and $\lfloor x+y\rfloor-\lfloor x\rfloor$ is $\lfloor y\rfloor$
or $\lfloor y\rfloor+1$, and equals $y$ when $y$ is an integer. Columns are the same with $\mu$ in place of $\nu$.
\end{proof}

With the rounding lemma the adversary can spread the two kinds of pieces over the rectangles almost evenly, and the
proof of Theorem~\ref{thm:p3-grid} goes through for weighted instances, at a loss that is small when the grid sides
are large.

\begin{theorem}[One block per string]\label{thm:p3-grid-sb}
Let $A$ be a single-block grid algorithm all of whose grid levels have both sides at least $3$ or are $2\times2$.
For every $\lambda$ and every instance $I$ on $K$ nodes,
\[\log_2\mathrm{cost}(A,I)\ \ge\ \bigl(\beta(I)-h_A\sigma(\lambda)\bigr)\binom K2\Big/\bigl(\Pi_A\log_2K\bigr).\]
In particular $1/\eK(K)\le\Pi_A\log_2K$ on uniform instances, and if the grid levels on every chain of $A$ satisfy
$\sum_l\bar\eps_{a_l,b_l}\le\delta<1$, then $1/\eK(K)\le\log_2K/(1-\delta)$.
\end{theorem}

\begin{proof}
Induction on the structure of $A$; write $\beta=\beta(I)$, $\sigma=\sigma(\lambda)$, $h=h_A$, and assume
$\beta>h\sigma$, since otherwise there is nothing to prove. Each level costs at least the cost of solving the merged
instance of any one of its terms under the merge chosen for that term. We use Proposition~\ref{prop:p3-pigeon} in
the form: a configuration whose edges crossing a split carry $r$ bits has a piece of $\log_2$-rank at least
$r-\sigma$ across it. For a base method the claim is the weighted base condition, and for a $2\times2$ level it is the
input floor shown above. Padding raises $\beta$, so the claim for the algorithm that is run gives it.

\emph{Grouping level.} Every one of the $c^2$ edges between two groups pays at least $\beta-\sigma$, so every merged
edge has at least $c^2(\beta-\sigma)$ bits, and the induction hypothesis (with at most $h-1$ levels and a factor at most
$\Pi_A$) gives $\log_2\mathrm{cost}\ge(\beta-h\sigma)c^2\binom n2/(\Pi_A\log_2n)$, using $c^2\ge1$. Here
$c^2\binom n2=(1-i)\binom K2$ with $i=(c-1)/(K-1)$, and $i\le\log_2c/\log_2K$, since $\log_2x/\log_2K$ is concave and
agrees with $(x-1)/(K-1)$ at $x=1$ and $x=K$. So $(1-i)/\log_2n\ge1/\log_2K$, which is the claim.

\emph{Grid level, $a,b\ge3$.} Let $\Pi'$ be the largest $\Pi$ of the algorithms run on its merged instances, so
$\Pi_A=\Pi'/(1-\bar\eps_{a,b})$, and put $x_1=1/(\Pi'\log_2a)$, $x_2=1/(\Pi'\log_2b)$. With $i_1,i_2,t$ as in the
proof of Theorem~\ref{thm:p3-grid}, let $\eps_1=2\binom a2/\binom K2$ and $\eps_2=2\binom b2/\binom K2$. For every
rectangle there is a piece of $\log_2$-rank at least $2\beta-\sigma$ across the row split and one across the column
split. Index the $\mu=\binom a2$ row pairs by $i$ and the $\nu=\binom b2$ column pairs by $j$; given $\theta\in[0,1]$,
let the term $Q$ take in rectangle $(i,j)$ the first piece if $X_{ij}=1$ (Lemma~\ref{lem:p3-sturm}) and the second
otherwise. Every choice of pieces is a term, so $Q$ is one. Under the row merge a merged edge carries the $b$ column
edges between its rows and at least $\lfloor\theta\nu\rfloor\ge\theta\nu-1$ rectangles with the first piece, so it has
at least $b(\beta-\sigma)+\lfloor\theta\nu\rfloor(2\beta-\sigma)\ge b(\beta-\sigma)+2\lfloor\theta\nu\rfloor(\beta-\sigma)
\ge(b+2\theta\nu-2)(\beta-\sigma)$ bits. Since
$b+2\theta\nu-2\ge1$, the induction hypothesis gives
\[\log_2\mathrm{cost}_{\mathrm{row}}\ge(\beta-h\sigma)\binom K2\,R(\theta),\qquad R(\theta)=x_1(i_2+\theta t-\eps_1),\]
because $b\binom a2=i_2\binom K2$, $2\nu\binom a2=t\binom K2$ and $2\binom a2=\eps_1\binom K2$. In the same way,
with at least $\mu-\lceil\theta\mu\rceil\ge(1-\theta)\mu-1$ rectangles of the second piece on each column pair,
$\log_2\mathrm{cost}_{\mathrm{col}}\ge(\beta-h\sigma)\binom K2\,C(\theta)$ with
$C(\theta)=x_2(i_1+(1-\theta)t-\eps_2)$. The bounds hold for every $\theta$, because $x_1,x_2$ do not depend on the
term. $R$ increases and $C$ decreases. If they meet at some $\theta\in[0,1]$, the common value $V$ satisfies
$V/x_1+V/x_2=i_1+i_2+t-\eps_1-\eps_2=1-\eps_1-\eps_2$, so $V=(1-\eps_1-\eps_2)/(\Pi'\log_2K)$. If $R(1)<C(1)$, take
$\theta=1$: since $1-i_1=b(a-1)/(K-1)$ we have $\eps_1=(2/b^2)(1-i_1)$, and $1-i_1\ge\log_2a/\log_2K$ as above, so
$R(1)\ge(1-2/b^2)/(\Pi'\log_2K)$; the case $C(0)<R(0)$ is symmetric with $2/a^2$. Finally
$\eps_1+\eps_2\le\bar\eps_{a,b}$, since $a(a-1)/(K(K-1))\le a^2/K^2=1/b^2$ and likewise for $b$. So in every case
$\min(R,C)\ge(1-\bar\eps_{a,b})/(\Pi'\log_2K)=1/(\Pi_A\log_2K)$, which is the claim. The last statement follows from
$\prod_l(1-\bar\eps_l)\ge1-\sum_l\bar\eps_l$.
\end{proof}

Applied to the recursion of \cite{anthropic} as it is written, where the grid sides are astronomically large, the
loss is negligible.

\begin{corollary}\label{cor:p3-sb-anthropic}
The recursion of \cite{anthropic} as written (single block, Theorem~A with the grids $g_l=2^{3l+x_0}K_l$ of its (40),
$x_0\ge18$) has exponent at least $(1-\delta)/\log_2K_l$ on $K_l$-node instances for every $l$, with
$\delta=\sum_{l\ge0}4/g_l^2<8.5\cdot10^{-22}$. Hence Corollary~\ref{cor:p3-grid-threshold} holds for it with
$\log_2m^*\ge(1-\delta)\,q/\barr$, and for general 3-SAT still $\log_2m^*\ge6.2\cdot10^3$.
\end{corollary}

\begin{proof}
Each level of that recursion is a $g_l\times g_l$ grid level with one piece per rectangle and $|I|\le(2\lambda)^c$
(so $\sigma=c\log_2(2\lambda)$), followed, under the merge with the smaller total, by the grouping of the $g_l$ merged
nodes into $K_l$ groups of $2^{3l+x_0}$ nodes (its Lemma~8, case~(ii)) and the padding to a $(K_l,L)$-clique; the base
is brute force. Keeping the internal strings of a merged node as settings costs at least the size charge
(Remark~\ref{rem:p3-contains}). So it is a single-block grid algorithm with $\bar\eps=4/g_l^2$ at level $l$, and every
chain meets each level once. By Theorem~A(a) of \cite{anthropic}, $\log_2g_l=3l+x_0+(3x_0+6)2^l-6l-2x_0-6$, which is
$36$ at $l=0$ for $x_0=18$ and grows faster than $2^l$, so $\delta\le4\cdot2^{-72}(1+2^{-100})<8.5\cdot10^{-22}$. The
loss $h\sigma=O(l\log\lambda)$ is $o(\lambda)$. The threshold statement is Proposition~\ref{prop:p3-dichotomy}(a)
with $C=1-\delta$ and $\alpha'=1$.
\end{proof}

For general single-block grid algorithms the factor depends on the grid sides. If every grid level has both sides at
least $g_0$, a chain meets at most $\log_2K/\log_2g_0$ grid levels, so $\delta\le4\log_2K/(g_0^2\log_2g_0)$; at
$\log_2K=1/\barr\approx3.1\cdot10^3$, the range that matters for Corollary~\ref{cor:p3-grid-threshold}, this is about
$10^{-3}$ for $g_0=1024$, but for $g_0\le16$ the product $\Pi_A$ is no longer close to $1$, so the corollary in the
single-block form is proved here only for large sides. The loss is real for unbalanced grids: on an $a\times b$ grid
with $a\ll b$ a column pair holds only $\binom a2$ rectangles, and one of them is lost to rounding.
Grids with a side $2$ are excluded except at $K=4$, where the input floor gives the bound for every algorithm. On a
$2\times b$ level with $b\ge3$ every column pair holds a single rectangle, so a term gives each merged column edge
either all of the rectangle's bits or none, and a bound in terms of the shortest edge gains nothing from mixing; the
method gives only $1/\eK(2b)\lesssim 2/\eK(b)$. We do not know whether the bound $1/\eK(K)\le\log_2K$ can fail for
such levels; that would need an algorithm that is markedly cheaper on an instance with two edge lengths than on the
uniform instance of the mean length, which is the robustness question of Remark~\ref{rem:p3-homog} for grids.

\begin{remark}[Homogeneity and the general balanced case]\label{rem:p3-homog}
(F1) holds for every balanced level, homogeneous or not. (F2) needs homogeneity: four-cell configurations can sit
on a few pairs of groups while $t$ stays small. For instance, with $P_1$ consisting of groups $G_1,\dots,G_{n_1}$ of
size $c$ and $P_2$ of pairs $\{u,v\}$ with $u\in G_a$, $v\in G_{a+1}$ (cyclically), every pair $(G_a,G_{a+1})$ is
entirely in four-cell configurations while $t$ is about $3c/K$; under merge~1 the merged instance is a clique with a
Hamiltonian cycle of light edges. The class has no lower bound for such near-uniform instances except restriction to
a uniform sub-clique, which loses a constant factor in the exponent per level, too much for a polylogarithmic
induction. The general balanced case of Theorem~\ref{thm:p3-A} therefore needs a robustness lemma for near-uniform
weights (Section~\ref{sec:p3-open}).
\end{remark}

For balanced homogeneous two-merge algorithms, Theorem~\ref{thm:p3-A} gives the following consequences, weaker in the
threshold but valid for a larger class than the grid algorithms.

\begin{corollary}\label{cor:p3-threshold}
Let a problem be reduced to uniform clique sums with $s$ edge bits per variable and solved by an algorithm of
$\pIIIRMbh$ whose base methods satisfy the base condition. Then:
\begin{enumerate}[label=(\alph*)]
\item on $K$-node instances the running time is at least $2^{S/(\log_2K)^3}$ with $S=sn$; with $K\le m$ and $m$
polynomial in $n$, this is $2^{\Omega(sn/\log^3n)}$, so no running time $2^{n^{1-c'}}$ with $c'>0$ arises when $s$
is bounded below by a positive constant;
\item in the threshold model of Section~\ref{sec:p3-dichotomy}, $\log_2m^*\ge q\,\barr^{-1/3}$; the threshold grows
exponentially in a power of $1/\barr$, not polynomially as under polynomial decay.
\end{enumerate}
\end{corollary}
\begin{proof}
(a) is Theorem~\ref{thm:p3-A}; (b) is Proposition~\ref{prop:p3-dichotomy}(a) with $C=1$ and $\alpha'=3$.
\end{proof}

For general 3-SAT with $\barr\approx3.2\cdot10^{-4}$, (b) gives only $\log_2K^*\ge\barr^{-1/3}\approx15$, i.e.\
$\log_2m^*\ge29$ at $q=2$, which is not galactic; the quantitative statement is the one for grid algorithms,
Corollary~\ref{cor:p3-grid-threshold}. Among balanced homogeneous two-merge levels the proved exponent is $3$;
lowering it to $1$ for the whole class is the quantitative open problem of this section.

\subsection{\texorpdfstring{Size arguments at fixed $K$}{Size arguments at fixed K}}\label{sec:p3-size}

Theorem~\ref{thm:p3-A} is asymptotic. At a fixed small $K$ the natural lower-bound tool is the size of the tables
that the terms must materialise, which by Proposition~\ref{prop:p3-sharing} bounds every evaluation of them, shared or
not. We define the \emph{size game} of a one-level layout: the edges of $K$ are raw or grouped into crossings (a
crossing is the configuration of two disjoint edges $[x=v][y=u]$ on four nodes, the rectangle of the grid); every
crossing is rewritten by a per-block family of pieces with given $\log_2$-ranks across the splits of its four legs;
a term chooses any contraction tree; the value is the maximum over the adversary's mixtures of the minimum over
trees of the largest table, as an exponent of $N=2^\lambda$, and the exponent per edge bit is the value divided by
$\binom K2$. Only sizes enter: no entropy of naming the pieces and no bilinear cost of building the tables. Two
families are used: \emph{ideal pieces}, costs $(2,0,2)$ and $(0,2,2)$ across (rows, columns, internal split), which
satisfy every constraint we know; and the pieces of \cite[App.~B]{anthropic}, with costs
$(\sqrt{2+u},\sqrt{2-u})$ and the least internal cost $|p-q|$ that the rank inequality below allows.

\begin{table}[ht]
\centering
\begin{tabular}{@{}lccccc@{}}
\toprule
$K$ (layout) & input floor & ideal pieces & App.~B pieces & tree (raw) \\
\midrule
$4$ ($2\times2$) & $3$ & $3$ & $3.414$ & $4$ \\
$6$ ($2\times3$) & $5$ & $[6.667,\,6.725]$ & $6.988$ & $8$ \\
$9$ ($3\times3$) & $8$ & $\ge13.90$ & $\ge14.69$ & $18$ \\
\bottomrule
\end{tabular}
\caption{Size-game values (exponents of $N$; divide by $\binom K2$ for the exponent per edge bit). The ideal value
at $K=6$ is a search value ($6.667$) and a grid bound with a Lipschitz estimate ($6.725$); the other values with
nontrivial pieces at $K=6,9$ are search values, i.e.\ lower estimates of the maximum over mixtures. Computed by
\texttt{calc/r123\_s2\_sizegame.py}.}
\label{tab:p3-size}
\end{table}

At $K=6$ the $2\times3$ grid with ideal pieces has value between $6.667$ and $6.725$ ($\eK=0.444$--$0.448$), with
the pieces of Appendix~B $6.988$ ($\eK=0.466$), against $8$ for the best tree ($\eK=8/15=0.533$) and $5$ for the input
floor. Among $280$ random layouts of one to seven crossings at $K=6$ with ideal pieces the best value found is $7.0$,
none below the grid. So among the layouts tried the grid is the best. As discussed after Table~\ref{pI:tab:eK}, even with ideal pieces one-level layouts move away from the input floor as $K$ grows, while staying a constant factor
below the tree.

By Proposition~\ref{prop:p3-sharing}, shared computation can lower the cost of a construction at most to its
size-game value. For the pieces of Appendix~B at $\omega=2$ the constructions of Section~\ref{sec:part1} cost $3.547$ (per-block accounting; $3.566$ in the grid model used at larger $K$), at
$7.357$--$7.360$ and $15.66$--$15.73$ at $K=4,6,9$, against size floors $3.414$, $6.988$ and $14.69$:
relative gaps of the construction's cost of $3.7\%$, $5.0\%$, and $6.2$--$6.6\%$: a saving of bounded relative size, not one that grows with $K$ as polynomial decay would need.

\paragraph{Why no size argument decides a fixed $K\ge6$ (argued).} A term is a tensor product of pieces on disjoint
legs, and the flattening rank of such a product across any cut is the product of the factors' ranks across the
induced splits. Hence every linear inequality between flattening log-ranks (subadditivity, complement symmetry,
the inequality of Proposition~\ref{prop:p3-scsh}, the inequalities of Cadney, Huber, Linden and Winter
\cite{chlw}, under any grouping of the nodes into parties) holds for every term as soon as it holds for every piece
under the induced grouping of its four legs. The only non-multiplicative constraint is the pigeonhole, which is per
block. Functional pigeonholes obtained by restricting to a sub-clique give floors of at most $K'-1\le5$ at $K=6$,
below $6.667$. So the size question at any fixed $K$ reduces to the question which families of four-leg pieces
exist, and no constraint at $K\ge6$ excludes what $K=4$ allows. At $K=4$ the ideal pieces reach the input floor
exactly, so every size argument is tight there; whether they exist is the question of the per-level factor in the
window $\rho\in[1/2,1/\sqrt2]$ (Proposition~\ref{prop:p3-twomerge} and \cite[App.~B]{anthropic}), equivalently of
the $K=4$ cost in $[3,3.547]$ at $\omega=2$. A one-parameter family $\{(2,x,2-x),(x,2,2-x)\}$, $0<x<1$, is also
admissible and might give less than $6.667$ at $K=6$; it has not been computed.

\paragraph{One level of pair pieces (argued).} With the mixture $1/2$ on every crossing, every split crossing costs at
least $1$; any contraction tree over $K$ leaves has a node with between $K/3$ and $2K/3$ leaves, and at that cut at
most $5K^2/36$ of the roughly $K^2/4$ crossings are unsplit, so the largest table is at least $K^2/9-K/4$ and
$\eK(K)\ge2/9-O(1/K)$ for every one-level layout with ideal pair pieces, against the tree's limit $4/9$. One level is
worth at most a factor $2$, the two-merge floor seen from the other side.

\paragraph{\texorpdfstring{The third cut at $K=4$.}{The third cut at K=4.}} A crossing on four nodes has three splits; the two merges use two of them. The
following is the four-party rank inequality conjectured by Cadney, Huber, Linden and Winter \cite[Hypothesis~1]{chlw}
and proved by Song, Chen, Sun and Hu \cite[Thm.~9]{scsh}, read for a pure four-party state, i.e.\ a four-tensor.

\begin{proposition}[\cite{scsh}]\label{prop:p3-scsh}
For every complex tensor $T$ with four legs $A,B,C,D$,
$\rk(T;BC\mid AD)\le\rk(T;AB\mid CD)\cdot\rk(T;AC\mid BD)$.
\end{proposition}

For a piece with costs $p,q$ on the two crossing splits and $r$ on the internal split this gives $r\ge|p-q|$: a
piece that is full on one merge split and free on the other is full on the third. With this inequality the
three-cut size game at $K=4$ (each of the three matchings of $K_4$ raw or rewritten, a term choosing one of the three
cuts) is no longer conditional: over all two-piece families on a grid of step $1/4$ in the parameters, the value is
at least $3$, the two-cut ideal value, so the third cut never lowers the size floor below the input floor (computed on
the grid, not proved for all real parameters). Without the per-matching structure, flattening inequalities alone
admit families below the input floor, and what excludes them is again a quantum functional, as in
Example~\ref{ex:p3-triangle}.

\subsection{Three escapes and the splits that block them}\label{sec:p3-escapes}

Three ways around the obstructions of this section fail for one reason: each leaves intact a split of the strings of a
configuration across which $f_C|_{\text{block}}$ has full rank, so Proposition~\ref{prop:p3-pigeon} (whose proof
applies verbatim to any split of the strings) gives a piece full across it, and the escape has no evaluation in
which that piece is cheap. In the first two this gives $\rho=1$, in the third $\rho\ge2/3$.

\paragraph{Identity search.} In $AG(2,3)$ with rows, columns and diagonals $x-y=\mathrm{const}$ as merges $k=0,1,2$,
two of the lines $x+y=\mathrm{const}$, $\{a_i\}$ and $\{b_i\}$ ($i\in\mathbb Z_3$), give a configuration $C$ of six
edges that merge $k$ meets in the groups $\{a_i,b_{i+k}\}$, each pair of groups carrying one edge of each triangle:
the candidate for $\rho<1/2$ with three orthogonal merges (Section~\ref{sec:p3-open}). $C$ meets
three groups of each merge, outside (L3); the identities searched for it, and every rewrite of its crossings one at
a time, are of the following form. A \emph{bundle} is the pair of edges of $C$ between two groups of $P_k$; a piece
is \emph{$k$-local} if it is a tensor product of one factor on the strings of each bundle, and it is charged
$\pIIIcost_k(\pi)$, the sum of the factors' $\log_2$-ranks across their pairs of groups, as (L4) would charge each
bundle on its own. Let $X_{ij}$, $Y_{ij}$ be the strings of $a_ia_j$, $b_ib_j$ at $a_i$, $b_i$, call them forward if
$j=i+1$ and backward if $j=i-1$, and let $Q$ split the twelve strings into forward and backward ones.

\begin{proposition}\label{prop:p3-direction}
(a) The rank of $f_C|_{\text{block}}$ across $Q$ is $N^6$. (b) Every $k$-local piece $\pi$ has
$\log_2\rk(\pi;Q)=\pIIIcost_k(\pi)$. (c) Hence, if every piece of a block identity on $C$ is $k$-local for some $k$,
some piece costs at least $6m-\log_2|I_C|$ under every merge for which it is local, and the term taking it in every
block pays the full $6\lambda-o(\lambda)$ bits of $C$ under whichever merge evaluates it: $\rho=1$ on $C$. An
identity of products of factors on pairs of strings inside groups has at least $N^6$ pieces.
\end{proposition}

\begin{proof}
The directions $j-i$ and $i-j$ of $X_{ij}$ and $X_{ji}$ differ since $3$ is odd, and likewise for $Y$; so all six
equalities cross $Q$, which gives (a). The bundle of merge $k$ on $\{a_i,b_{i+k}\}$, $\{a_j,b_{j+k}\}$ has the strings
$X_{ij},Y_{i+k,j+k}$ of direction $j-i$ on one group and $X_{ji},Y_{j+k,i+k}$ of direction $i-j$ on the other: shifts
preserve direction, so $Q$ restricted to a bundle is the split between its groups, and ranks across $Q$ multiply over
the factors, (b). (c) is Proposition~\ref{prop:p3-pigeon} across $Q$ with (b), since a term evaluated under merge $k$
has only $k$-local pieces on $C$; pair factors inside groups have cost $0$.
\end{proof}

Accordingly a search of the last form ($N=2$, $L$ products per merge) finds the identity of \cite[App.~H]{anthropic}
on the rectangle at $L=1$ and nothing on $C$ (\emph{computed} residuals $0.98$, $0.95$ at $L=1,2$; $3L<64$). The
proof applies to two copies of $K_n$ glued by the $n$ cyclic shifts for every odd $n$ (for even $n$ the edges of
direction $n/2$ do not cross $Q$). Left open on $C$ are pieces coupling different bundles, hyperedges under the
merge, for which Section~\ref{sec:p3-open} gives $\rho\ge1/2$.

\paragraph{Structured merged tables (argued).} A level grows $K$ by a constant factor only if one merge has groups of
constant size $g$; on a grid the other merge then has $g$ groups and the input floor $2/g$. Using that merged tables
are contractions (leaving (L5)) does not help within rewrite-and-merge. The term with, in every rectangle, a piece
full across the row split pays it in full in every step bilinear across a cut inducing that split
(Proposition~\ref{prop:p3-sharing}); so a cheap evaluation keeps the two corners in each column of each rectangle
together, and hence whole columns. Building the column tables is the column merge with its floor; keeping them as
networks leaves each column's internal clique ($1/g$ of the bits) untouched and an instance on all $K$ nodes, so the
recursion does not advance.

\paragraph{The shape of App.~B.} A graded space $V=\bigoplus_EV_E$ is given, $D(E)=\log_2\dim V_{\le E}$, and $h$ is
the least $E$ with $D(E)\ge\lambda$; the strings $x,v$ index operator labels $a,b$ and $y,u$ index states, all of
weight at least $h$; for every integer
$s\in[-H,W]$ the piece $A_s$ factors across the row split $(x,y\mid u,v)$ through $V_{\le s+H}$ and across the column
split through $V_{\le W-s}$, with counted bits $c^{\mathrm{row}}_s\ge D(s+H)-o(\lambda)$, $c^{\mathrm{col}}_s\ge
D(W-s)-o(\lambda)$. In App.~B, $W=\mathrm{wt}\,b+\mathrm{wt}\,y=2h$ ($\mathrm{wt}\,y$ the weight of $y$'s state) and $H=\mathrm{wt}\,a+\mathrm{wt}\,b=2h$ (Lemma~9).
Let $\rho_B=\max_s(c^{\mathrm{row}}_s+c^{\mathrm{col}}_s)/(4\lambda)$, the $\rho$ of \cite[Def.~7]{anthropic}.

\begin{proposition}\label{prop:p3-window}
In this shape let $(W+H)/2\ge\mu h$, $\mu\ge1$, let $D$ be concave on $[h,W+H]$ and $\log_2|I|=o(\lambda)$. Then
$\rho_B\ge\frac12\bigl(1+\frac{\mu-1}{2\mu-1}\bigr)-o(1)$; in particular $\rho_B\ge2/3-o(1)$ when $W,H\ge2h$, as in
App.~B.
\end{proposition}

\begin{proof}
Put $M=(W+H)/2$. The pieces sum to $d\,[x=v][y=u]$, of rank $2^{2\lambda}$ across the row split, and $A_s$ has rank at
most $2^{D(2M)}$ there, so $D(2M)\ge2\lambda-\log_2|I|$ (Proposition~\ref{prop:p3-pigeon}). The shift
$s^*=\lfloor(W-H)/2\rfloor$ has both windows at least $M'=\lfloor M\rfloor$. With $t=(M'-h)/(2M-h)$, concavity gives
$D(M')\ge(1-t)\lambda+tD(2M)$, so $4\lambda\rho_B\ge c^{\mathrm{row}}_{s^*}+c^{\mathrm{col}}_{s^*}\ge2\lambda(1+t)-o(\lambda)$;
and $(M-h)/(2M-h)$ increases in $M$, so $t\ge(\mu-1)/(2\mu-1)-1/h$.
\end{proof}

Equality in the concavity step holds exactly for $D$ affine on $[h,2M]$ (for $M=2h$: slope $\lambda/(3h)$, a third
of the mean slope below $h$). For $D(E)=cE^\alpha$ and $W,H\ge2h$ the middle piece gives $\rho_B\ge
\max(2^{\alpha-1},2^{-\alpha})\ge1/\sqrt2$, with equality only at $\alpha=1/2$, so App.~B is optimal among power
profiles. Within the shape $\rho_B\in[2/3,1/\sqrt2]$, and below $1/\sqrt2$ one needs a \emph{kinked} profile (fast
growth up to the weight of the strings, slower beyond) with Lemmas~9 and~10 of \cite{anthropic}; we know no such
algebra. A window $H$ as small as the weight of one factor ($\mu=3/2$) would give $5/8$.

\paragraph{Fermionic Fock space (argued).} The exterior algebra on $n$ modes, $D(E)=\log_2\sum_{j\le E}\binom nj$,
has a nearly linear, then flat profile, but even in the shape above with Lemma~9's windows its counted factor is at
best $0.763$ (\emph{computed}, $n=2000$; $0.660$ needs $H=h$). And it has no such pieces: App.~B's $Q_s$ factor in
both orders (Lemma~9(a)), a consequence of vertex-algebra locality, while Wick ordering
$\psi_T\psi^*_S=\sum_{U\subseteq S\cap T}\pm\psi^*_{S\setminus U}\psi_{T\setminus U}$ ($x,v$ encoded by $h$-sets
$T,S$, $y,u$ by states of weight $h$; the term $U=S=T$ is $[x=v][y=u]$) gives pieces that factor in one order only.
The reordered terms with $|U|=j$ factor across the row split through the shared $U$ and a state of weight $j$, at
$\log_2$-rank up to $2\log_2\binom nj=2\lambda-O(\log(n/h))$ for $j=h-1$, and are dense across the column split. So
one piece is nearly full across both splits, unlike the pieces of Example~\ref{ex:p3-triangle}, and $\rho\to1$. The
infinite free-fermion vertex algebra is local but has, by the boson--fermion correspondence, the partition profile;
truncating its modes at $n$ kinks the profile only at weights of order $n^2$, where locality fails.

\subsection{What remains open}\label{sec:p3-open}

The first three items below would sharpen or extend the bounds of this section without changing their message; the
last is the only place where the message itself could change.

\paragraph{The general four-cell case.} The rectangle-content bound of Section~\ref{sec:p3-cells}, that only
$\min(d,a)$ bits of a four-cell configuration can be traded, is open. Its smallest open case is $K_4$ minus a
diagonal; the $4$-cycle is open when over-full pieces are allowed. It would sharpen the exponent of
Theorem~\ref{thm:p3-A}, not its structure.

\paragraph{Non-homogeneous balanced levels.} Theorem~\ref{thm:p3-A} for all balanced two-merge levels needs a
robustness lemma: a clique sum whose weights are uniform except on a sparse set of pairs (such as the cycle of
Remark~\ref{rem:p3-homog}) costs nearly as much as the uniform one, in a form that loses only $1-o(1)$ in the
exponent per level. We believe it, but it requires the level analysis for weighted instances.

\paragraph{Three or more merges.} With $p$ merges the count applies to every pair of them. In the $d$-dimensional
grid with line merges, an edge whose endpoints differ in $r$ coordinates lies in one $r$-subcube, and the facet cuts
of that subcube are unions of groups of all but one of its merges; the pigeonhole then forces every identity to pay at
least $1-1/r$ of the subcube's bits, so line merges are worse than the two-dimensional grid. Three pairwise
orthogonal merges can have a floor as low as $1/3$, but orthogonality again forces squaring. Below $1/2$ one would
need identities beyond those we know: we argue that pieces that are products of factors on disjoint pairs of strings
cannot go below $\rho=1/2$ on the smallest configuration that escapes the known obstructions, the translation matching
of the affine plane $AG(2,4)$ with three of its other parallel classes as merges; and for every configuration whose
nodes split into two halves, each meeting every block of every merge at most once, the uniform quantum functional
proves $\rho\ge1/2$ for all pieces. What is not covered is a piece that is not a product over pairs of strings on
a configuration without such halves: on $AG(2,4)$ this means a hyperedge over four blocks of the third merge, whose
cost the model does not yet define. With $p$ merges the analogue of Lemma~\ref{lem:p3-Rprime} concerns $2^p$ cells.
Even a success here would be a lever on $\rho$, i.e.\ on $\alpha'$, not a change of the type of decay.

\paragraph{The open door: hyperedge configurations.} Inside $\RM_2$ every configuration meets at most two groups of
each merge, and Remark~\ref{rem:p3-hyper} shows that this is exactly what a non-product grouping of terms across
different merged edges violates. We state the role of such configurations as a necessary condition, not a sufficient
one. Within the class, and modulo the counting lemma's premise in the general four-cell case, polynomial decay of
$\eK(K)$ needs a size-game value that decays polynomially, and by Sections~\ref{sec:p3-cells}
and~\ref{sec:p3-size} only pieces of growing arity whose configurations meet three or more groups of a merge
(equivalently, non-product groupings of terms) could give one. Such pieces would also have to be realisable, which no
current tool decides, since rank inequalities for six or more parties are incomplete; their tables would have to be
buildable in time close to their size, which is a bilinear-complexity question; a block would have to need few of
them, since their naming entropy is paid per block; and the constants $C$ and $\beta$ would have to be moderate,
since by Proposition~\ref{prop:p3-dichotomy} the threshold is $\log_2m^*\approx(q/\beta)\log_2(C/\barr)$ plus the
lower-order constants of the reduction. No construction gives evidence of polynomial decay, and for the full
model of Signed Network Sums no barrier is known: the tools used here (ranks, flattenings, quantum functionals) bound
pieces, not the way they are combined beyond pairs of groups.

\section{Open problems}\label{sec:open}

The problems are ordered from the small end to the large end. The first four would sharpen the picture of Sections~\ref{sec:part1}
and~\ref{sec:part3}; the fifth is the only one whose solution could give a usable threshold inside Signed Network Sums; the last
two would lower the thresholds of Section~\ref{sec:part2} without changing their scale.

\begin{enumerate}[leftmargin=*]
\item \emph{The window at four nodes.} At $K=4$ every lower-bound tool we know stops at the input floor $N^3$, and the
best construction at $\omega=2$ uses Appendix~B pieces of \cite{anthropic}. Is there a four-leg piece family with
per-merge factor $\rho$ below $1/\sqrt2$ whose tables are built near their size? Equivalently, close the window
$\rho\in[1/2,1/\sqrt2]$; at fixed $K$ every table-size argument reduces to this question (Section~\ref{sec:part3}). Inside the shape of the pieces of Appendix~B of \cite{anthropic}, every concave profile gives $\rho\ge2/3$
(Proposition~\ref{prop:p3-window}), so $\rho\in[1/2,2/3)$ needs pieces of a different shape, and $\rho\in[2/3,1/\sqrt2)$ a
kinked profile.
\item \emph{Every-$K$ separation.} Do one-level constructions beat every contraction tree at every $K\ge4$, on the
current matrix-multiplication curve and not only at $\omega=2$? At $K=6$ the construction beats the best tree on the current curve
($8.017$ against $9$) but not the width bound $8$.
\item \emph{The four-cell lemma.} Does a configuration on all four cells of two merges trade at most its rectangle
content? The smallest open case is $K_4$ minus one diagonal. It would only sharpen constants in the class barrier.
\item \emph{Robustness and more merges.} Extend the class barrier to balanced but non-homogeneous partitions (a
robustness lemma for near-uniform weights; for grid levels with both sides at least $3$ the near-uniform case is
settled by Theorem~\ref{thm:p3-grid-sb}, and $2\times b$ levels with $b\ge3$ remain open) and to three or more merges (the four-cell lemma for $2^p$ cells; the
open case is non-pair pieces on pairwise orthogonal designs).
\item \emph{The open door.} Is there an exact identity whose configurations meet three or more groups of a merge and
whose pieces are cheap? A usable threshold for general 3-SAT inside Signed Network Sums needs one; it would also have
to be realisable, buildable near its size, and use few pieces per block.
\item \emph{A linear-size kill.} Can the assignments that violate many clauses be eliminated with $O(n)$ edge bits?
This would remove the factor $\log\log n$ from the general bound and lower the general threshold by orders of
magnitude, though not to a usable size.
\item \emph{Clustering for the variable network.} Prove the clustering lemma for the variable network of Theorem~D
of~\cite{anthropic}; it is the condition of row G4, and one of the two conditions of row G7, of the general table of Section~\ref{sec:part2}.
\end{enumerate}

\appendix
\clearpage
\section*{Appendices}\addcontentsline{toc}{section}{Appendices}
\IfFileExists{appendix-part1.tex}{\section{Values of the one-level constructions at $K=9$ and $K=16$}\label{pI:app}

Table~\ref{pI:tab:k916} lists, for each construction, the certified interval over all term profiles (lower end: the
value at the worst profile found in closed form; upper end: the certified bound) on three curves. G is the eight-square
expansion grouped by naming bits, E the grouping counted by ranks, EM the rank count for at most half column-choice
blocks and G above, B the Appendix~B pieces with $m\to\infty$. The finite-$m$ rows use exact partition numbers.
Scripts: calc/r122d\_certify.py, r122d\_tables9.py, r122d\_subst.py, r122d\_bfinite.py.

\begin{table}[h]
\centering
\small
\begin{tabular}{rlccc}
\toprule
$K$ & construction & $\omega=2$ & current curve & VXXZ\\
\midrule
9 & best tree found & $18$ & $20.25$ & $21.344$\\
9 & G & $[17.333,17.721]$ & $[19.125,19.50]$ & $[20.158,20.554]$\\
9 & EM & $[17.269,17.657]$ & $[18.908,19.428]$ & $[19.929,20.478]$\\
9 & E & $[17.204,17.593]$ & $[18.908,19.356]$ & $[19.929,20.402]$\\
9 & B, $m\to\infty$ & $[15.656,15.725]$ & $[16.295,17.614]$ & $[17.176,18.566]$\\
9 & B, $m=256$ / $512$ & $16.456$ / $16.187$ & -- / $18.019$ & --\\
\midrule
16 & best tree found & $60$ & $64$ & $65.376$\\
16 & G & $60$ & $[60,62.25]$ & $[61.28,64.42]$\\
16 & EM & $[59.229,59.615]$ & $[59.229,61.865]$ & $[60.50,64.04]$\\
16 & E & $59.229$ & $[59.229,61.480]$ & $[60.50,63.65]$\\
16 & B, $m\to\infty$ & $[49.94,57.95]$ & $[49.94,59.95]$ & $[51.01,60.92]$\\
16 & B, $m=512$ & $59.718$ & $61.72$ & --\\
\bottomrule
\end{tabular}
\caption{Exponents of $N$ at $K=9$ (width bound $18$) and $K=16$ (width bound $60$).}\label{pI:tab:k916}
\end{table}
}{}
\IfFileExists{appendix-part2.tex}{\section{Computed schedules and the constant \texorpdfstring{$c$}{c}}\label{app:p2}

\subsection{The optimal admissible schedules}\label{app:p2-schedules}

Table~\ref{tab:p2-paths} lists, level by level, the schedules behind columns (X) and (T) of
Table~\ref{tab:p2-density} at densities $1$ and $4.26$: $\log_2K_\ell$, $e_\ell$, and the chosen $\log_2\zeta_\ell$ and
$\log_2s_\ell$, where $s_\ell$ is the least value with (S1) and (S2). The exponents decrease along every schedule,
so (S4) holds, and (S3) holds since $x_0\ge4$ and $c=30$. Values are rounded; the model is that of
\texttt{calc/r123\_levers.py} (main term, $\rho=1/\sqrt2$, $\tau=1$, overheads $3\zeta$ and $3/K$).

\begin{table}[ht]
\centering\small
\begin{tabular}{@{}l r r r r@{\qquad}l r r r r@{}}
\toprule
\multicolumn{5}{@{}l}{density 1, brute-force base, $x_0=7$} & \multicolumn{5}{l}{density 1, tree base, $x_0=8$}\\
$\ell$ & $\log_2K_\ell$ & $e_\ell$ & $\log_2\zeta_\ell$ & $\log_2s_\ell$ & $\ell$ & $\log_2K_\ell$ & $e_\ell$ & $\log_2\zeta_\ell$ & $\log_2s_\ell$\\
\midrule
0 & 7    & 0.500 & $-4.5$  & 12.2 & 0 & 8    & 0.255 & $-4.5$ & 12.3\\
1 & 38.5 & 0.443 & $-6.0$  & 17.0 & 1 & 40.7 & 0.226 & $-6.0$ & 17.1\\
2 & 111  & 0.334 & $-7.0$  & 20.4 & 2 & 116  & 0.171 & $-7.5$ & 21.5\\
3 & 263  & 0.244 & $-8.5$  & 24.6 & 3 & 274  & 0.123 & $-9.0$ & 25.6\\
4 & 575  & 0.175 & $-9.5$  & 27.7 & 4 & 600  & 0.0880 & & \\
5 & 1205 & 0.124 & $-11.0$ & 31.7 & & & & & \\
6 & 2474 & 0.0880 & & & & & & & \\
\midrule
\multicolumn{5}{@{}l}{density 4.26, brute-force base, $x_0=7$} & \multicolumn{5}{l}{density 4.26, tree base, $x_0=8$}\\
\midrule
0 & 7     & 0.500 & $-5.5$  & 14.3 & 0 & 8    & 0.255 & $-5.5$  & 14.4\\
1 & 42.7  & 0.410 & $-6.5$  & 18.2 & 1 & 44.9 & 0.209 & $-7.0$  & 19.3\\
2 & 122   & 0.304 & $-8.0$  & 22.5 & 2 & 128  & 0.153 & $-8.0$  & 22.6\\
3 & 289   & 0.218 & $-9.0$  & 25.7 & 3 & 302  & 0.110 & $-9.5$  & 26.8\\
4 & 629   & 0.156 & $-9.5$  & 27.8 & 4 & 657  & 0.0781 & $-10.0$ & 28.9\\
5 & 1313  & 0.111 & $-10.5$ & 30.8 & 5 & 1372 & 0.0555 & $-11.0$ & 31.9\\
6 & 2687  & 0.0785 & $-11.5$ & 33.9 & 6 & 2808 & 0.0393 & $-13.0$ & 36.9\\
7 & 5442  & 0.0556 & $-12.5$ & 36.9 & 7 & 5689 & 0.0278 & & \\
8 & 10957 & 0.0393 & $-14.0$ & 40.9 & & & & & \\
9 & 21996 & 0.0278 & & & & & & & \\
\bottomrule
\end{tabular}
\caption{Admissible schedules behind Table~\ref{tab:p2-density}; the last row of each block is the top clique,
$\log_2\pIImstar=9\log_2K_L$ (main term, $q=9$).}
\label{tab:p2-paths}
\end{table}

\subsection{The requirements on \texorpdfstring{$c$}{c} on the exact range}

Table~\ref{tab:p2-c} gives, for selected $\lambda$, the least $c$ that each condition of Definition~7
of~\cite{anthropic} requires for the pieces of its Theorem~B, with partition numbers computed exactly; $\Delta$ is the
number of partitions of $0,\dots,4h$ and the last column is the maximum before the $+2$ of the truncation.

\begin{table}[ht]
\centering\small
\begin{tabular}{@{}r r r r r r r r@{}}
\toprule
$\lambda$ & $h$ & $\log_2\Delta$ & $\max(c^{\rm row}+c^{\rm col})$ & (b) & (c1) & (c2) & (c3)\\
\midrule
1   & 2    & 6.1   & 8   & 5.17 & 3.17 & 6.97 & \textbf{27.20}\\
2   & 4    & 9.8   & 14  & 4.17 & 2.04 & 4.38 & 20.76\\
4   & 8    & 15.4  & 21  & 3.23 & 1.68 & 3.56 & 15.31\\
8   & 17   & 24.4  & 34  & 2.84 & 1.53 & 3.20 & 11.44\\
12  & 29   & 33.4  & 46  & 2.63 & 1.50 & 3.12 & 10.05\\
25  & 86   & 61.3  & 85  & 2.53 & 1.49 & 3.09 & 8.37\\
50  & 271  & 113.7 & 158 & 2.50 & 1.52 & 3.12 & 7.43\\
100 & 927  & 216.3 & 303 & 2.64 & 1.55 & 3.18 & 6.84\\
200 & 3361 & 419.1 & 589 & 2.70 & 1.59 & 3.24 & 6.49\\
271 & 5988 & 562.3 & 792 & 2.81 & 1.60 & 3.27 & 6.39\\
\bottomrule
\end{tabular}
\caption{Requirements on $c$ per condition of Definition~7 (\texttt{calc/r123\_trace\_c.py}). The maximum over
$1\le\lambda\le271$ is $27.198$ at $\lambda=1$; beyond, at most $9.59$ (Lemma~\ref{lem:p2-c}).}
\label{tab:p2-c}
\end{table}
}{}
\IfFileExists{appendix-part3.tex}{\section{Size-game computations}\label{app:p3-size}

The values of Table~\ref{tab:p3-size} were computed by \texttt{calc/r123\_s2\_sizegame.py}. A layout lists the
crossings; a term's profile assigns to each crossing type the adversary's mixture parameter; for each profile the
script finds the contraction tree minimising the largest table by dynamic programming over subsets of nodes, and the
value is the maximum over profiles, found by search. Search values are therefore lower estimates of the maximum over
mixtures, and the margins below the tree derived from them are upper bounds. Checks: for raw layouts the value equals
the best tree ($4$, $8$, $18$ at $K=4,6,9$), and ideal pieces at $K=4$ give the input floor $3$. For the $2\times3$
grid with ideal pieces a grid evaluation over profiles with a Lipschitz bound gives the upper value $6.725$ (grid
maximum $6.65$). The worst profiles found are listed in the output file.

\begin{table}[ht]
\centering
\begin{tabular}{@{}cccc@{}}
\toprule
crossings & layouts & least value & median value \\
\midrule
1 & 40 & 8.000 & 8.000 \\
2 & 40 & 7.000 & 7.000 \\
3 & 40 & 7.000 & 7.500 \\
4 & 40 & 7.000 & 7.500 \\
5 & 40 & 7.286 & 7.667 \\
6 & 40 & 7.500 & 8.000 \\
7 & 40 & 7.885 & 8.137 \\
\bottomrule
\end{tabular}
\caption{Random layouts of crossings at $K=6$ with ideal pieces (search values; the least value $6.9999$ for four
crossings is rounded). None is below the $2\times3$ grid ($6.667$--$6.725$). Computed by
\texttt{calc/r123\_s2\_layouts.py}.}
\label{tab:p3-layouts}
\end{table}
}{}

\end{document}